%% file: paper.tex
\documentclass[conference]{IEEEtran} 
\input{headers}

\newtoggle{long}
\togglefalse{long}

\IEEEoverridecommandlockouts
\title{{Byzantine Distributed Function Computation}}
\author{
    \IEEEauthorblockN{Hari Krishnan P. Anilkumar}
    \IEEEauthorblockA{TIFR, Mumbai, India\\
    hari.a@tifr.res.in
    }
    \and
    \IEEEauthorblockN{Neha Sangwan}
    \IEEEauthorblockA{UCSD, USA\\
    {nehasangwan010@gmail.com}
    }
    \and
    \IEEEauthorblockN{Varun Narayanan}
    \IEEEauthorblockA{UCLA, USA\\
    {varunnkv@gmail.com}
    }
    \and
    \IEEEauthorblockN{Vinod M. Prabhakaran}
    \IEEEauthorblockA{TIFR, Mumbai, India\\
    {vinodmp@tifr.res.in}
    }}

\begin{document}
\maketitle
\input{abstract}
\input{introduction}

\input{example}

\input{appendix-2-theorem}

\section*{Acknowledgment}
We thank Aaron Wagner for calling our attention to the works by Kosut and Tong on byzantine distributed source coding~\cite{KosutTTransIT08,KosutTISIT08,KosutTAllerton08,KosutTISIT09}.

\appendices

\input{appendix_5page}
\input{appendix}
\bibliography{IEEEabrv,bib}
\end{document}

%% file: headers.tex
\usepackage{url}
\usepackage{hyperref}
\hypersetup{
  colorlinks   = true, 
  linkcolor    = black!50!brown, 
  citecolor   = black!50!brown, 
}
\usepackage{balance}
\usepackage{verbatim}
\usepackage{bm}
\usepackage{color,graphicx,xcolor}
\usepackage{mdwtab}
\usepackage{mathtools,tikz}
\usepackage{pgfplots}
\pgfplotsset{soldot/.style={color=myblue,only marks,mark=*}} \pgfplotsset{holdot/.style={color=myblue,fill=white,only marks,mark=*}}
\usepackage{hhline}
\usepackage{multirow}
\usepackage{pdfpages}
\usepackage{enumitem}  
\usepackage{caption}
\usepackage{subcaption}
\DeclareCaptionFormat{myformat}{\fontsize{8}{9}\selectfont#1#2#3}
\usepackage[noadjust]{cite}
\usepackage{amsmath,amssymb,amsfonts,amsthm,xspace,bm}
\usepackage[bb=dsserif]{mathalpha}
\usepackage{algorithmic}
\usepackage{graphicx}
\usepackage{textcomp}
\usepackage{todonotes}
\usepackage{etoolbox}
\usetikzlibrary{arrows,shapes.geometric,calc}
\usepackage[rightcaption]{sidecap}
\sidecaptionvpos{figure}{m}
\allowdisplaybreaks
\definecolor{bluishgreen}{RGB}{0,158,115}
\definecolor{vermillion}{RGB}{213,94,0}
\definecolor{myblue}{RGB}{0,114,178}
\definecolor{myorange}{RGB}{230,159,0}

\theoremstyle{definition}

\newtheorem{remark}{Remark}
\newtheorem{theorem}{Theorem}
\newtheorem{definition}{Definition}

\newtheorem{protocol}{Protocol}

\newtheorem{claim}[theorem]{Claim}

\newtheorem*{definition*}{Definition}
\newtheorem*{lemma*}{Lemma}

\newtheorem{lemma}[theorem]{Lemma}

\newtheoremstyle{thmnum}{\topsep}{\topsep}{\itshape}{0pt}{\bfseries}{.}{ }{\thmname{#1}\thmnote{ \bfseries #3}}
\theoremstyle{thmnum}

\usepackage{pgffor}
\foreach \x in {a,...,z}{%
\expandafter\xdef\csname vec\x \endcsname{\noexpand\ensuremath{\noexpand\bm{\x}}}
}

\foreach \x in {A,...,Z}{%
\expandafter\xdef\csname vec\x \endcsname{\noexpand\ensuremath{\noexpand\bm{\x}}}
}

\foreach \x in {A,...,Z}{%
\expandafter\xdef\csname c\x \endcsname{\noexpand\ensuremath{\noexpand\mathcal{\x}}}
}

\foreach \x in {A,...,Z}{%
\expandafter\xdef\csname bb\x \endcsname{\noexpand\ensuremath{\noexpand\mathbb{\x}}}
}

\foreach \x in {A,...,Z}{%
\expandafter\xdef\csname s\x \endcsname{\noexpand\ensuremath{\noexpand\sf{\x}}}
}

\foreach \x in {J,M,U,X,Y,Z}{%
\expandafter\xdef\csname h\x \endcsname{\noexpand\ensuremath{\noexpand\widehat{\x}}}
}

\foreach \x in {m,u,x,y,z}{%
\expandafter\xdef\csname h\x \endcsname{\noexpand\ensuremath{\noexpand\hat{\x}}}
}

\foreach \x in {W,X,m,x,y,z}{%
\expandafter\xdef\csname t\x \endcsname{\noexpand\ensuremath{\noexpand\widetilde{\x}}}
}
\xdef\tzero{\noexpand\ensuremath{\noexpand\tilde{0}}}
\xdef\tone{\noexpand\ensuremath{\noexpand\tilde{1}}}

\foreach \x in {M,P,U,W,X,Y,Z,u,x,y,z}{%
\expandafter\xdef\csname b\x \endcsname{\noexpand\ensuremath{\noexpand\overline{\x}}}
}

\foreach \x in {M,P,U,X,Y,Z,u,x,y,z}{%
\expandafter\xdef\csname u\x \endcsname{\noexpand\ensuremath{\noexpand\underline{\x}}}
}

\newcommand{\mc}{\ensuremath{\leftrightarrow}}
\newcommand{\eras}{\ensuremath{\mathsf{e}}}

\newcommand{\ina}[1]{\left<#1\right>}

\newcommand{\inp}[1]{\left(#1\right)}

\newcommand{\red}[1]{{\textcolor{red}{#1}}}
\newcommand{\blue}[1]{{\textcolor{blue}{#1}}}
\newcommand{\coll}{\ensuremath{\mathbb{A}}}
\newcommand{\pa}{\ensuremath{\mathsf{P}}}
\newcommand{\bcA}{\ensuremath{\cA^c}}

\newcommand{\pr}{\ensuremath{\mathbb{P}}}

\newcommand{\rad}{\delta}

\newcommand{\ham}{\ensuremath{d_{\text{H}}}}
\newcommand{\dtv}{\ensuremath{d_{\tt{TV}}}}
\newcommand{\type}[1]{\ensuremath{\mathbf{p}_{#1}}}
\newcommand{\supp}{\ensuremath{\mathsf{Supp}}}
\newcommand{\dmc}{\ensuremath{\mathsf{dmc}}}

\newcommand{\Wsym}{\ensuremath{W^{\mathtt{sym}}}}
\newcommand{\bb}[1]{\ensuremath{\mathbb{#1}}}
\newcommand{\ttf}{\ensuremath{\Phi}}
\newcommand{\mss}{\ensuremath{\searrow}}
\newcommand{\Xdag}{\ensuremath{X^{\dagger}}}
\newcommand{\xdag}{\ensuremath{x^{\dagger}}}

\newcommand{\txa}{\ensuremath{\tX^1}}
\newcommand{\txb}{\ensuremath{\tX^2}}
\newcommand{\txc}{\ensuremath{\tX^3}}
\newcommand{\Pabc}{\ensuremath{P_{X_1 X_2 X_3 Y}}}

\def\BibTeX{{\rm B\kern-.05em{\sc i\kern-.025em b}\kern-.08em
    T\kern-.1667em\lower.7ex\hbox{E}\kern-.125emX}}

\newcommand{\protdecoder}{\ensuremath{  \mathsf{decode} _{(2,1)} }}
\newcommand{\protdecoderr}{\ensuremath{  \mathsf{decode} _{(k,1)} }}

%% file: abstract.tex
\begin{abstract} 
    We study the distributed function computation problem with $k$ users of which at most $s$ may be controlled by an adversary and characterize the set of functions of the sources the decoder can reconstruct robustly in the following sense -- if the users behave honestly, the function is recovered with high probability (w.h.p.); if they behave adversarially, w.h.p, either one of the adversarial users will be identified or the function is recovered with vanishingly small distortion.
\end{abstract}

%% file: introduction.tex
\section{Introduction}\label{sec:intro}

In the distributed source coding problem, spatially distributed users observe correlated sources and send encodings of their observations to a central decoder which, using these messages and any source observation (side-information) of its own, attempts to reconstruct all the sources or a function of them. Following the seminal work of Slepian and Wolf~\cite{SlepianW73}, a wide variety of settings has been studied (e.g., see~\cite{dragotti2009distributed} and references therein). The focus of these works is mainly on understanding the amount of communication required, i.e., how short can the messages from the users be so that the decoder can reconstruct the sources at the desired fidelity. In this paper we turn our attention from limited communication to limited {\em trust}. Specifically, the decoder can no longer trust all the users to act faithfully. Out of $k$ users, up to $s\leq k$ may be controlled by an adversary (we call these users malicious); the upper bound $s$ is assumed to be known beforehand, but the identity of which users are controlled by the adversary is unknown to the decoder and the honest users. We ask what, if anything, can the central decoder hope to learn about the sources in such a setting. 
{A similar question with somewhat different motivations from ours was studied by Kosut and Tang~\cite{KosutTTransIT08,KosutTISIT08,KosutTAllerton08,KosutTISIT09}. We will discuss the similarities and differences with the problem we study towards the end of this section.}

In this paper we are interested in the following form of {\em robust recovery} (see Figure~\ref{fig:setup}): the decoder should either be able to compute a desired function $f$ of the sources or it must be able to identify a malicious user (if any). The main result will be a characterization of all functions $f$ for which this is possible for a given joint distribution of the sources. Note that if a malicious user is identified, that user may be removed to yield a new instance of the problem with $k-1$ users of whom at most $s-1$ are malicious and the process continued. 
{We do not impose any communication restrictions on the users. In fact, the users send their observations uncoded to the decoder. Furthermore, we assume that the adversary has no additional side-information other than the observations of users it controls.}

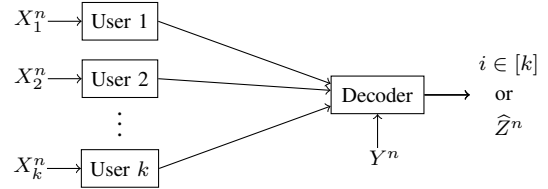
\begin{figure}[tb]
\centering

\begin{tikzpicture}[node distance=0.25cm, every node/.style={draw, minimum size=0.5cm, align=center}, scale = 0.6]

\node (user1) [rectangle] {\footnotesize User~1};
\node (user2) [rectangle, below=of user1] {\footnotesize User~2};
\node (dots) [draw=none, below=of user2, yshift = 0.4cm] {$\vdots$};
\node (userk) [rectangle, below=of dots, yshift = 0.2cm] {\footnotesize User~$k$};

\draw[->] ($(user1.west)+(-0.75,0)$) -- (user1.west) node[midway, left, draw=none, xshift=-0.1cm] {\footnotesize \(X_1^n\)};
\draw[->] ($(user2.west)+(-0.75,0)$) -- (user2.west) node[midway, left, draw=none, xshift=-0.1cm] {\footnotesize \(X_2^n\)};
\draw[->] ($(userk.west)+(-0.75,0)$) -- (userk.west) node[midway, left, draw=none, xshift=-0.1cm] {\footnotesize \(X_k^n\)};

\node (decoder) [rectangle, right=2.8cm of $(user1)!0.5!(userk)$, text width=1cm] {\footnotesize Decoder};

\draw[->] ($(decoder.south)+(0,-0.75)$) -- (decoder.south) node[midway, below, draw=none, xshift = 0.1cm, yshift = -0.1cm] {\footnotesize \(Y^n\)};

\draw[->] (user1.east) -- ($(decoder.west)+(0,0.25)$);
\draw[->] (user2.east) -- ($(decoder.west)+(0,0.1)$);
\draw[->] (userk.east) -- ($(decoder.west)+(0,-0.25)$);

\draw[->] (decoder.east) -- ($(decoder.east)+(1,0)$) node[midway, right, draw=none, xshift=0.3cm, yshift=0.4cm] {\footnotesize \(i\in[k]\)};
\draw[->] (decoder.east) -- ($(decoder.east)+(1,0)$) node[midway, right, draw=none, xshift=0.5cm] {\footnotesize or};
\draw[->] (decoder.east) -- ($(decoder.east)+(1,0)$) node[midway, right, draw=none,xshift=0.5cm, yshift=-0.4cm] {\footnotesize \(\hZ^n\)};

\end{tikzpicture}
\caption{At most $s$ out of the $k$ users are malicious (i.e., controlled by an adversary who has access to the observations of the malicious users, but no additional side-information). The decoder, with high probability, either identifies a malicious user or outputs a substantially correct estimate of $Z^n$, where $Z=f(X_1,\ldots,X_k,Y)$. The main result is a characterization of functions $f$ for which this is possible for a given $P_{X_1,\ldots,X_kY}$. If the decoder identifies a malicious user, that user can be removed to get a new instance of the problem with $k-1$ users of which at most $s-1$ are malicious and the process repeated.}\label{fig:setup}
\end{figure}

While not the focus of this paper, we are also motivated by potential applications to information theoretically secure signatures and byzantine broadcast~\cite{fitzi2004pseudo,narayanan2023complete}. Indeed, our starting point is the following observation which forms the basis of the algorithm in~\cite{fitzi2004pseudo}:
Consider a single source node with observation vector $X^n=(X_1,\ldots,X_n)$ and a decoder with side-information vector $Y^n=(Y_1,\ldots,Y_n)$ such that 
$(X_i,Y_i)$, $i=1,\ldots,n$ are independent and
identically distributed according to a joint distribution $P_{XY}$
(this is the $k=s=1$ case). Let $X\searrow Y:=\psi_{X\searrow Y}(X)$ be a minimal sufficient statistic for estimating $Y$ given $X$, i.e., $\psi_{X\searrow Y}$ function is such that $\psi_{X\searrow Y}(x)=\psi_{X\searrow Y}(x')$ if and only if $P_{Y|X=x}\equiv P_{Y|X=x'}$. Then the decoder can recover $\psi_{X\searrow Y}(X^n):=(\psi_{X\searrow Y}(X_1),\ldots,\psi_{X\searrow Y}(X_n))$ robustly in the following sense -- if the sender is honest, the decoder outputs $\psi_{X\searrow Y}(X^n)$ with high probability (w.h.p.); if the sender is malicious, w.h.p., either the decoder outputs a vector which is still substantially correct (in the sense of vanishing average Hamming distortion w.r.t. $\psi_{X\searrow Y}(X^n)$) or detects that the sender is malicious. In other words, a malicious sender is unable to induce the decoder to produce an erroneous output (without being detected). This is accomplished by a simple joint typicality test by the decoder. It is also easy to see that $\psi_{X\searrow Y}(X^n)$ (along with $Y^n$) is the most that the decoder can hope to learn robustly since $X\mc \psi_{X\searrow Y}(X) \mc Y$ is a Markov chain and any alterations to $X$ which preserve $\psi_{X\searrow Y}(X)$ (and the marginal of $X$) cannot be detected by the decoder relying only on its side-information $Y$. 

In~\cite{narayanan2023complete}, the above was extended to the case of $k=2$ senders with at most $s=1$ corrupt user using a simple idea -- for a joint distribution $P_{X_1X_2Y}$, suppose the users observe $X_1$ and $X_2$ (dropping the vector notation for convenience), respectively, and the decoder has side-information $Y$. Using the single-user scheme, the decoder can first robustly recover $X_1\searrow Y$ (or determine that user 1 is malicious); then using what it has learned, robustly recover $X_2\searrow Y_1$ (or determine that user~2 is malicious), where $Y_1:=(Y,X_1\searrow Y)$; then recover $X_1\searrow Y_2$ (or determine that user~1 is malicious), where $Y_2:=(Y_1,X_2\searrow Y_1)$ and so on. This process can be shown to saturate in a finite number of steps (which depends on their joint distribution and is at most the product of cardinalities of the alphabets of $X_1$ and $X_2$). Thus, either the  decoder recovers this (saturated) function or identifies the adversary. Notice that if the adversary is identified, the decoder can proceed with the knowledge that the other user is honest. A consequence of the results in the present paper is that this function (along with $Y$) is the most a decoder can learn robustly\footnote{\cite{narayanan2023complete} did not formulate the precise  question we formulate here (since the goal there was a signature/byzantine broadcast scheme which imposes further restrictions) and the converses proved there do not directly imply an optimality for the formulation here.}.

While we can extend the idea of repeatedly using the single-user scheme to $k>2$ users for at most $s=1$ corrupt users (see Appendix~\ref{app:k=1}), the scheme does not generalize\footnote{One of the several issues which prevent an easy generalization of the single-user scheme to the $s>1$ case is that collusions provide additional side-information to a malicious party to craft its report (and this side-information is unavailable in a trustworthy form to the decoder); in such cases the minimum sufficient statistic or multiparty analogues of it are no longer robustly recoverable.} for $s>1$. One of the main technical contributions of this paper, in addition to the new formulation and the converse result, is an optimal achievability scheme which works for all $s\leq k$.


{
\paragraph*{Related works} Apart from~\cite{fitzi2004pseudo,narayanan2023complete}, the closest works to ours are by Kosut and Tong on byzantine distributed source coding~\cite{KosutTTransIT08} and the byzantine CEO problem~\cite{KosutTISIT08,KosutTAllerton08,KosutTISIT09}. Distributed source coding with byzantine users studied in~\cite{KosutTTransIT08} is very similar to our setup except that the decoder outputs reconstructions for the observations of all the users with the objective that the reconstructions must be correct for the observations of the honest users. An obvious way to do this is to require the users to send their observations uncoded (as we do in this paper) and output; but this may not be communication efficient. The paper presents the optimal communication rate-tradeoff region of deterministic and randomized fixed-length codes and the optimal sum-rate for variable length codes. Unlike in our formulation, the adversary is allowed to have additional side-information. The main difference with our problem is that the decoder's aim is not to correctly compute a function of the observations (or declare an outage by identifying a malicious user), but to produce reconstructions for the observations of all users of which those for honest users' are correct; the decoder may not be able to identify which of its reconstructions are correct. Our focus in this paper is on the feasibility question of which functions are robustly recoverable and we do not study communication efficiency, whereas the feasibility question is trivial for the formulation in~\cite{KosutTTransIT08} and the paper characterizes the optimal rates of communication. Kosut and Tong also studied the byzantine CEO problem, in which the setup is similar except that the decoder's~(CEO) goal is to estimate an underlying source conditioned on which the observations of the users are conditionally independent. They obtained bounds on the rate-region~\cite{KosutTISIT08,KosutTISIT09} and the error-exponents~\cite{KosutTAllerton08}. In addition, broadly related are works on byzantine users in multiterminal information theory and network coding such as~\cite{Jaggi7,WangSK10,Yener,KTong,FanKW18,SangwanNP22,NehaBDP23,NehaBDP25} apart from the extensive literature in the areas of cryptography and distributed computing.
}

\section{Notation}

Sets are denoted by calligraphic letters. ${\mathcal A}^c$ denotes the complement of set ${\mathcal A}$. For a positive integer $n$, we denote $\{1,\ldots,n\}$ by $[n]$.
For $\cA=\{j_1,j_2,\ldots,j_l\}\subseteq[k]$ with $j_1<j_2<\ldots<j_l$, we write $x_{\cA}$ to denote $(x_{j_1},x_{j_2},\ldots,x_{j_l})$.
All probability mass functions (p.m.f.) and conditional p.m.fs (i.e., channels) are denoted exclusively by $P,Q,W$ or their variants; their domains will be omitted where it is clear from the context. We write $P^n$ to denote the p.m.f. of length-$n$ vector of independent and indentically distributed random variables, each with p.m.f. $P$.
We define $\cP(\cV|\cU)$ as the set of all channels from $\cU$ to $\cV$, i.e., the set of all conditional p.m.f.s $W_{V|U}$. The conditional p.m.f. of $n$-fold independent use of a channel $W_{V|U}$ will be denoted by $W^n_{V|U}$.

$\type{u^n}$ denotes the empirical distribution (type) of the vector $u^n=(u_1,u_2,\ldots,u_n)$.
The set of all empirical distributions (types) of vectors in $\cU^n$ will be denoted by $\cP_n(\cU)$; we drop $\cU$ from this notation when it is clear from the context. For a $Q_{U}\in\cP_n(\cU)$, we write $\cT_{U}$ to denote the set (type class) of all vectors in $\cU^n$ with empirical distribution $Q_{U}$ (the type involved will be clear from the context).

$\dtv(Q,Q')$ denotes the total variation distance between p.m.fs $Q$ and $Q'$. $\bb{1}$ is the indicator function; for instance, $\bb{1}_{a=b}=\bb{1}\{a=b\}$ takes on value one if $a=b$ is true and 0 otherwise.
$\ham(x^n,\hx^n)= \frac{1}{n} \sum_{t=1}^n {\bb{1}}_{x_t\neq \hx_t}$ denotes the {\em normalized/average} Hamming distortion between $x^n$ and $\hx^n$.
We write $g(n)=\Omega(h(n))$ to mean $\liminf_{n\to\infty} g(n)/h(n)>0$ (\`a la Knuth~\cite{Knuth-Omega}).
All $\log$ are natural logarithms and $\exp$ are to base $e$. 

For jointly distributed random variables $U,V,W$, we write ``$U\mc V\mc W$ is a Markov chain'' to mean that $U$ and $W$ are conditionally independent conditioned on $V$; we often suppress the phrase "is a Markov chain" in the sequel.

For a function $f:\cU\times\cV\to\cW$, we write $f(u^n,v^n)$ for vectors $u^n=(u_1,\ldots,u_n)$ and $v^n=(v_1,\ldots,v_n)$ to mean $f(u^n,v^n)=(f(u_1,v_1),\ldots,f(u_n,v_n))$. More generally, we extend function definitions along the same lines when the domain is a product of a finite number of sets (and not just two as above).

\section{Problem Statement}\label{sec:problem}

Consider the following $(k,s)$-{\em byzantine distributed source coding} problem which has a decoder and $k$ users (source nodes) of which at most $s\leq k$ are controlled by an adversary. Let $\cX_1,\ldots,\cX_k,\cY$ be finite alphabets. Suppose $P_{X_1\ldots X_k Y}$ is a known p.m.f. over $\cX_1\times\ldots\times\cX_k\times\cY$. Let $(X_{1,t},\ldots,X_{k,t},Y_t)\sim P_{X_1\ldots X_k Y}$, $t\in[n]$, be independent and identically distributed (i.i.d.) over the (discrete) time index $t$. For $i\in[k]$, (source) node-$i$ observes $X_i^n:=(X_{i,1},\ldots,X_{i,n})$ and the decoder observes the side-information $Y^n$. The decoder is interested in recovering a function $f$ of the sources (i.e., the domain of $f$ is $\cX_1\times\ldots\times\cX_k\times\cY$). Specifically, if $Z_t=f(X_{1,t},\ldots,X_{k,t},Y_t)$, $t\in[n]$, the decoder desires to recover $Z^n$.

The nodes are required to send their observations\footnote{Our model assumes that there are no communication constraints which require the nodes to compress their observations. As mentioned, the focus of this paper is on limitations on what the decoder can recover imposed by a deficit of trust rather than that of communication. 
} to the decoder. We want our decoder to either recover $Z^n$ (with a vanishing average Hamming distortion) or correctly identify {\em one} of the users controlled by the adversary\footnote{As was already mentioned, once such a user is identified, this user may be removed to obtain a new instance of the problem with $k-1$ users and at most $s-1$ users controlled by the adversary. 
}. Let the decoder be $\phi:\cX_1^n\times\ldots\times\cX_k^n\times\cY^n\to [k]\cup\cZ^n$, where $\cZ$ is the co-domain of $f$.

Let $\gamma>0$. Suppose $\cA\subseteq[k]$ is such that $|\cA|\leq s$ (note that $\cA$ may be empty). When the adversary controls the nodes in $\cA$, based on the observations $X_{\cA}^n:=(X_i^n)_{i\in\cA}$, it produces purported observations $\bX_{\cA}^n$ which are sent to the decoder. Let $W_{\cA}$ denote the randomized map (not necessarily memoryless) that the adversary uses to map $X_{\cA}^n$ to $\bX_{\cA}^n$. Both $\cA$ and $W_{\cA}$ are unknown to the decoder (other than the fact that $|\cA|\leq s$). 

To define the error event, let ${\cA}^c=[k]\setminus\cA$ and\footnote{We use the $\langle\; \rangle$ notation to indicate that the reported observation vectors in $\langle\bX_{\cA}^n,X_{{\cA}^c}^n\rangle$ are arranged in the order of $i=1,2,\ldots,k$ (and not with the ones in $\cA$ first followed by those in ${\cA}^c$).}
\begin{align*}
    \cE_1(\cA) &= (\phi(\langle\bX_{\cA}^n,X_{{\cA}^c}^n\rangle,Y^n)\in {\cA}^c)\\
    \cE_2(\gamma,\cA) &= \big((\phi(\langle\bX_{\cA}^n,X_{{\cA}^c}^n\rangle,Y^n)\notin [k])\cap\\
    &\qquad(\ham({\phi(\langle\bX_{\cA}^n,X_{{\cA}^c}^n\rangle,Y^n),Z^n}) > \gamma)\big),
\end{align*}
where the average Hamming distortion $\ham(\hz^n,z^n)$ measures the fraction of locations where the vectors differ.
We define the error event when the adversary controls $\cA$ as $\cE(\gamma,\cA)=\cE_1(\cA)\cup \cE_2(\gamma,\cA).$
Thus, an error occurs if either the decoder identifies a node outside $\cA$ or it outputs an estimate of $Z^n$ which incurs an average Hamming distortion more than $\gamma$. Notice that for $\cA=\varnothing$ (i.e., when the adversary is absent), an error occurs unless the decoder outputs an estimate $\hZ^n$ and it is of average Hamming distortion no larger than $\gamma$.
Denote the probability of the error event by
\[ \eta(\gamma,\cA,W_{\cA}) = \pr(\cE(\gamma,\cA)),\]
where the probability is evaluated under the joint distribution
\[ P(x_{[k]}^n,y^n,\bx_{\cA}^n) = \left( \prod_{t=1}^n P_{X_{[k]},Y}(x_{[k],t},y_t)\right) W_{\cA}(\bx_{\cA}^n|x_{\cA}^n).\]

For $\gamma>0$, the error probability of the decoder $\phi$ is defined as the maximum over the choices available to the adversary (including $\cA=\varnothing$ corresponding to an inactive/absent adversary)
\[ \epsilon(\gamma,\phi) = \max_{\cA\subseteq[k]: |\cA|\leq s} \;\sup_{W_{\cA}} \; \eta(\gamma,\cA,W_{\cA}).\]
We say that a function $f$ with domain $\cX_1\times\ldots\times\cX_k\times\cY$ is $s$-{\em robustly recoverable} if, for all $\gamma>0$, there is a sequence of decoders $\phi_n$, $n\in\bbZ^+$, such that 
$\liminf_{n\to\infty}\epsilon(\gamma,\phi_n)=0$. We note that our achievability results are shown with $\lim$.

\begin{figure*}[h]
\centering

\begin{tikzpicture}[node distance=0.7cm, every node/.style={draw, minimum size=0.7cm, align=center}, scale = 0.4]

\node (user1_left) [rectangle, red] {\scriptsize \(Q^n_{\uX_1|\tX_1}\)};
\node (user2_left) [rectangle, below=of user1_left] {\footnotesize User~2};

\node (adv) [draw, minimum height=1.4cm, minimum width=2cm, right=0cm of user1_left, xshift = -1.4cm, red, dashed, dash pattern=on 2pt off 1pt] {};
\node [draw = none, below = of user1_left, anchor = center, yshift = 0.55cm, red] {\scriptsize User~1};
\draw[->] ($(user1_left.west)+(-1.2,0)$) -- (user1_left.west) node[midway, left, draw=none, xshift=-0.2cm] {\footnotesize \(\tX_1^n\)};
\draw[->] ($(user2_left.west)+(-1.2,0)$) -- (user2_left.west) node[midway, left, draw=none, xshift=-0.2cm] {\footnotesize \(\uX_2^n\)};
\draw[->, red] (user1_left.east) -- node[above, draw = none, yshift = -0.1cm] {\scriptsize\(\uX_1^n\)} ($(adv.east)$);

\node (decoder) [rectangle, right=2.5cm of $(user1_left)!0.5!(user2_left)$, text width=1cm] {\footnotesize Decoder};

\draw[->] ($(decoder.south)+(0,-0.75)$) -- (decoder.south) node[midway, below, draw=none] {\footnotesize \(\uY^n\)};

\draw[->, red] (adv.east) -- ($(decoder.west)+(0,0.3)$);
\draw[->] (user2_left.east) -- ($(decoder.west)+(0,-0.3)$);

\draw[->] (decoder.east) -- ($(decoder.east)+(1,0)$) node[midway, right, draw=none, xshift=0.3cm] {\footnotesize \(\phi(\uX_1^n, \uX_2^n,  \uY^n)\)};

\node (user1_right) [rectangle, right=9.6cm of user1_left] {\footnotesize User-1};
\node (user2_right) [rectangle, below=of user1_right, red] {\scriptsize \(Q^n_{\uX_2|\tX_2}\)};

\node (adv_mirror) [draw, minimum height=1.4cm, minimum width=2cm, left=0cm of user2_right, xshift = 1.4cm, red, dashed, dash pattern=on 2pt off 1pt] {};

\node [draw = none, below = of user2_right, anchor = center, yshift = 0.55cm, red] {\scriptsize User-2};
\draw[->] ($(user1_right.east)+(1.2,0)$) -- (user1_right.east) node[midway, right, draw=none, xshift=0.2cm] {\footnotesize \(\uX_1^n\)};
\draw[->] ($(user2_right.east)+(1.2,0)$) -- (user2_right.east) node[midway, right, draw=none, xshift=0.2cm] {\footnotesize \(\tX_2^n\)};
\draw[->, red] (user2_right.west) -- node[above, draw = none, yshift = -0.1cm] {\scriptsize\(\uX_2^n\)} ($(adv_mirror.west)$);

\node (decoder_right) [rectangle, left=2.5cm of $(user1_right)!0.5!(user2_right)$, text width=1cm] {\footnotesize Decoder};

\draw[->, red] (adv_mirror.west) -- ($(decoder_right.east)+(0,-0.3)$);
\draw[->] (user1_right.west) -- ($(decoder_right.east)+(0,0.3)$);

\draw[->] ($(decoder_right.south)+(0,-0.75)$) -- ($(decoder_right.south)$) node[midway, below, draw=none] {\footnotesize \(\uY^n\)};
\draw[->] (decoder_right.west) -- ($(decoder_right.west)+(-1,0)$) node[midway, right, draw=none, xshift=0.5cm] {};

\end{tikzpicture}
\caption{For any $Q_{\uX_1\tX_1\uX_2\tX_2\uY}$ satisfying the conditions (i) and (ii) in Definition~\ref{def:1-viable}, the observations reported to the decoder and the decoder's side-information can be jointly distributed as $Q_{\uX_1\uX_2\uY}$ i.i.d. under the two scenarios shown to the left and right. Here, the malicious users are shown in color.}
\label{fig:2user-converse}
\end{figure*}
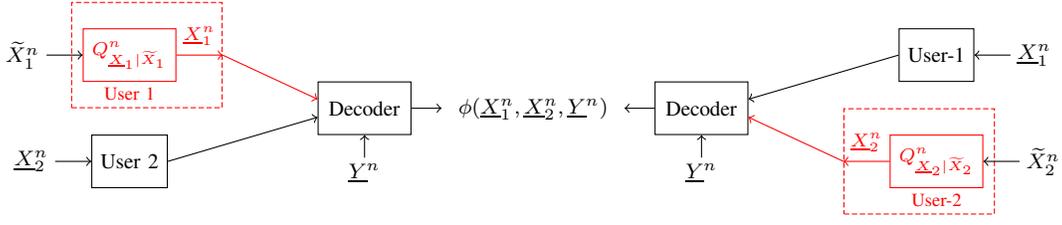

Before presenting our characterization of robustly recoverable functions for all $s\leq k$, to illustrate the main ideas, we discuss the case of $k=2$ users with adversary controlling at most $s=1$ of them in the next section. As mentioned in Section~\ref{sec:intro} and as we show in Appendix~\ref{app:k=1}, the scheme in~\cite{narayanan2023complete} is in fact optimal for $k=2, s=1$ though it does not generalize for $s>1$. In addition to presenting the idea of the converse, our purpose is to present a scheme which in fact generalizes to all values of $s\leq k$ in the simplest setting. 

\section{The $(k=2, s=1)$ case}\label{sec:2user}

\begin{definition}\label{def:1-viable} We say a function $f:\cX_1\times\cX_2\times\cY\to\cZ$ is 1-{\em viable} 
if for every joint p.m.f. $Q_{\uX_1\tX_1\uX_2\tX_2\uY}$ over $\cX_1\times\cX_1\times\cX_2\times\cX_2\times\cY$ 
such that
\begin{itemize}
\item[(i)] 
$Q_{\tX_1\uX_2\uY}=P_{X_1X_2Y}$ and $\uX_1 \mc \tX_1 \mc (\uX_2,\uY)$, and
\item[(ii)] 
$Q_{\uX_1\tX_2\uY}=P_{X_1X_2Y}$ and $\uX_2 \mc \tX_2 \mc (\uX_1,\uY)$,
\end{itemize}
we have (with probability 1)
\[ f(\tX_1,\uX_2,\uY) = f(\uX_1,\tX_2,\uY).\]
\end{definition}

\begin{theorem}[2 users of which at most 1 is corrupt] \label{thm:2user}
A function $f$ is 1-robustly recoverable if and only if it is 1-viable.
\end{theorem}

We start with a sketch of the proof of the converse (``only if'' part) which also provides an intuition for the characterization. For any $Q_{\uX_1\tX_1\uX_2\tX_2\uY}$ satisfying the conditions (i) and (ii) in Definition~\ref{def:1-viable}, we will argue that the observations reported to the decoder and its own side-information can be jointly distributed as $Q_{\uX_1\uX_2\uY}$ i.i.d. under two scenarios (see Figure~\ref{fig:2user-converse}):
\begin{enumerate}
\item[(i)] adversary controls user~1 and uses the discrete memoryless channel (DMC) $Q_{\uX_1|\tX_1}$ ($n$-times) as its randomized map $W_{\{1\}}$ with the underlying original source being $Q_{\tX_1\uX_2\uY}$ (which is $P_{X_1X_2Y}$) i.i.d., and 
 \item[(ii)] adversary controls user~2 and uses the (DMC) $Q_{\uX_2|\tX_2}$ as $W_{\{2\}}$ with the underlying original source being $Q_{\uX_1\tX_2\uY}$ (which is also $P_{X_1X_2Y}$) i.i.d. 
 \end{enumerate}
 Since the decoder is unable to tell between these scenarios, it must (w.h.p.) output an estimate $\hZ^n$. Moreover, this estimate must match (with vanishing distortion) the true $Z^n$ under both scenarios. This is possible only if the true $Z^n$ in both cases are themselves (substantially) equal. From this we will conclude that $f(\tX_1,\uX_2,\uY) = f(\uX_1,\tX_2,\uY)$.  
See Section~\ref{sec:2user-converse-proof} for details.


The proof of achievability (``if part'') relies on a decoder which checks whether the empirical joint distribution (type) of the reported observations and its own side-information (i.e., inputs to the decoder) can be ``explained'' by exactly one of the users behaving adversarially. If so, it names that user as the adversary and stops. Otherwise, the decoder outputs an estimate of $Z^n$ by applying the function element-wise ($n$-times) on its inputs. The intuition is that when there are alternative ``explanations'' and the decoder cannot be sure who the adversary is, the conditions in the theorem can be used to argue that applying the function element-wise will still lead to a (substantially) correct estimate of $Z^n$ (even if one of the reported observations itself may not be accurate).

We denote by $\cV_1$ the set of {\em single-letter} (i.e., $n=1$) joint distributions of the inputs to the decoder that user-1 acting adversarially can induce (using some channel $W_{\bX_1|X_1}$ to produce its report).
\begin{align*} \cV_1=\big\{& Q_{\uX_1\uX_2\uY}: \exists W_{\bX_1|X_1} \text{s.t. } 
              Q_{\uX_1\uX_2\uY}(\ux_1,\ux_2,\uy)=
\\&\quad\sum_{x_1} P_{X_1X_2Y}(x_1,\ux_2,\uy)W_{\bX_1|X_1}(\ux_1|x_1), \forall \ux_1,\ux_2,\uy\big\}.
\end{align*}
For $\delta>0$, let $\cV_1^{\delta}=\bigcup_{Q\in\cV_1} \cB(Q,\delta)$, where $\cB(Q,\delta)=\{ Q': \dtv(Q,Q')\leq\delta\}$ is the set of p.m.f.s within total-variation distance $\delta$ of $Q$. Analogously, $\cV_2, \cV_2^{\delta}$ are defined for user-2.
\paragraph*{Decoder} Let $\delta>0$. For $n\in\mathbb{Z}^+$, consider the decoder
\begin{align*}
    \phi_n(\ux_1^n,\ux_2^n,\uy^n)=
        \begin{cases}
            1, & \type{\ux_1^n,\ux_2^n,\uy^n}\in \cV_{1}^{\delta}\cap({\cV_{2}^{\delta}})^c\\
            2, & \type{\ux_1^n,\ux_2^n,\uy^n}\in ({\cV_{1}^{\delta}})^c\cap{\cV_{2}^{\delta}}\\
            g(\ux_1^n,\ux_2^n,\uy^n), & \text{otherwise},
        \end{cases}
\end{align*}
where we will define $g$ presently. Notice that the obvious choice of $g$ is $f$ itself since, if the adversary is absent, we would like the output to be (close to) $f$ applied on $(x_{1,t},x_{2,t},y_t)_{t\in[n]}$. This intuition turns out to be correct when $P_{X_1X_2Y}$ has full support. When the support is not full, clearly the definition of $f$ outside the support should not affect whether $f$ is robustly recoverable (also notice that Definition~\ref{def:1-viable}  depends only on the definition of $f$ in the support). Our decoder may indeed need to deal with inputs which fall outside the support. We obtain a $g$ for the decoder through the following lemma. Besides addressing the support, it also shows a property of 1-viable $f$'s as we explain after stating the lemma\footnote{A word on the notation is in order.
We use underbar ($\uX_1,\uX_2,\uY$) and $Q_{\uX_1\uX_2\uY}$ to denote the ``views'' of the decoder (only one of $\uX_1,\uX_2$ is potentially produced by an adversary, but the decoder does not know which one and the notation reflects this fact). We also use $Q_{\uX_1\tX_1\uX_2\tX_2\uY}$ and $\tX_1,\tX_2$ to denote the underlying true observations which are replaced by the adversary; often both $\tX_1,\tX_2$ are used together which present two plausible interpretations (see the discussion of the converse above). On the other hand, to represent things from the adversary's perspective, as in Lemma~\ref{lem:decoding-function-2user}, we use overbar $\bX_1,\bX_2$ and $W_{\bX_1|X_1},W_{\bX_2|X_2}$ to denote the replacement (manipulated) random variables and channels used to produce them, respectively.}.
\begin{lemma}\label{lem:decoding-function-2user}
If $f:\cX_1\times\cX_2\times\cY\to\cZ$ is 1-viable, there is a $g:\cX_1\times\cX_2\times\cY\to\cZ$ such that for any pair $(W_{\bX_1|X_1},W_{\bX_2|X_2})$ 
satisfying, for all $\ux_1,\ux_2,\uy$,
\begin{align*}
  \sum_{x_1}&P_{X_1X_2Y}(x_1,\ux_2,\uy)W_{\bX_1|X_1}(\ux_1|x_1)\\&\qquad=\sum_{x_2}P_{X_1X_2Y}(\ux_1,x_2,\uy)W_{\bX_2|X_2}(\ux_2|x_2), \text{ we have}
\end{align*}
\begin{itemize}
\item[(i)] $f(X_1,X_2,Y)=g(\bX_1,X_2,Y)$ under 
$P_{X_1X_2Y}W_{\bX_1|X_1}$,
\item[(ii)] $f(X_1,X_2,Y)=g(X_1,\bX_2,Y)$ under 
$P_{X_1X_2Y}W_{\bX_2|X_2}$.
\end{itemize}
If $P_{X_1X_2Y}(x_1,x_2,y)>0$, then $g(x_1,x_2,y)=f(x_1,x_2,y)$.
\end{lemma}
The lemma (which is proved in Section~\ref{sec:decoding-function-2user-proof}) says that, for a 1-viable $f$, there is a $g$ with the following property: Consider any pair of DMCs $W_{\bX_1|X_1},W_{\bX_2|X_2}$ that a malicious user 1 or 2, respectively (not simultaneously) may employ to attack such that they induce the same joint distribution on the reported observations and side-information at the decoder. Then, under either attack, applying $g$ on the reported observations and side-information recovers $f(X_1,X_2,Y)$. 

Note that the above lemma does not immediately lead to a proof of achievability of Theorem~\ref{thm:2user} since (i) the adversary need not use a DMC and (ii) even if a DMC is used, the adversary could select one which induces a distribution too close to $\cV_1\cap\cV_2$ for the decoder to detect reliably who the adversary is (note that the adversary may select its attack for each given decoder, i.e., for each $\delta,n$).
We employ the method of types~\cite[Chapter~2]{csiszar2011information} to show the following (in Section~\ref{sec:2user-typicality-proof}):
\begin{lemma}\label{lem:2user-typicality}
Consider a 1-viable $f$ along with $g$ from Lemma~\ref{lem:decoding-function-2user}.
Suppose $X_1^n,X_2^n,Y^n,\bX_1^n$ are jointly distributed as
\[P(x_1^n,x_2^n,y^n,\bx_1^n) = \left( \prod_{t=1}^n P_{X_1X_2Y}(x_{1,t},x_{2,t},y_t)\right) W_1(\bx_1^n|x_1^n),\]
where $W_1$ is {\em not} necessarily memoryless. Then
\begin{itemize}
    \item[(i)] $\pr(\type{\bX_1^n,X_2^n,Y^n}\notin\cV_1^{\delta}) \leq 2^{-\Omega(n)}$ for all $\delta>0$, and
    \item[(ii)] there is a function $\gamma:\mathbb{R}^+\to\mathbb{R}^+$, which does not depend on $W_1$ or $n$, such that $\gamma(\delta)\to0$ as $\delta\to 0$ and
\begin{align*}\pr\big((\type{\bX_1^n,X_2^n,Y^n}&\in\cV_1^{\delta}\cap\cV_2^{\delta})\cap(\ham(f(X_1^n,X_2^n,Y^n),\\ &\qquad g(\bX_1^n,X_2^n,Y^n))>\gamma(\delta))\big)\leq 2^{-\Omega(n)}.
\end{align*}
\end{itemize}
The factors hidden in $\Omega(n)$ do not depend on $W_1$.
\end{lemma}
The lemma considers the case where a (potentially) malicious user~1 employs a (not necessarily memoryless) $W_1$ to produce a purported $\bX_1^n$ from $X_1^n$. Then, the lemma asserts that (i) the empirical distribution of the reported observations and side-information will (w.h.p.) lie in $\cV_1^{\delta}$. Hence the decoder will not make the error of naming user~2 as malicious. Furthermore, (ii) the event where it decides to output an estimate (which depends on whether the empirical distribution also lies in $\cV_2^{\delta}$) and the output exceeds distortion $\gamma(\delta)$ has a vanishingly small probability, where $\gamma(\delta)$ vanishes as $\delta\to0$. It is easy to see that this lemma, along with its analog for user~2, proves the achievability of Theorem~\ref{thm:2user} -- when either user is malicious (but not both), the above argument suffices; when both users are honest, the lemma (and its analog for user~2) may be invoked with $W$ being the identity channel to conclude (from its part~(i)) that the empirical distribution lies in $\cV_1^\delta\cap\cV_2^\delta$ w.h.p. and the decoder will output an estimate; moreover, (by part~(ii)) this output will incur a distortion no more than $\gamma(\delta)$ w.h.p.

For the sake brevity, here we present a proof of the lemma when the $W_1$ is a DMC; the proof for the general case (in Section~\ref{sec:2user-typicality-proof}) uses an additional symmetrization trick via a random permutation of $[n]$; the same idea was used in the achievability proof of~\cite{narayanan2023complete}. Suppose $W_1$ is a DMC, i.e., let $W_1(\bx_1^n|x_1^n)=W_{\dmc}^n(\bx_1^n|x_1^n):=\prod_{t=1}^n W_{\dmc}(\bx_{1,t}|x_{1,t})$ for some channel $W_{\dmc}$. 
Then, part~(i) of the lemma follows from standard properties of types (\cite[Lemma~2.6]{csiszar2011information} and Pinsker's inequality) since $(X_1^n,X_2^n,Y^n,\bX_1^n)\sim P_{X_1X_2Y}W_{\dmc}$ i.i.d. To show part~(ii), we use the following technical lemma which is a continuous version of Lemma~\ref{lem:decoding-function-2user}. Notice that (part (i) of) Lemma~\ref{lem:decoding-function-2user} says that there is a $g$ such that for every $W_{\bX_1|X_1}$ for which the induced p.m.f. 
$Q_{\uX_1\uX_2\uY}(\ux_1,\ux_2,\uy):=\sum_{x_1}P_{X_1X_2Y}(x_1,\ux_2,\uy)W_{\bX_1|X_1}(\ux_1|x_1)$ belongs to $\cV_1\cap\cV_2$ (which is another way of saying that there is a corresponding $W_{\bX_2|X_2}$ which induces the same $Q_{\uX_1\uX_2\uY}$), we have $\pr(f(X_1,X_2,Y)\neq g(\bX_1,X_2,Y))=0$ under $P_{X_1X_2Y}W_{\bX_1|X_1}$. The following lemma (proved in Section~\ref{sec:2user-technical-single-letter-proof}) generalizes this:

\begin{lemma}\label{lem:2user-technical-single-letter}
For a 1-viable $f$ along with $g$ as in Lemma~\ref{lem:decoding-function-2user}, 
there is a 
$\gamma:\bbR^+\to\bbR^+$ such that $\gamma(\delta)\to0$ as $\delta\to0$ 
satisfying the following: for $\delta>0$, if 
$Q_{X_1\uX_2\uY\uX_1}$ is such that $Q_{\uX_1\uX_2\uY}\in\cV_1^{\delta}\cap\cV_2^{\delta}$ and $\exists$ $W_{\bX_1|X_1}$ for which $\dtv(Q_{X_1\uX_2\uY\uX_1}, P_{X_1X_2Y}W_{\bX_1|X_1})\leq \delta$, then $\pr(f(X_1,\uX_2,\uY)\neq g(\uX_1,\uX_2,\uY))\leq \gamma(\delta)$ under the p.m.f. $Q_{X_1\uX_2\uY\uX_1}$. 
\end{lemma}
 
Now, to see part~(ii) of Lemma~\ref{lem:2user-typicality} for 
$(X_1^n,X_2^n,Y^n,\bX_1^n)\sim P_{X_1X_2Y}W_{\dmc}$ i.i.d., let $h_{x_1^n,x_2^n,y^n,\bx_1^n}$ be 1 if
$(\type{\bx_1^n,x_2^n,y^n}\in\cV_1^{\delta}\cap\cV_2^{\delta})\cap(\ham(f(x_1^n,x_2^n,y^n),g(\bx_1^n,x_2^n,y^n))>\gamma(\delta))$
 and 0 otherwise. Define
 \begin{align*}\cQ_n(\delta)&=\{Q_{X_1\uX_2\uY\uX_1}\in\cP_n
 :Q_{\uX_1\uX_2\uY}\in \cV_1^{\delta}\cap\cV_2^{\delta},\\ 
 &\qquad\qquad
 \dtv(Q_{X_1\uX_2\uY\uX_1}, P_{X_1X_2Y}W_{\dmc})\leq \delta\}.\end{align*} 
 For any $(x_1^n,x_2^n,y^n,\bx_1^n)$, the average distortion $\ham(f(x_1^n,x_2^n,y^n),g(\bx_1^n,x_2^n,y^n)$ coincides with $\pr(f(X_1,\uX_2,\uY)\neq g(\uX_1,\uX_2,\uY))$ under the p.m.f. $Q_{X_1\uX_2\uY\uX_1}:=\type{x_1^n,x_2^n,y^n,\bx_1^n}$.
Hence, by Lemma~\ref{lem:2user-technical-single-letter}, $h_{x_1^n,x_2^n,y^n,\bx_1^n}=0$ if $\type{x_1^n,x_2^n,y^n,\bx_1^n}\in\cQ_n(\delta)$.
  Using this observation in step (a) below, the probability of interest in part (ii) of Lemma~\ref{lem:2user-typicality} is
 \begin{align*}
 &\sum_{x_1^n, x_2^n, y^n,\bx_1^n} P^n_{X_1 X_2 Y}(x_1^n,x_2^n,y^n) W^n_{\dmc}(\bx^n_1|x^n_2)h_{x_1^n,x_2^n,y^n,\bx_1^n}\\
    & \stackrel{\text{(a)}}{=} \sum_{\substack{x_1^n,x^n_2,y^n,\bx_1^n:\\\type{x_1^n,x^n_2,y^n,\bx_1^n} \in \cQ_n(\delta)^c}} \hspace{-1cm}P^n_{X_1 X_2 Y}(x_1^n,x_2^n,y^n)W^n_{\dmc}(\bx_1^n|x^n)
\stackrel{\text{(b)}}{\leq} \; 2^{-\Omega(n)},
 \end{align*} where (b) is a consequence of the fact that $\pr(\dtv(\type{X_1^n,X_2^n,Y^n,\bX_1^n},P_{X_1X_2Y}W_{\dmc})>\delta)=2^{-\Omega(n)}$ which follows from \cite[Lemma~2.6]{csiszar2011information} and Pinsker's inequality.

\section{Main Result}\label{sec:main-result}
\paragraph*{Notation} For this section, we use the following compact notation: Recall that for $\cA=\{j_1,j_2,\ldots,j_l\}\subseteq[k]$, we write $X_{\cA}$ to denote $(X_{j_1},X_{j_2},\ldots,X_{j_l})$. Similarly $\uX_{\cA}$ and so on. We will also write $\tX^i_{\cA}$ to denote $(\tX^i_{j_1},\tX^i_{j_2},\ldots,\tX^i_{j_l})$, where $\tX^i_j$ denotes a random variable indexed by $(i,j)$.

\begin{definition}[$s$-viable] \label{def:s-viable} We say a function $f$ with domain $\cX_1\times\ldots\cX_k\times\cY$ is s-{\em viable} 
if for any collection $\cA_1,\cA_2,\ldots,\cA_m$ of distinct subsets of $[k]$ such that $|\cA_i|\leq s,\, i\in[k]$ and $\bigcap_{i=1}^m \cA_i=\varnothing$, under every joint p.m.f. $Q_{\uX_{[k]},\uY,\left(\tX^{i}_{\cA_i}\right)_{i\in[m]}}$ over $\cX_1\times\ldots\cX_k\times\cY\times\prod_{i=1}^m\prod_{j\in \cA_i}\cX_j$ 
satisfying, for each $i\in[m]$,
\begin{itemize}
\item[(a)] 
$Q_{\ina{\tX^i_{\cA_i},\,\uX_{{\cA}^c_i}},\uY}=P_{X_{[k]}Y}$, and
\item[(b)] 
$\uX_{\cA_i} \mc \tX^i_{\cA_i} \mc (\uX_{{\cA}^c_i},\uY)$,
\end{itemize}
we have, for all $i,i'\in[m]$, (with probability 1)
\[ f\inp{\ina{\tX^{i}_{\cA_i},\,\uX_{{\cA}^c_i}},\uY} = f\inp{\ina{\tX^{i'}_{\cA_{i'}},\,\uX_{{\cA}^c_{i'}}},\uY}.\]
\end{definition}

The ideas in the previous section readily generalize to give the following single-letter characterization.
\begin{theorem}\label{thm:main-threshold}
A function $f$ with domain $\cX_1\times\ldots\times\cX_k\times\cY$ is $s$-{robustly} {recoverable} if and only if it is $s$-viable.
\end{theorem}
\begin{remark}
In Appendix~\ref{app:general-adversary-structure} we obtain the characterization for a more general case where the adversary may control the users in any one of the subsets in a given collection $\coll$ (this collection is called the {\em adversary structure} in secret sharing/cryptography literature). The above theorem follows if $\coll$ consist of subsets of cardinality at most $s$.
\end{remark}

\begin{remark}\label{rmk:decidability}
We show in Appendix~\ref{app:LP-remark} that the viability condition in Definition~\ref{def:s-viable} can be checked using a linear program. 
\end{remark}

%% file: example.tex
\section{An Example} \label{sec:example}

To illustrate the characterization in Theorem~\ref{thm:main-threshold}, we consider the following $k=3$ and threshold $s=2$ toy-example. Let $U,V,W$ be i.i.d. uniform bits and $X_1=U$,
\begin{align*}
(X_2,X_3)=\begin{cases}
             (U,U),&V=0\\
             (\eras_2,U),&V=1,W=0\\
             (U,\eras_3),&V=1,W=1
        \end{cases} \quad
        Y=\begin{cases}
            U,&V=0\\
            \eras,&V=1
        \end{cases}
\end{align*}
i.e., $X_1$ is a uniform bit $U$, while the side-information $Y$ is an erased version of $U$ with erasure probability 1/2 (erasure state represented by $V$ which is independent of $U$). When the bit is unerased (i.e.,  $V=0$), both $X_2=X_3=U$. When the bit is erased for $Y$ (i.e., $V=1$), exactly one of $X_2$ and $X_3$ is equal to $U$ while the other is an erasure; which of these occurs is equiprobable and is represented by the uniform bit $W$ which is independent of $U,V$. It is clear that if users 2 and 3 collude, since working together they can learn $U,V$ (and $W$), they may resample $W$ and replace $X_2,X_3$ according to the resampled $W$. Thus, the decoder cannot hope to learn $W$. 
Can it learn $U$ where it is erased in $Y$? It turns out it can.
\begin{claim}\label{cl:(3,2)-example}
    $f(X_1,X_2,X_3,Y):=(U,V)$ is 2-viable. Moreover, any 2-viable function is a function of $f$.
\end{claim}
The second statement will follow from the first and the argument above that $W$ cannot be recovered robustly. 
To show that $U$ can be learned robustly, we argue in Appendix~\ref{app:proof-for-(3,2)-example} that the above $f$ is 2-viable. 
The intuition is that if user~1 colludes with user~2 (resp., 3), the reports of $X_2$ (resp. $X_3$) on whether an erasure happened for user~2 (resp., 3) is mostly (we omit such qualifiers below) trustworthy since, not knowing $V$ fully, user~2 (resp., 3) cannot convincingly lie as (a) reporting an erasure when there is no erasure is liable to be discovered by the decoder, and (b) the average number of erasures it reports must match the statistics. User~1, colluding with user~2, cannot lie unless user~2 sees an erasure (since otherwise they do not know for sure whether the decoder experienced an erasure). However, doing so (which results in conflicting reports from users~1 \&~3) leaves a distinct signature -- decoder observes a conflict between the reports of users 1 \&~3 mostly when user~2 reports an erasure. 
This could \emph{only} be the result of either users 1 \&~2 colluding as described above, or users 2 \&~3 colluding, both possibilities implicating user~2.
Therefore, user~1, when colluding with user~2, cannot lie. By symmetry user~1 cannot do so when colluding with user~3 either. Finally, if users 2 \&~3 collude, user~1's report are anyway correct. Thus, unless the decoder is able to detect a malicious user through the means above, it may trust user~1's reports and recover $U$.


%% file: appendix-2-theorem.tex
\section{Proof of Theorem~\ref{thm:2user}}\label{sec:2user-proof}

\subsection{Converse}\label{sec:2user-converse-proof}

\textbf{Converse part of Theorem~\ref{thm:2user} (restated):} For a distribution $P_{X_1X_2Y}$, if a function $f:\cX_1 \times \cX_2 \times \cY \rightarrow \cZ$ is 1-robustly recoverable, then for any distribution $Q_{\uX_1\tX_1\uX_2\tX_2Y}$ 
such that
\begin{enumerate}[label=(\roman*)]
    \item   $Q_{\tX_1\uX_2\uY} = P_{X_1X_2Y}$ and $\uX_1 \mc \tX_1 \mc (\uX_2,\uY)$
    \item   $Q_{\uX_1 \tX_2 \uY} = P_{X_1 X_2 Y}$ and $\uX_2 \mc \tX_2 \mc (\uX_1,\uY)$,
\end{enumerate}
we must have (with probability 1)
$$f(\tX_1,\uX_2,\uY) = f(\uX_1,\tX_2,\uY).$$
\begin{proof}
    For any $Q_{\uX_1\tX_1\uX_2\tX_2\uY}$ satisfying the conditions (i) and (ii) above, the observations reported to the decoder and its own side-information are jointly distributed as $Q_{\uX_1\uX_2\uY}$ i.i.d. under two scenarios (see Figure~\ref{fig:2user-converse}):
    \begin{enumerate}[label=(\roman*)]
        \item The adversary controls user~1 (i.e., $\cA=\{1\}$). The underlying observations are $(\tX_1^n,\uX_2^n,\uY^n)$ distributed as $Q_{\tX_1 \uX_2 \uY} = P_{X_1 X_2 Y}$ i.i.d.  User~2 reports $\uX_2^n$ honestly and the side-information at the decoder is $\uY^n$. The adversary (user~1) produces its report $\uX_1^n$ by passing $\tX_1^n$ through the DMC $Q_{\uX_1|\tX_1}$. The fact that $\uX_1 \mc \tX_1 \mc (\uX_2,\uY)$ is a Markov chain ensures that the resulting $\uX_1^n$ is jointly distributed with the report $\uX_2^n$ from user~2 and decoder's side information $\uY^n$ as $Q_{\uX_1\uX_2\uY}$ i.i.d.
        \item  This is the analogous case with $\cA=\{2\}$. The underlying observations are $(\uX_1^n,\tX_2^n,\uY^n)$ distributed as $Q_{\uX_1 \tX_2 \uY} = P_{X_1 X_2 Y}$ i.i.d.  User~1 reports $\uX_1^n$ honestly while the adversary generates the report $\uX_2^n$ from user~2 by passing $\tX_2^n$ through the DMC $Q_{\uX_2|\tX_2}$. The fact that $\uX_2 \mc \tX_2 \mc (\uX_1,\uY)$ is a Markov chain ensures that the resulting $\uX_2^n$ is jointly distributed with the report $\uX_1^n$ from user~1 and decoder's side information $\uY^n$ as $Q_{\uX_1\uX_2\uY}$ i.i.d.
    \end{enumerate}

    Let $\gamma>0$ and $\phi_n$ be a decoder such that $\epsilon(\gamma,\phi_n)\leq \delta$. Then, for the first scenario, we have $\eta(\gamma,\cA,W_{\cA})\leq \delta$, where $\cA=\{1\}$ and $W_{\cA}=Q_{\uX_1|\tX_1}^{n}$. This implies that, under $Q_{\uX_1\tX_1\uX_2Y}$ i.i.d. and, therefore, also under $Q_{\uX_1\tX_1\uX_2\tX_2Y}$ i.i.d.,
    \begin{align}
        \pr(\phi_n(\uX_1^n,\uX_2^n,\uY^n)=2)&\leq \delta,\label{eq:twoconv1}\\
        \pr\big((\phi_n(\uX_1^n,\uX_2^n,\uY^n)\neq 1)\cap\quad\qquad\qquad\qquad\qquad\notag\\ (\ham(f(\tX_1^n,\uX_2^n,\uY^n),\phi(\uX_1^n,\uX_2^n,\uY^n)) > \gamma)\big)&\leq \delta.\label{eq:twoconv2}
    \end{align}
    Similarly, from the second scenario, under $Q_{\uX_1\tX_1\uX_2\tX_2Y}$ i.i.d.,
    \begin{align}
        \pr(\phi_n(\uX_1^n,\uX_2^n,\uY^n)=1)&\leq \delta,\label{eq:twoconv3}\\
        \pr\big((\phi_n(\uX_1^n,\uX_2^n,\uY^n)\neq 2)\cap\quad\qquad\qquad\qquad\qquad\notag\\ (\ham(f(\uX_1^n,\tX_2^n,\uY^n),\phi(\uX_1^n,\uX_2^n,\uY^n)) > \gamma)\big)&\leq \delta.\label{eq:twoconv4}
    \end{align}
    From \eqref{eq:twoconv1} and \eqref{eq:twoconv3}, we have
    \begin{align}
        \pr(\phi_n(\uX_1^n,\uX_2^n,\uY^n)\in\{1,2\})&\leq 2\delta. \label{eq:twoconv5}
    \end{align}
    Hence, with at least probability $1-2\delta$, the decoder outputs an estimate.
    Now observe that, under $Q_{\uX_1\tX_1\uX_2\tX_2Y}$ i.i.d.,
    \begin{align*}
        &\pr\left(\ham(f(\tX^n_1,\uX^n_2,\uY^n),f(\uX^n_1, \tX^n_2, \uY^n)) \leq 2\gamma \right)\\
        &\geq 
         \pr\big((\phi_n(\uX_1^n,\uX_2^n,\uY^n)\notin\{1,2\})\cap\\
         &\qquad\qquad (\ham(f(\tX_1^n,\uX_2^n,\uY^n),\phi(\uX_1^n,\uX_2^n,\uY^n)) \leq \gamma)\cap\\
         &\qquad\qquad\qquad (\ham(f(\uX_1^n,\tX_2^n,\uY^n),\phi(\uX_1^n,\uX_2^n,\uY^n)) \leq \gamma)\big)\\
        &\geq 1-4\delta,
    \end{align*}
    where the first inequality follows from the triangle inequality and the last inequality follows from \eqref{eq:twoconv2}, \eqref{eq:twoconv4}, and \eqref{eq:twoconv5} via the union bound.

    For any $\gamma>0$ and $\delta>0$, the above must hold for an increasing sequence of $n$. i.e., for any $\gamma>0$, passing to a subsequence if needed, $\pr(\ham(f(\tX^n_1,\uX^n_2,\uY^n),f(\uX^n_1, \tX^n_2, \uY^n)) \leq 2\gamma )\to0$ as $n\to\infty$. Since, by the law of large numbers, $\ham(f(\tX^n_1,\uX^n_2,\uY^n),f(\uX^n_1, \tX^n_2, \uY^n))\to \pr(f(\tX_1,\uX_2,\uY) \neq f(\uX_1,\tX_2,\uY))$ a.s., the result follows.

\end{proof}

\subsection{Achievability}
Here we flesh out the sketch of the proof in Section~\ref{sec:2user} for the achievability part of Theorem~\ref{thm:2user}. After a brief review of some properties of typical sets, we provide proofs of Lemmas~\ref{lem:decoding-function-2user},~\ref{lem:2user-technical-single-letter}, and~\ref{lem:2user-typicality} followed by the proof of achievability.

\subsubsection{Method of Types}
We recall some properties from \cite[Chapter~2]{csiszar2011information}. Let $X$ and $Y$ be two jointly distributed random variables according to a joint type $P_{XY}\in \cP_n(\cX\times \cY)$. For $(x^n, y^n)\in \cT^n_{XY}$, a distribution $Q$ on $\cX$, we have 
\begin{align}
|\cP_n(\cX)|&\leq (n+1)^{|\cX|}\label{eq:type1}\\
(n+1)^{-|\cX|}\exp\inp{nH(X)}&\leq \cT^n_{X}\leq \exp\inp{nH(X)}\label{eq:type2}\\
(n+1)^{-|\cX||\cY|}\exp\inp{nH(Y|X)}&\leq \cT^n_{Y|X}(x^n)\notag\\ &\leq \exp\inp{nH(Y|X)}\label{eq:type3}\\
(n+1)^{-|\cX|}\exp\inp{-nD(P_X||Q)}&\leq \sum_{\tilde{x}^n\in \cT^n_{X}}Q^n(\tilde{x}^n)\nonumber\\
& \leq \exp\inp{-nD(P_X||Q)}.\label{eq:type4}
\end{align}

\subsubsection{Proof of Lemma~\ref{lem:decoding-function-2user}}\label{sec:decoding-function-2user-proof}

\paragraph*{When $P_{X_1X_2Y}$ has full support} We start with a proof for the case when $P_{X_1X_2Y}$ has full support.
When $P_{X_1X_2Y}$ has full support, we define $g(x_1, x_2, y)=f(x_1, x_2, y)$ for all $x_1, x_2, y$. We will now show {\em (i)}.
Suppose
\begin{align}
&\bar{Q}_{\tX_1\uX_2\uY\uX_1}(\tx_1, \ux_2, \uy, \ux_1):=P_{X_1X_2Y}(\tx_1, \ux_2, \uy)\nonumber\\
&\qquad\qquad\inp{\inp{1-\lambda}W_{\bX_1|X_1}(\ux_1|\tx_1)+\lambda\mathbb{1}_{\ux_1=\tx_1}}, \text{ and}\label{eq:lem21}\\
&\tilde{Q}_{\uX_1\tX_2\uY\uX_2}(\ux_1, \tx_2, \uy, \ux_2):=P_{X_1X_2Y}(\ux_1, \tx_2, \uy)\nonumber\\
&\qquad\qquad\inp{\inp{1-\lambda}W_{\bX_2|X_2}(\ux_2|\tx_2)+\lambda\mathbb{1}_{\ux_2=\tx_2}}\label{eq:lem22}
\end{align} for some $0<\lambda<1$ and all $\tx_1, \tx_2, \ux_1, \ux_2$ and $\uy$. 
Then, $\bar{Q}_{\uX_1\uX_2\uY} = \tilde{Q}_{\uX_1\uX_2\uY}$ since, for any $\ux_1, \ux_2, \uy$, 
\begin{align*}
&\sum_{\tx_1}P_{X_1X_2Y}(\tx_1, \ux_2, \uy)\inp{\inp{1-\lambda}W_{\bX_1|X_1}(\ux_1|\tx_1)+\lambda\mathbb{1}_{\ux_1=\tx_1}}\\
&\stackrel{\text{(a)}}{=}\sum_{\tx_2}P_{X_1X_2Y}(\ux_1, \tx_2, \uy)\inp{\inp{1-\lambda}W_{\bX_2|X_2}(\ux_2|\tx_2)+\lambda\mathbb{1}_{\ux_2=\tx_2}}
\end{align*}
where $\text{(a)}$ follows from the assumption in the lemma and  because $\sum_{\tx_1}P_{X_1X_2Y}(\tx_1, \ux_2, \uy)\mathbb{1}_{\ux_1=\tx_1}  = \sum_{\tx_2}P_{X_1X_2Y}(\ux_1, \tx_2, \uy)\mathbb{1}_{\ux_2=\tx_2}$. 
Define ${Q}_{\uX_1\uX_2\uY}=\bar{Q}_{\uX_1\uX_2\uY} = \tilde{Q}_{\uX_1\uX_2\uY}$. 
Next,  define \begin{align}Q_{\uX_1\uX_2\uY\tX_1\tX_2}=Q_{\uX_1\uX_2\uY}\bar{Q}_{\tX_1|\uX_1\uX_2\uY}\tilde{Q}_{\tX_2|\uX_1\uX_2\uY}.\label{eq:lem23}
\end{align}$Q_{\uX_1\uX_2\uY\tX_1\tX_2}$ satisfies (i) and (ii) in Definition~\ref{def:1-viable}. Hence, by $1$-viability of $f$, we have $f(\tx_1,\ux_2,\uy) = f(\ux_1,\tx_2,\uy)$ for any $\ux_1, \ux_2,  \tx_1, \tx_2, \uy,$ in the support of $Q_{\uX_1\uX_2\tX_1\tX_2\uY}$.
Since $P_{X_1X_2Y}$ has full support, by definition of $\tilde{Q}_{\uX_1\tX_2\uY\uX_2}$, we have $\tilde{Q}_{\uX_1\uX_2\uY\tX_2}(\ux_1,\ux_2,\uy, \ux_2)>0$ for all $\ux_1,\ux_2,\uy$. This implies that $\tilde{Q}_{\tX_2|\uX_1\uX_2\uY}( \ux_2|\ux_1,\ux_2,\uy)>0$. Thus, for all $(\ux_1, \ux_2, \uy, \tx_1)$ for which $Q_{\uX_1\uX_2\uY\tX_1}(\ux_1, \ux_2, \uy, \tx_1)>0$, we have 
 $f(\tx_1,\ux_2,\uy) = f(\ux_1,\ux_2,\uy)$ . From \eqref{eq:lem21} and \eqref{eq:lem23}, this implies that 
\begin{align*}
1
&=\sum_{\ux_1, \tx_1, \ux_2, y}P_{X_1X_2Y}(\tx_1, \ux_2, \uy)\Big(\inp{1-\lambda}W_{\bX_1|X_1}(\ux_1|\tx_1)\\
&\qquad\qquad+\lambda\mathbb{1}_{\ux_1=\tx_1}\Big)\mathbb{1}_{f(\tx_1,\ux_2,\uy) = f(\ux_1,\ux_2,\uy)}\\
&\stackrel{\text{(a)}}{=}(1-\lambda)\sum_{\ux_1, \tx_1, \ux_2, y}P_{X_1X_2Y}(\tx_1, \ux_2, \uy)W_{\bX_1|X_1}(\ux_1|\tx_1)\\
&\qquad\qquad\mathbb{1}_{f(\tx_1,\ux_2,\uy) = f(\ux_1,\ux_2,\uy)}+\lambda
\end{align*}
where in $\text{(a)}$ we use the fact that $\sum_{\substack{\ux_1, \tx_1, \ux_2, y}}P_{X_1X_2Y}(\tx_1, \ux_2, \uy)\mathbb{1}_{\ux_1=\tx_1}\mathbb{1}_{f(\tx_1,\ux_2,\uy) = f(\ux_1,\ux_2,\uy)} =1$. Thus, $f(\tX_1,\uX_2,\uY) = f(\uX_1,\uX_2,\uY)$ under the distribution $P_{X_1X_2Y}(\tx_1, \ux_2, \uy)W_{\bX_1|X_1}(\ux_1|\tx_1)$. This proves part {\em (i)} of the lemma for the full-support case. Part {\em (ii)} can be proved similarly. 


\paragraph*{When $P_{X_1X_2Y}$ may not have full support}
For $P_{X_1 X_2 Y}$, define the set
    \begin{align*}
&\cQ_{P_{X_1X_2Y}} = \Bigg\{Q_{\uX_1 \uX_2 \uY\tX_1\tX_2}:\\
            & \qquad Q_{\uX_1 \uX_2 \uY\tX_1\tX_2} =  Q_{\uX_1 \uX_2 \uY}Q_{\tX_1|\uX_1 \uX_2 \uY}Q_{\tX_2|\uX_1 \uX_2 \uY}\\
            &\qquad \text{ satisfying} \quad Q_{\uX_1 \tX_2 \uY \uX_2}(\ux_1, \tx_2 ,\uy, \ux_2)\\
            &\qquad \qquad \qquad \qquad =P_{X_1X_2Y}(\ux_1, \tx_2, \uy)Q_{\uX_2|\tX_2}(\ux_2|\tx_2), \\
            &\qquad Q_{\tX_1 \uX_2 \uY \uX_1}(\tx_1, \ux_2, \uy, \ux_1) \\
            &\qquad \qquad \qquad =P_{X_1X_2Y}(\tx_1, \ux_2, \uy)Q_{\uX_1|\tX_1}(\ux_1|\tx_1) \text{ where }\\
            & Q_{\uX_1 \uX_2 \uY}(\ux_1, \ux_2, \uy) = \sum_{\tx_2}P_{X_1X_2Y}(\ux_1, \tx_2, \uy)Q_{\bX_2|\tX_2}(\ux_2|\tx_2) \\
            &\qquad \qquad\qquad = \sum_{\tx_1}P_{X_1X_2Y}(\tx_1, \ux_2, \uy)Q_{\uX_1|\tX_1}(\ux_1|\tx_1)\Bigg\}.
            \end{align*}

Notice that for any $(W_{\bX_1|X_1}, W_{\bX_2|X_2})$ satisfying the assumption in Lemma~\ref{lem:decoding-function-2user}, there exists a $Q_{\uX_1 \uX_2 \uY\tX_1\tX_2}\in \cQ_{P_{X_1X_2Y}}$.

For any $(\ux_1, \ux_2, \uy)$ such that $Q_{\uX_1\uX_2\uY}(\ux_1, \ux_2, \uy)>0$ for some $Q_{\uX_1 \uX_2 \uY\tX_1\tX_2}\in \cQ_{P_{X_1X_2Y}}$, we define $g(\ux_1, \ux_2, \uy) = f(\tx_1, \ux_2, \uy)$ for some $\tx_1$ such that $Q_{\tX_1|\uX_1 \uX_2 \uY}(\tx_1|\ux_1,\ux_2,\uy)>0$. By definition~\ref{def:1-viable}, this also implies that $g(\ux_1, \ux_2, \uy) = f(\tx_1, \ux_2, \uy) = f(\ux_1,\tx_2,\uy)$ for any $\tx_2$ such that $Q_{\tX_2|\uX_1 \uX_2 \uY}(\tx_2|\ux_1,\ux_2,\uy)>0$.

We will argue that the function $g$ as defined above is the same for every $Q_{\uX_1 \uX_2 \uY\tX_1\tX_2}\in \cQ_{P_{X_1X_2Y}}$ (and hence for every pair $(W_{\bX_1|X_1}, W_{\bX_2|X_2})$ satisfying the assumption in Lemma~\ref{lem:decoding-function-2user}). Suppose not, then there exists $(\ux_1, \ux_2, \uy)\in \cX_1\times\cX_2\times\cY$ such that $Q^{(1)}_{\uX_1\uX_2\uY}(\ux_1, \ux_2, \uy)>0$ and $Q^{(2)}_{\uX_1\uX_2\uY}(\ux_1, \ux_2, \uy)>0$ for some $Q^{(1)}_{\uX_1 \uX_2 \uY\tX_1\tX_2}, Q^{(2)}_{\uX_1 \uX_2 \uY\tX_1\tX_2}\in \cQ_{P_{X_1X_2Y}}$ resulting in distinct functions $g^{(1)}$ and $g^{(2)}$ such that $g^{(1)}(\ux_1, \ux_2, \uy) \neq g^{(2)}(\ux_1, \ux_2, \uy)$. This also implies that there exist $x_1, x_1'$ where $ x_1\neq x_1'$ such that $Q^{(1)}_{\tX_1|\uX_1\uX_2\uY}(x_1|\ux_1, \ux_2, \uy), Q^{(2)}_{\tX_1|\uX_1\uX_2\uY}(x_1'|\ux_1, \ux_2, \uy)>0$ and $g^{(1)}(\ux_1, \ux_2, \uy) = f(x_1, \ux_2, \uy)\neq f(x_1', \ux_2, \uy) = g^{(1)}(\ux_1, \ux_2, \uy)$.

We define $$\bar{Q}_{\uX_1 \uX_2 \uY\tX_1\tX_2} =  \bar{Q}_{\uX_1 \uX_2 \uY}\bar{Q}_{\tX_1|\uX_1 \uX_2 \uY}\bar{Q}_{\tX_2|\uX_1 \uX_2 \uY}$$ where  $$\bar{Q}_{\tX_1\uX_1 \uX_2 \uY}(\tx_1,\ux_1,\ux_2,\uy) = P_{X_1X_2Y}(\tx_1,\ux_2,\uy)\bar{Q}_{\uX_1|\tX_1}(\ux_1|\tx_1)$$  and $$\bar{Q}_{\tX_2\uX_1 \uX_2 \uY}(\tx_2,\ux_1,\ux_2,\uy) = P_{X_1X_2Y}(\ux_1,\tx_2,\uy)\bar{Q}_{\uX_2|\tX_2}(\ux_2|\tx_2)$$ for $\bar{Q}_{\uX_1|\tX_1}$ and $\bar{Q}_{\uX_2|\tX_2}$ defined as below.
\begin{align*}
\bar{Q}_{\uX_1|\tX_1} = \inp{1-\lambda}Q^{(1)}_{\uX_1|\tX_1}+\lambda Q^{(2)}_{\uX_1|\tX_1}\\
\bar{Q}_{\uX_2|\tX_2} = \inp{1-\lambda}Q^{(2)}_{\uX_2|\tX_2}+\lambda Q^{(2)}_{\uX_2|\tX_2}
\end{align*}
From the definitions of $\bar{Q}_{\uX_1 \uX_2 \uY\tX_1\tX_2}, {Q}^{(1)}_{\uX_1 \uX_2 \uY\tX_1\tX_2}$ and ${Q}^{(2)}_{\uX_1 \uX_2 \uY\tX_1\tX_2}$, it follows that  $\bar{Q}_{\uX_1 \uX_2 \uY\tX_1\tX_2}\in \cQ_{P_{X_1X_2Y}}$. As $Q^{(1)}_{\tX_1|\uX_1\uX_2\uY}(x_1|\ux_1, \ux_2, \uy)>0$ and $ Q^{(2)}_{\tX_1|\uX_1\uX_2\uY}(x_1'|\ux_1, \ux_2, \uy)>0$, we have $\bar{Q}_{\tX_1|\uX_1\uX_2\uY}(x_1|\ux_1, \ux_2, \uy),\bar{Q}_{\tX_1|\uX_1\uX_2\uY}(x_1'|\ux_1, \ux_2, \uy)>0$. Consider any $x_2$ such that $\bar{Q}_{\tX_2|\uX_1\uX_2\uY}(x_2|\ux_1, \ux_2, \uy)>0$. Then, from definition of $\bar{Q}_{\uX_1 \uX_2 \uY\tX_1\tX_2}$ and definition~\ref{def:1-viable}, we have $f(x_1, \ux_2, \uy) = f(\ux_1, x_2, \uy) = f(x_1', \ux_2, \uy)$, leading to a contradiction. Thus, $g$ is defined uniquely.  

Additionally, we can use the same argument as the one for case of full support $P_{X_1X_2Y}$ to further argue that $g(\ux_1, \ux_2, \uy)=f(\ux_1, \ux_2, \uy)$ for any $(\ux_1, \ux_2, \uy)$ such that $P_{X_1X_2Y}(\ux_1, \ux_2, \uy)>0$.

\subsubsection{Proof of Lemma~\ref{lem:2user-technical-single-letter}}   \label{sec:2user-technical-single-letter-proof}
We introduce some notation for the sake of brevity.
Define $\ttf_1: \cP(\cX_1|\cX_1) \rightarrow \cP(\cX_1 \times \cX_2 \times \cY)$ as $\ttf_1(Q_{\bX_1|X_1}) := R_{\bX_1 X_2 Y}$, where, for $\bx_1,x_1\in\cX_1$, $x_2\in\cX_2$, and $y\in\cY$, 
\begin{align*}
R_{\bX_1 X_2 Y}&(\bx_1, x_2, y) \\
&\qquad=\sum_{x_1 \in \cX_1} P_{X_1 X_2 Y}(x_1, x_2, y) Q_{\bX_1|X_1}(\bx_1|x_1).\end{align*}
For $R_{\bX_1 X_2 Y} \in \cP(\cX_1 \times \cX_2 \times \cY)$, let $\ttf_1^{-1}(R_{\bX_1 X_2 Y}) := \left\{Q_{\bX_1|X_1} \in \cP(\cX_1|\cX_1): \ttf_1(Q_{\bX_1|X_1}) = R_{\bX_1 X_2 Y}\right\}$. And for $\cQ \subseteq \cP(\cX_1 \times \cX_2 \times \cY)$, $\Phi_1^{-1}(\cQ) := \bigcup_{Q \in \cQ} \ttf^{-1}(Q)$. For $P \in \cP(\cX)$ and $\cQ \subseteq \cP(\cX)$, $\dtv(P,\cQ) := \inf_{Q \in \cQ} \dtv(P,Q)$. For $Q,Q' \in \cP(\cX|\cX)$, define $d(Q,Q') := \sum_{x \in \cX} \dtv\inp{Q(\cdot|x), Q'(\cdot|x)}$ where $Q(\cdot|x)$ is the distribution defined by $Q$ conditional on the character $x$ (and similarly for $Q'$). For a closed $\cQ \subseteq \cP(\cX|\cX)$ and $Q' \in \cP(\cX|\cX)$, define $d(Q',\cQ) = \min_{Q \in \cQ} d(Q',Q)$.

We first state and prove two helpful observations in Lemmas~\ref{lem:comp-int} and~\ref{lem:2continuous}.
    \begin{lemma}
        \label{lem:comp-int}
        There exists a function $\varepsilon: \mathbb{R}_{>0} \rightarrow \mathbb{R}_{\geq 0}$ such that
        \begin{enumerate}[label=(\alph*)]
            \item
            as $\rad \rightarrow 0$, $\varepsilon(\rad) \rightarrow 0$, and 
            \item
            for $\rad>0$, if $R \in \cV_{1}^{\rad} \cap \cV_{2}^{\rad}$, $\exists\ S \in \cV_1 \cap \cV_2$ such that $\dtv(R,S) \leq \varepsilon(\rad)$.
        \end{enumerate}
    \end{lemma}
    \begin{proof}
        We first note that $\cV_1 \cap \cV_2\neq \varnothing$ since $P_{X_1X_2Y}\in\cV_1\cap\cV_2$ (by choosing the channel $W_{\bX_1|X_1}$ as identity in the definition of $\cV_1$ and similarly for $\cV_2$). Furthermore, since both $\cV_1$ and $\cV_2$ are closed sets, $\cV_1 \cap \cV_2$ is a closed set. Hence, $\min_{S \in \cV_1 \cap \cV_2} \dtv(R,S)$ is well-defined for every $R \in \cV_{1}^{\rad} \cap \cV_{2}^{\rad}$.
        We will show that the function $\varepsilon(\rad) = \sup_{R \in \cV_{1}^{\rad} \cap \cV_{2}^{\rad}} \min_{S \in \cV_1 \cap \cV_2} d(R,S)$ satisfies both the given conditions. The fact that $\varepsilon(\rad)$ satisfies the condition (b) is obvious from its definition. Suppose it does not satisfy the condition (a), then, taking into account that $\varepsilon(\rad)$ is a non-negative, non-decreasing function of $\rad$,
        \begin{equation}\label{eq:comp-int}
            \exists\ \varepsilon_0>0 \text{ such that } \forall\ \rad>0,\ \sup_{R \in \cV_{1}^{\rad} \cap \cV_{2}^{\rad}} \dtv(R, \cV_1 \cap \cV_2) \geq \varepsilon_0.
        \end{equation}
        For $n \in \bb{N}$, setting $\rad = 1/n$ in \eqref{eq:comp-int}, we obtain a sequence $R_n$, $n \in \bb{N}$ such that, $\dtv(R_n,\cV_1)\leq 1/n$, $\dtv(R_n,\cV_2) \leq 1/n$ and $\dtv(R_n, \cV_1 \cap \cV_2) > \varepsilon_0/2$, $n \in \bb{N}$. 
        Appealing to the fact that $\cV_1$ and $\cV_2$ are closed, we define sequences $A_n \in \cV_1, n \in \bb{N}$ and $B_n \in \cV_2, n \in \bb{N}$ as follows:
        \begin{align*}
            A_n &= \arg \min_{S \in \cV_1} \dtv(R_n,S)\\
            B_n &= \arg \min_{S \in \cV_2} \dtv(R_n,S).
        \end{align*}        where $A_n$ ($B_n$, resp.) is chosen arbitrarily from among the minimizers in case more than one exists.
        Now, note that since $\cV_1$ is compact, the sequence $A_n$ has a limit point $A^{\ast} \in \cV_1$.
        Furthermore, since $\dtv(R_n,A_n) \rightarrow 0$ as $n \rightarrow \infty$, sequence $R_n$ also has $A^{\ast}$ as one of its limit points.
        This, along with the assumption that $\dtv(R_n, \cV_1 \cap \cV_2) > \varepsilon_0/2$, $\forall n \in \bb{N}$ implies that $\dtv(A^{\ast}, \cV_1 \cap \cV_2) \geq \varepsilon_0/2$.
        Now, in order to observe the contradiction, note that since $\dtv(R_n,B_n) \rightarrow 0$ as $n \rightarrow \infty$ as well, sequence $B_n$ also has $A^{\ast}$ as one of its limit points. But since $\cV_2$ is closed, $A^{\ast} \in \cV_2$ and therefore $A^{\ast} \in \cV_1 \cap \cV_2$, i.e., $\dtv(A^{\ast}, \cV_1 \cap \cV_2) = 0 < \varepsilon_0/2$.
    \end{proof}

    \begin{lemma}\label{lem:2continuous}
        Given a non-empty closed $\cV \subseteq \cV_1$, there exists $\eta_{\cV}: \bb{R}_{>0} \rightarrow \bb{R}_{\geq 0}$ such that
        \begin{enumerate}[label=(\alph*)]
            \item   as $\varepsilon \rightarrow 0$, $\eta_{\cV}(\varepsilon) \rightarrow 0$, and
            \item   for $\varepsilon>0$, if $Q\in \cP(\cX_1|\cX_1)$ and $\dtv(\ttf_1(Q), \cV) \leq \varepsilon$, then, there exists $Q' \in \ttf^{-1}_1(\cV)$ such that $d(Q,Q') \leq \eta_{\cV}(\varepsilon).$
        \end{enumerate}
    \end{lemma}
    \begin{proof}
        We will show that
        $$\eta_{\cV}(\varepsilon) = \sup_{\substack{Q \in \cP(\cX_1|\cX_1):\\ \dtv(\ttf_1(Q), \cV) \leq \varepsilon}} \min_{Q' \in \ttf_1^{-1}(\cV)}  d(Q, Q')$$
        satisfies both the conditions given in the lemma. 
        Notice that the $\min$ above is well-defined because $\ttf_1^{-1}(\cV)$ is closed (since $\ttf$ is continuous, $\cV$ is closed and the domain of $\ttf_1$ is closed) and non-empty (since $\cV \subseteq \cV_1$ is non-empty) 
        Condition (b) is implied by the definition of $\eta_{\cV}(\varepsilon)$.
        Suppose it does not satisfy condition (a), then, taking into account that $\eta_{\cV}(\varepsilon)$ is a non-negative, non-decreasing function of $\varepsilon$, 
        \begin{multline}\label{eq:2userconteq}
            \exists\ \eta_0>0 \text{ such that } \forall\ \varepsilon>0,\\ \sup_{\substack{Q \in \cP(\cX_1|\cX_1):\\ \dtv(\ttf_1(Q), \cV) \leq \varepsilon}} d(Q, \ttf_1^{-1}(\cV)) \geq \eta_0\ .
        \end{multline}
        For $n \in \bb{N}$, setting $\varepsilon = 1/n$ in \eqref{eq:2userconteq}, we obtain a sequence $Q_n, n \in \bb{N}$ such that 
        (i) $\dtv(\ttf_1(Q_n), \cV) \leq 1/n$ and 
        (ii) $\forall\ n \in \bb{N},\ d(Q_n, \ttf_1^{-1}(\cV)) \geq \eta_0/2$.
        Since $\cP(\cX_1|\cX_1)$ is compact, $Q_n$ has a limit point $Q^{\ast} \in \cP(\cX_1|\cX_1)$.
        From (ii), we have that $d(Q^{\ast}, \ttf_1^{-1}(\cV)) \geq \eta_0/2$ and therefore $Q^{\ast} \not\in \ttf_1^{-1}(\cV)$,i.e., $\ttf_1(Q^{\ast}) \not\in \cV$.
        Now, since $\ttf_1$ is a continuous map, $\ttf_1(Q^{\ast})$ should be a limit point of $\ttf_1(Q_n)$.
        From (i), we know that $\dtv(\ttf_1(Q_n), \cV) \rightarrow 0$ as $n \rightarrow \infty$, therefore, all limit points of $\ttf_1(Q_n)$ must lie in $\cV$, and therefore $\ttf_1(Q^{\ast}) \in \cV$, which contradicts the fact that $\ttf_1(Q^{\ast}) \not\in \cV$.
    \end{proof}

Next we prove a lemma which will imply Lemma~\ref{lem:2user-technical-single-letter} as we argue further below.
    \begin{lemma}\label{lem:old-2user-technical-single-letter}
        For a 1-viable $f$ along with $g$ as in Lemma~\ref{lem:decoding-function-2user}, there is a $\gamma':\bbR^+\to\bbR^+$ such that $\gamma'(\delta)\to0$ as $\delta\downarrow0$ satisfying the following: for $\delta>0$, if $W_{\bX_1|X_1}$ is such that the p.m.f. 
        \\$Q_{\uX_1\uX_2\uY}(\ux_1,\ux_2,\uy):=\sum_{x_1}P_{X_1X_2Y}(x_1,\ux_2,\uy)W_{\bX_1|X_1}(\ux_1|x_1)$\\
        belongs to $\cV_1^{\delta}\cap\cV_2^{\delta}$, then $\pr(f(X_1,X_2,Y)\neq g(\bX_1,X_2,Y))\leq \gamma'(\delta)$ under the p.m.f.  $P_{X_1X_2Y}W_{\bX_1|X_1}$.
    \end{lemma}
    \begin{proof}

        We will show that $\gamma'(\delta) := |\cX_1|\eta_{\cV_1 \cap \cV_2}(\varepsilon(\delta))$ satisfies the requirements, where $\varepsilon$ and $\eta_{\cV_1 \cap \cV_2}$ functions are obtained from Lemma~\ref{lem:comp-int} and Lemma~\ref{lem:2continuous} respectively (recall that $\cV_1\cap\cV_2$ is non-empty and closed as we noted at the beginning of proof of Lemma~\ref{lem:comp-int}). Clearly, $\gamma'(\delta) \rightarrow 0$ as $\delta \rightarrow 0$. Now, consider any $W_{\bX_1|X_1}$ such that $\sum_{x_1} P_{X_1 X_2 Y}(x_1,x_2,y) W_{\bX_1|X_1}(\bx_1|x_1) = R_{\bX_1 X_2 Y}(\bx_1,x_2,y)$ and $R_{\bX_1 X_2 Y} \in \cV^{\delta}_1 \cap \cV^{\delta}_2$, i.e., $\ttf_1(W_{\bX|X}) \in \cV_1^{\delta} \cap \cV_2^{\delta}$.    
        
        By Lemma~\ref{lem:comp-int},
        \begin{equation} \label{eq:2user-eps}
            \dtv\inp{\ttf_1(W_{\bX_1|X_1}), \cV_1 \cap \cV_2} \leq \varepsilon(\delta).
        \end{equation}
        Hence, by Lemma~\ref{lem:2continuous},
        \begin{equation}\label{eq:2user-eta}
            d\left(W_{\bX_1|X_1}, C_{\bX_1|X_1}\right) \leq \eta_{\cV_1 \cap \cV_2}(\varepsilon(\rad))
        \end{equation}
        for some $C_{\bX_1|X_1} \in \ttf^{-1}_1(\cV_1 \cap \cV_2)$. 
        Now, recall that 
        \begin{multline*}
            \cV_2 = \{ Q_{\uX_1 \uX_2 \uY} : \exists\ W_{\bX_2|X_2} \text{ s.t. } Q_{\uX_1 \uX_2 \uY}(\ux_1,\ux_2,\uy) \\= \sum_{x_2} P_{X_1 X_2 Y}(\ux_1, x_2, \uy) W_{\bX_2|X_2}(\ux_2|x_2), \forall\ \ux_1,\ux_2,\uy \}.
        \end{multline*}
        Since $\ttf_1\inp{C_{\bX_1|X_1}} \in \cV_1 \cap \cV_2$, there exist some $C_{\bX_2|X_2}$ such that 
        \begin{multline*}
            \sum_{x_2} P_{X_1 X_2 Y}(\ux_1, x_2, \uy) C_{\bX_2|X_2}(\ux_2|x_2)\\
            = \sum_{x_1} P_{X_1 X_2 Y}(x_1, \ux_2, \uy) C_{\bX_1|X_1}(\ux_1|x_1)
        \end{multline*}
        Consider the function $g$ from Lemma~\ref{lem:decoding-function-2user} (recall that $f$ is 1-viable). We have,
        \begin{align}
            \sum_{x_1,x_2,z,\bx_1} P_{X_1 X_2 Z}&(x_1,x_2,z) C_{\bX_1|X_1}(\bx_1|x_1)\notag\\ &\bb{1} \left\{f(x_1,x_2,z) \neq g(\bx_1,x_2,z)\right\} = 0, \label{eq:old-tech-lemma1}
        \end{align}
        using which, we get that under $P_{X_1 X_2 Y}W_{\bX_1|X_1}$,
        \begin{align*}
            \MoveEqLeft[2] \bbP[f(X_1,X_2,Y) \neq g(\bX_1,X_2,Y)]\\
            =& \sum_{x_1,\bx_1,x_2,y} P_{X_1 X_2 Y}(x_1,x_2,y) W_{\bX_1|X_1} \\ 
            &\qquad\qquad \times \bb{1} \left\{ f(x_1,x_2,y) \neq g(\bx_1,x_2,y)\right\}\\
            =& \sum_{x_1,\bx_1,x_2, z} P_{X_1 X_2 Y}(x_1,x_2,y) \left(W_{\bX_1|X_1}(\bx_1|x_1) \right.\\ 
            &\qquad\qquad\left.- C_{\bX_1|X_1}(\bx_1|x_1) + C_{\bX_1|X_1}(\bx_1|x_1) \right) \\
            &\qquad\qquad\qquad\times \bb{1}\left\{ f(x_1,x_2,y) \neq g(\bx_1,x_2,y)\right\}\\
            \stackrel{\text{(a)}}{=}& \sum_{x_1,\bx_1,x_2,y} P_{X_1 X_2 Y}(x_1,x_2,y) \left(W_{\bX_1|X_1}(\bx_1|x_1) - \right.\\ &\qquad\quad\left. C_{\bX_1|X_1}(\bx_1|x_1)\right) \bb{1}\left\{ f(x_1,x_2,y) \neq g(\bx_1,x_2,y)\right\}\\
            \leq& \sum_{x_1,\bx_1,x_2,y} P_{X_1 X_2 Y}(x_1,x_2,y)\\&\qquad\qquad\quad\left[W_{\bX_1|X_1}(\bx_1|x_1) - C_{\bX_1|X_1}(\bx_1|x_1)\right]_+\\
            =& \sum_{x_1,\bx_1} P_{X_1}(x_1)\left[W_{\bX_1|X_1}(\bx_1|x_1) - C_{\bX_1|X_1}(\bx_1|x_1)\right]_+\\
            \leq& \sum_{x_1,\bx_1} \left[W_{\bX_1|X_1}(\bx_1|x_1) - C_{\bX_1|X_1}(\bx_1|x_1)\right]_+\\
            =& \sum_{x_1} \sum_{\bx_1} \left[W_{\bX_1|X_1}(\bx_1|x_1) - C_{\bX_1|X_1}(\bx_1|x_1)\right]_+\\
            =& \sum_{x_1} d(W_{\bX_1|X_1}, C_{\bX_1|X_1})\\
            \leq& |\cX_1| \eta_{\cV_1\cap\cV_2}(\varepsilon(\rad))=\gamma'(\delta),
        \end{align*}
        where (a) follows from \eqref{eq:old-tech-lemma1} and the last inequality is due to \eqref{eq:2user-eta}. This completes the proof.
    \end{proof}

Finally we argue that Lemma~\ref{lem:old-2user-technical-single-letter} implies Lemma~\ref{lem:2user-technical-single-letter} (which is restated below)\\
    \textbf{Lemma~\ref{lem:2user-technical-single-letter} (restated):} 
        For a 1-viable $f$ along with $g$ as in Lemma~\ref{lem:decoding-function-2user}, 
        there is a 
        $\gamma:\bbR^+\to\bbR^+$ such that $\gamma(\delta)\to0$ as $\delta\to0$ 
        satisfying the following: for $\delta>0$, if 
        $Q_{X_1\uX_2\uY\uX_1}$ is such that $Q_{\uX_1\uX_2\uY}\in\cV_1^{\delta}\cap\cV_2^{\delta}$ and $\exists$ $W_{\bX_1|X_1}$ for which $\dtv(Q_{X_1\uX_2\uY\uX_1}, P_{X_1X_2Y}W_{\bX_1|X_1})\leq \delta$, then $\pr(f(X_1,\uX_2,\uY)\neq g(\uX_1,\uX_2,\uY))\leq \gamma(\delta)$ under the p.m.f. $Q_{X_1\uX_2\uY\uX_1}$. 
    \begin{proof}[Proof of Lemma~\ref{lem:2user-technical-single-letter}]
Define $Q'_{X_1'\uX_2'\uY'\uX_1'}$ as
\begin{align}
    Q'_{X_1'\uX_2'\uY'\uX_1'}&(x_1,\ux_2,\uy,\ux_1)\notag\\
    &:=P_{X_1X_2Y}(x_1,\ux_2,\uy)W_{\bX_1|X_1}(\ux_1|x_1), \label{eq:new-2user-tech1}
\end{align}
$x_1,\ux_1\in\cX_1$, $\ux_2\in\cX_2$, $\uy\in\cY$. We first note that 
$\dtv(Q_{X_1\uX_2\uY\uX_1},Q'_{X_1'\uX_2'\uY'\uX_1'})\leq \delta$ implies  $\dtv(Q_{\uX_1\uX_2\uY},Q'_{\uX_1'\uX_2'\uY'})\leq \delta$. Also, since $Q_{\uX_1\uX_2\uY}\in \cV_1^{\delta}$, there is a $Q''\in\cV_1$ such that $\dtv(Q_{\uX_1\uX_2\uY},Q'')\leq\delta$. Hence, $\dtv(Q'_{\uX_1'\uX_2'\uY'},Q'')\leq 2\delta$, i.e., $Q'_{\uX_1'\uX_2'\uY'}\in\cV_1^{2\delta}$.
Similarly, since $Q_{\uX_1\uX_2\uY}\in \cV_2^{\delta}$, we have $Q'_{\uX_1'\uX_2'\uY'}\in\cV_2^{2\delta}$. Hence, $Q'_{\uX_1'\uX_2'\uY'}\in\cV_1^{2\delta}\cap\cV_2^{2\delta}$.
Moreover, since $Q'_{X_1'\uX_2'\uY'\uX_1'}$ is of the form in \eqref{eq:new-2user-tech1}, by Lemma~\ref{lem:old-2user-technical-single-letter}, 
\[
\pr(f(X_1', \uX_2', \uY')\neq g(\uX_1', \uX_2',\uY'))\leq \gamma'(2\delta)
\]
under $Q'_{X_1'\uX_2'\uY'\uX_1'}$. Therefore, under $Q_{X_1\uX_2\uY\uX_1}$,
\begin{align*}
&\pr\inp{f(X_1, \uX_2, \uY)\neq g(\uX_1, \uX_2, \uY)} \\
&= \sum_{x_1, \ux_1, \ux_2, \uy}Q_{X_1\uX_2\uY\uX_1}(x_1, \ux_2, \uy, \ux_1)\mathbb{1}_{{f(x_1, \ux_2, \uy)\neq g(\ux_1, \ux_2, \uy)}}\\
& \leq \sum_{x_1, \ux_1, \ux_2, \uy}Q'_{X_1\uX_2\uY\uX_1}(x_1, \ux_2, \uy, \ux_1)\mathbb{1}_{{f(x_1, \ux_2, \uy)\neq g(\ux_1, \ux_2, \uy)}}\\
 &\quad\quad + \dtv(Q_{X_1\uX_2\uY\uX_1},Q'_{X_1'\uX_2'\uY'\uX_1'})\\
        &\leq \gamma'(2\delta) + \delta=:\gamma(\delta).
        \end{align*} 
        This concludes the proof.%
    \end{proof}

\subsubsection{Proof of Lemma~\ref{lem:2user-typicality}}\label{sec:2user-typicality-proof}
For $\delta>0$, define
\begin{align*}
\cQ_n(\delta):=\{&Q_{X_1'X_2'Y'\bX_1'}\in\cP_n : \\
 &\qquad D(Q_{X_1'X_2'Y'\bX_1'}||P_{X_1X_2Y}Q_{\bX_1'|X_1'})\leq 2{\delta}^2\}.
\end{align*}
We first show that 
\begin{align}
\pr\inp{\type{X_1^n,X_2^n,Y^n,\bX_1^n}\in \cQ_n(\delta)}\geq 1-2^{-\Omega(n)}.\label{eq:lemma3parta}
\end{align} This will imply that, with high probability, there exists a conditional distribution $Q_{\bX_1'|X_1'}$ such that
\begin{align*}
2{\delta}^2&\geq D(\type{X_1^n,X_2^n,Y^n,\bX_1^n}\|P_{X_1X_2Y}Q_{\bX_1'|X_1'})\\
&\geq D\inp{\type{X_2^n,Y^n,\bX_1^n}\middle\|\sum_{x_1}P_{X_1X_2Y}(x_1,\cdot,\cdot)Q_{\bX_1'|X_1'}(\cdot|x_1)}\\
&\stackrel{\text{(a)}}{\geq} 2\dtv^2\inp{\type{X_2^n,Y^n,\bX_1^n},\sum_{x_1}P_{X_1X_2Y}(x_1,\cdot,\cdot)Q_{\bX_1'|X_1'}(\cdot|x_1)},
\end{align*} where (a) follows from Pinsker's inequality. This will prove part {\em (i)} of the lemma. We now  show \eqref{eq:lemma3parta} using a symmetrization trick which was also used in~\cite{narayanan2023complete}.
\begin{align*}
&\pr\inp{\type{X_1^n,X_2^n,Y^n,\bX_1^n}\notin \cQ_n(\delta)}\\
&=\sum_{\substack{{x_1^n, x_2^n, y^n, \bx_1^n}:\\\type{x_1^n, x_2^n, y^n, \bx_1^n}\in \cQ_n(\delta)}} P^n_{X_1 X_2 Y}(x_1^n,x_2^n,y^n)  W_1(\bx_1^n|x_1^n)\\
& \stackrel{\text{(a)}}{=} \sum_{\pi\in\Pi_n}\frac{1}{|\Pi_n|}\sum_{\substack{{x_1^n, x_2^n, y^n, \bx_1^n}:\\\type{x_1^n, x_2^n, y^n, \bx_1^n}\notin \cQ_n(\delta)}}
\\&\qquad\quad
P^n_{X_1 X_2 Y}(\pi\inp{x_1^n},\pi\inp{x_2^n},\pi\inp{y^n}) 
W_1(\pi\inp{\bx_1^n}|\pi\inp{x_2^n})\\
& \stackrel{\text{(b)}}{=} \sum_{\substack{{x_1^n, x_2^n, y^n, \bx_1^n}:\\\type{x_1^n, x_2^n, y^n, \bx_1^n}\notin \cQ_n(\delta)}}P^n_{X_1 X_2 Y}({x_1^n},{x_2^n},{y^n})\\
&\qquad \qquad \inp{\sum_{\pi\in\Pi_n}\frac{1}{|\Pi_n|} W_1(\pi\inp{\bx_1^n}|\pi\inp{x_2^n})}\\
& \stackrel{\text{(c)}}{=} \sum_{\substack{{x_1^n, x_2^n, y^n, \bx_1^n}:\\\type{x_1^n, x_2^n, y^n, \bx_1^n}\notin \cQ_n(\delta)}}P^n_{X_1 X_2 Y}({x_1^n},{x_2^n},{y^n})W_{\textsf{sym}}({\bx_1^n}|{x_1^n})\\
&\stackrel{\text{(d)}}{\leq} \sum_{\substack{Q_{X_1'X_2'Y'} \in \cP_n
:\\D(Q_{X_1'X_2'Y'}||P_{X_1X_2Y})> \delta^2}}\sum_{\substack{(x_1^n, x_2^n, y^n)\\\in \cT^n_{X_1'X_2'Y'}}}P^n_{X_1 X_2 Y}(x_1^n,x_2^n,y^n)\\
    &\qquad + \sum_{\substack{Q_{X_1'X_2'Y\bX_1'} \in \cP_n:
    \\
    I(X_2'Y';\bX_1'|X_1')> \delta^2}}\sum_{\substack{(x_1^n, x_2^n, y^n)\\\in \cT^n_{X_1'X_2'Y'}}}P^n_{X_1 X_2 Y}(x_1^n,x_2^n,y^n)\\
    &\qquad\qquad \qquad\qquad \qquad\qquad \sum_{\bx_1^n\in \cT^n_{\bX_1'|X_1'X_2'Y'}}W_{\textsf{sym}}(\bx^n_1|x^n_1)\\
    & \stackrel{\text{(e)}}{=}: A + B,
\end{align*} where, in (a), $\Pi_n$ denotes the set of all permutations of $(1,2, \ldots, n)$ and for $\pi\in \Pi_n$, $\pi(x^n)=\pi\inp{x_1, x_2, \ldots, x_n}$ denotes $\inp{x_{\pi(1)}, x_{\pi(2)},\ldots,x_{\pi(n)}}$; 
(b) follows by noticing that $P^n_{X_1 X_2 Y}$ is invariant under permutations; in (c), we define $W_{\textsf{sym}}({\bx_1^n}|{x_1^n}) := \inp{\sum_{\pi\in \Pi_n}\frac{1}{|\Pi_n|}W_1(\pi\inp{\bx_1^n}|\pi\inp{x_1^n})}$, which we note is a channel since summing over all $\bx_1^n$ for a fixed $x_1^n$ yields 1; and (d) follows by noticing that $Q_{X_1'X_2'Y'\bX_1'}\notin \cQ_n(\delta)$ only if either $D(Q_{X_1'X_2'Y'}||P_{X_1X_2Y})> \delta^2$ or $I(X_2'Y';\bX_1'|X_1')> \delta^2$ (or both). This is because $D(Q_{X_1'X_2'Y'}||P_{X_1X_2Y}) + I(X_2'Y';\bX_1'|X_1') = D(Q_{X_1'X_2'Y'\bX_1'}||P_{X_1X_2Y}Q_{\bX_1'|X_1'})$. We also used the fact that $W_{\textsf{sym}}$ is a channel in step (d) in writing the upper bound in the first term. In step~(e), we define the two terms in the previous step as $A$ and $B$, respectively. We first analyze $A$.
\begin{align*}
&\sum_{\substack{Q_{X_1'X_2'Y'} \in \cP_n
:\\D(Q_{X_1'X_2'Y'}||P_{X_1X_2Y})> \delta^2}}\sum_{\substack{(x_1^n, x_2^n, y^n)\\\in \cT^n_{X_1'X_2'Y'}}}P^n_{X_1 X_2 Y}(x_1^n,x_2^n,y^n)\\
&\stackrel{\text{(a)}}{\leq} \sum_{\substack{Q_{X_1'X_2'Y'} \in \cP_n
:\\D(Q_{X_1'X_2'Y'}||P_{X_1X_2Y})> \delta^2}}\exp{\inp{-nD\inp{Q_{X_1'X_2'Y'}||P_{X_1X_2Y}}}}\\
&\stackrel{\text{(b)}}{\leq} (n+1)^{|\cX_1||\cX_2||\cY|}\exp\inp{-n\delta^2}\\
&= 2^{-\Omega(n)},
\end{align*} where (a) follows from \eqref{eq:type1} and (b) follows from \eqref{eq:type4}. Now, we analyse $B$.
\begin{align*}
&\sum_{\substack{Q_{X_1'X_2'Y'\bX_1'}
\in \cP_n:
\\I(X_2'Y';\bX_1'|X_1')> \delta^2}}\sum_{\substack{(x_1^n, x_2^n, y^n)\\\in \cT^n_{X_1'X_2'Y'}}}P^n_{X_1 X_2 Y}(x_1^n,x_2^n,y^n)\\
    &\qquad\qquad \qquad\quad \sum_{\bx_1^n\in \cT^n_{\bX_1'|X_1'X_2'Y'}(x_1^n, x_2^n, y^n)}W_{\textsf{sym}}(\bx^n_1|x^n_1)\\
&\stackrel{\text{(a)}}{\leq}    \sum_{\substack{Q_{X_1'X_2'Y'\bX_1'} 
\in \cP_n:
\\I(X_2'Y';\bX_1'|X_1')> \delta^2}}\sum_{\substack{(x_1^n, x_2^n, y^n)\\\in \cT^n_{X_1'X_2'Y'}}}P^n_{X_1 X_2 Y}(x_1^n,x_2^n,y^n)\\
    &\quad\exp\inp{nH(\bX_1'|X_1'X_2'Y')} (n+1)^{|\cX_1|^2}\exp\inp{-nH(\bX_1'|X_1')}\\
&\stackrel{\text{(b)}}{\leq}\sum_{\substack{Q_{X_1'X_2'Y'\bX_1'} 
\in \cP_n:
\\I(X_2'Y';\bX_1'|X_1')> \delta^2}}(n+1)^{|\cX_1|^2}\exp\inp{-n\delta^2}\\
    &\leq (n+1)^{|\cX_1||\cX_2||\cY||\cX_1|}(n+1)^{|\cX_1|^2}\exp\inp{-n\delta^2}\\
    & = 2^{-\Omega(n)},
\end{align*} where (a) follows by noting that $W_{\textsf{sym}}(\bx^n_1|x^n_1)$ is the same for all $\bx^n_1\in \cT^n_{\bX_1|X_1}(x^n_1)$ and hence, by using $\sum_{\bx^n_1\in \cT^n_{\bX_1'|X_1'}(x^n_1)}W_{\textsf{sym}}(\bx^n_1|x^n_1)\leq 1$, we have  $W_{\textsf{sym}}(\bx^n_1|x^n_1)\leq \frac{1}{|\cT^n_{\bX_1'|X_1'}(x^n_1)|}\leq (n+1)^{|\cX_1|^2}\exp\inp{-H(\bX_1'|X_1')}$ from \eqref{eq:type3}; moreover, the size of $\cT^n_{\bX_1'|X_1'X_2'Y'}(x_1^n, x_2^n, y^n)$ is at most $\exp\inp{nH(\bX_1'|X_1'X_2'Y')}$ by \eqref{eq:type3}. The inequality (b) follows by noticing that  $H(\bX_1'|X_1') - H(\bX_1'|X_1'X_2'Y') = I(\bX_1';X_2'Y'|X_1')> \delta^2$. Thus, we have shown~\eqref{eq:lemma3parta}.

We show part~(ii) of the lemma now.
\begin{align*}
&\pr\left((\type{\bX_1^n X_2^n Y^n} \in \cV_1^{\delta} \cap \cV_2^{\delta}) \cap \right.\\
&\quad \left. (\ham(f(X_1^n,X_2^n,Y^n),g(\bX_1^n,X_2^n,Y^n))>\gamma(\delta))\right)\\
&\stackrel{\text{(a)}}{=} \pr\left((\type{\bX_1^n X_2^n Y^n} \in \cV_1^{\delta} \cap \cV_2^{\delta}) \cap \right.\\ 
&\qquad\left.(\ham(f(X_1^n,X_2^n,Y^n),g(\bX_1^n,X_2^n,Y^n))>\gamma(\delta))\cap\right.\\
&\qquad\left.(\type{X_1^n X_2^n Y^n\bX_1^n}\in \cQ_n(\delta))\right)
+ 2^{-\Omega(n)}\\
&\stackrel{\text{(b)}}{=}: C + 2^{-\Omega(n)},
\end{align*} where $(a)$ follows from \eqref{eq:lemma3parta} and in (b) we define $C$ as the first term. We will show that $C = 0$.
This will be a consequence of Lemma~\ref{lem:2user-technical-single-letter} (restated below). 
\begin{lemma*}
For a 1-viable $f$ along with $g$ as in Lemma~\ref{lem:decoding-function-2user}, there is a ${\gamma}:\bbR^+\to\bbR^+$ such that $\lim_{\delta\to 0} {\gamma}(\delta)=0$
satisfying the following: if 
$Q_{X_1\uX_2\uY,\uX_1}$ is such that $Q_{\uX_1\uX_2\uY}$ belongs to $\cV_1^{\delta}\cap\cV_2^{\delta}$ and there is a $W_{\bX_1|X_1}$ for which $\dtv(Q_{X_1\uX_2\uY,\uX_1}, P_{X_1X_2Y}W_{\bX_1|X_1})\leq \delta$, then $\pr(f(X_1,\uX_2,\uY)\neq g(\uX_1,\uX_2,\uY))\leq {\gamma}(\delta)$ under the p.m.f. $Q_{X_1\uX_2\uY,\uX_1}$. 
\end{lemma*}
 For any joint type $Q_{X_1'X_2'Y'\bX_1'}\in \cQ_n\inp{\delta}$, we have $D(Q_{X_1'X_2'Y'\bX_1'}||P_{X_1X_2Y}Q_{\bX_1'|X_1'})\leq 2{\delta}^2$ and hence $d_{TV}(Q_{X_1'X_2'Y'\bX_1'},P_{X_1X_2Y}Q_{\bX_1'|X_1'})\leq \delta$ (using Pinsker's inequality). We conclude from Lemma~\ref{lem:2user-technical-single-letter} that 
\begin{align*}
{\gamma}&( \delta)\geq \bbP\inp{f(X_1', X_2', Y')\neq g(\bX_1', X_2', Y')}\\
&=\sum_{x_1,  x_2, y, \bx_1}Q_{X_1'X_2'Y'\bX_1'}(x_1, x_2, y, \bx_1)\mathbb{1}_{\inp{f(x_1, x_2, y)\neq g(\bx_1, x_2, y)}}.
\end{align*}
Since, for any $(x_1^n,x_2^n,y^n,\bx_1^n)$, the average distortion $\ham(f(x_1^n,x_2^n,y^n),g(\bx_1^n,x_2^n,y^n)$ coincides with $\pr(f(X_1,\uX_2,\uY)\neq g(\uX_1,\uX_2,\uY))$ under the p.m.f. $Q_{X_1\uX_2\uY\uX_1}:=\type{x_1^n,x_2^n,y^n,\bx_1^n}$, 
We conclude that for all sequences $(x_1^n, x_2^n, y^n, \bx_1^n)$ of joint type $Q_{X_1'X_2'Y'\bX_1'}\in \cQ_n(\delta)$ such that $Q_{\bX_1'X_2'Y'}\in\cV_1^{\delta}\cap\cV_2^{\delta}$, we have $\ham(f(x_1^n,x_2^n,y^n),g(\bx_1^n,x_2^n,y^n)\leq \gamma(\delta)$. Hence, $C= 0$.

\subsubsection{Achievability proof of Theorem~\ref{thm:2user}}

The achievability of Theorem~\ref{thm:2user} follows from Lemma~\ref{lem:2user-typicality} (and its analog for user~2) by the discussion following the statement of that lemma in page~\pageref{lem:2user-typicality}.

%% file: appendix.tex
\input{appendix-k-theorem}
\input{appendix-LP}
\input{appendix-upgrades}
\input{appendix-example}

%% file: appendix-k-theorem.tex
\section{General Adversary Structure}\label{app:general-adversary-structure}
    Let $k \in \mathbb{N}$ and $\cX_1, \dots, \cX_{k}, \cY$ be finite alphabets.
    Consider a joint distribution $P_{X_{[k]},Y} := P_{X_1 \dots X_k Y}$ defined on $\cX_1 \times \dots \times \cX_k \times \cY$.
    Let $(X_{1,t}, \dots, X_{k,t}, Y_t) \sim P_{X_{[k]},Y}$, $t \in [n]$ be independent and identically distributed (i.i.d).

    \begin{definition}\label{def:adversary-structure} 
     We call a collection $\coll = \{\cA_1, \cA_2 ,\dots, \cA_m\} \subseteq 2^{[k]}$ of subsets of $[k]$ such that ${\varnothing} \in \coll$ an {\em adversary structure}. The subsets $A\in\coll$ are called {\em adversary sets}.
    \end{definition}
    
    The \emph{$(k,\coll)$-byzantine distributed source coding} problem for an adversarial structure $\coll = \{\cA_1, \cA_2 ,\dots, \cA_m\} \subseteq 2^{[k]}$, where ${\varnothing} \in \coll$, is defined as follows: $\pa_1, \dots, \pa_{k}$ are nodes connected to a decoder via noiseless channels and a set of nodes $\cA_i \in \coll$ is controlled by the adversary. For each $i \in [k]$, node $\pa_{i}$ observes $X_i^{n}:=(X_{i,1}, \dots, X_{i,n})$ and the decoder observes $Y^n:=(Y_1, \dots, Y_n)$.
    The decoder is interested in recovering a function $f$ of the observations (i.e., the domain of $f$ is $\cX_1 \times \dots \times \cX_k \times \cY$). Specifically, if $Z_t = f(X_{1,t}, \dots, X_{k,t},Y_t), t \in [n]$, the decoder desires to recover $Z^n$ (with a vanishing average Hamming distortion) or correctly identify {\em one} of the users controlled by the adversary. 
    The nodes are required to send their observations to the decoder. Let the decoder be $\phi:\cX_1^n\times\ldots\times\cX_k^n\times\cY^n\to [k]\cup\cZ^n$, where $\cZ$ is the co-domain of $f$.
    
    When an adversary controls a set of nodes $\cA \in \coll$, it generates $\bX^n_{\cA}$ by using a channel $W_{\cA}$ (both $\cA$ and $W_{\cA}$ are unknown to the decoder, other than the fact that $\cA \in \coll$). Denoting $\bcA = [k] \setminus \cA$, for $\gamma>0$, the error event when the adversary controls the set $\cA$ is $\cE(\gamma,\cA)=\cE_1(\cA) \cup \cE_2(\gamma,\cA)$ where
    \begin{align*}
    \cE_1(\cA) &= \inp{\phi\inp{\ina{\bX^n_{\cA},X^n_{\bcA}},Y^n} \in \bcA},
    \text{ and}\\ 
    \cE_2(\gamma,\cA) &= \inp{\phi\inp{\ina{\bX^n_{\cA},X^n_{\bcA}},Y^n} \not\in [k]} \bigcap\\ &\qquad\inp{\ham\inp{\phi\inp{\ina{\bX^n_{\cA},X^n_{\bcA}},Y^n},Z^n}>\gamma}.
    \end{align*}
    Notice that for $\cA=\varnothing$ (i.e., when the adversary is absent; recall that we always have $\varnothing\in\coll$), an error occurs unless the decoder outputs an estimate $\hZ^n$ and it is of average Hamming distortion no larger than $\gamma$.
    Denote the probability of the error event by $\eta(\gamma, \cA, W_{\cA}) = \pr(\cE(\gamma,\cA))$, where the probability is evaluated under the joint distribution 
    $$P(x^n_{[k]},y^n,\bx^n_{\cA}) = \inp{\prod_{t=1}^n P_{X_{[k]}, Y}(x_{[k],t},y_t)} W_{\cA}(\bx^n_{\cA}|x^n_{\cA}).$$
    For $\gamma>0$, the error probability of decoder $\phi$ is defined as
    $$\epsilon(\gamma,\phi) = \max_{\cA \in \coll} \sup_{W_{\cA}} \eta(\gamma,\cA,W_{\cA}).$$
    \begin{definition}
        For a distribution $P_{X_{[k]},Y}$, $f:\cX_1 \times \dots \times \cX_{k} \times \cY \rightarrow \cZ$ is \emph{$\coll$-robustly recoverable} if for all $\gamma>0$, there exists a sequence of decoders $\phi_n, n \in \mathbb{N}$, such that $\lim \inf_{n \rightarrow \infty} \epsilon(\gamma,\phi_n) =0$.
    \end{definition}

    \begin{definition}\label{def:non-intersecting}
         For an adversary structure $\coll = \{\cA_1, \cA_2, \dots, \cA_m\}$, we say that a non-empty subset $\cI \subseteq [m]$ (corresponding to adversary sets $\cA_i,i\in\cI$) is {\em non-intersecting} if $\bigcap_{i \in \cI} \cA_i = \varnothing$. We denote the set of all such non-intersecting subsets by $\mathbb{I}$.
    \end{definition}

    We now define the class of functions which are analogous to Definition~\ref{def:s-viable} of $s$-viable functions for the  $(k,s)$-byzantine distributed source coding problem.
    \begin{definition}\label{lem:A-viable}
        We say a function $f$ with domain $\cX_1 \times \cX_2 \times \dots \times \cX_k \times \cY$ is \emph{$\coll$-viable} 
        if for a given collection $\coll = \{\cA_1 \dots, \cA_m\}$, for any non-intersecting (and hence also non-empty; see Definition~\ref{def:non-intersecting} above) $\cI \subseteq [m]$, under every joint p.m.f. $Q_{\uX_{[k]}, \uY, \inp{\tX^i_{\cA_i}}_{i \in \cI}}$ over $\cX_1 \times \dots \times \cX_k \times \cY \times \prod_{i \in \cI} \prod_{j \in \cA_i} \cX_j$ 
        satisfying, for each $i \in \cI$, 
        \begin{enumerate}
            \item[(a)]   $Q_{\ina{\tX^i_{\cA_i} \uX_{\bcA_i}},Y} = P_{X_{[k]},Y}$, and
            \item[(b)]   $\uX_{\cA_i} \mc \tX^i_{\cA_i} \mc \inp{\uX_{\bcA_i},Y}$, 
        \end{enumerate}
        we have for all $i,i' \in \cI$ (with probability $1$), 
        $$f \inp{\ina{\tX_{\cA_i}^i, \uX_{\bcA_i}},Y} = f \inp{\ina{\tX_{\cA_{i'}}^{i'}, \uX_{\bcA_{i'}}},\uY}.$$
    \end{definition}

We have the following characterization which we prove in the following subsections.
    \begin{theorem}\label{thm:kuser-main}
        For $P_{X_{[k]}Y}$ and an adversarial structure $\coll$, $f:\cX_1 \times \dots \times \cX_k \times \cY \rightarrow \cZ$ is $\coll-$robustly recoverable if and only if $f$ is $\coll$-viable.
    \end{theorem}

    \subsection{Converse}
    \textbf{Converse part of Theorem~\ref{thm:kuser-main} (restated):} For a distribution $P_{X_{[k]}Y}$ and $\coll = \{\cA_1, \dots, \cA_m\}$,  if a function $f$ is $\coll$-robustly recoverable, then for any non-intersecting $\cI \subseteq [m]$ and for any distribution $Q_{\uX_{[k]},\uY, \inp{\tX^i_{\cA_i}}_{i \in \cI}}$ 
    satisfying, for each $i \in \cI$, 
        \begin{enumerate}
            \item[(a)]   $Q_{\ina{\tX^i_{\cA_i} \uX_{\bcA_i}}Y} = P_{X_{[k]}Y}$ and
            \item[(b)]   $\uX_{\cA_i} \mc \tX^i_{\cA_i} \mc \inp{\uX_{\bcA_i},Y}$,
        \end{enumerate}
        we have for all $i,i' \in \cI$ (with probability $1$), 
        $$f \inp{\ina{\tX_{\cA_i}^i, \uX_{\bcA_i}},Y} = f \inp{\ina{\tX_{\cA_{i'}}^{i'}, \uX_{\bcA_{i'}}},\uY}.$$

    \begin{proof}
        Fix any non-intersecting $\cI \subseteq [m]$. For any $Q_{\uX_{[k]},\uY, \inp{\tX^i_{\cA_i}}_{i \in \cI}}$ satisfying conditions (a) and (b) for every $i\in \cI$, the observations reported to the decoder and its side information are jointly distributed as $Q_{\uX_{[k]},\uY}$ i.i.d. under $|\cI|$ different scenarios, each corresponding to one of the adversarial sets $\cA_i$, $i \in \cI$. In more detail, for $i \in \cI$, scenario $i$ is realized as follows: Adversary controls the set $\cA_i$. The underlying observations are $(\langle {\tX^{i,n}_{\cA_i},\uX_{\bcA_i}^n}\rangle, \uY^n)$ distributed as $Q_{\ina{\tX^i_{\cA_i}, \uX_{\bcA_i}},\uY} = P_{X_{[k]},Y}$ i.i.d. Users in the set $\bcA_i$ report their observations $\uX^n_{\bcA_i}$ honestly and the side information at the decoder is $\uY^n$. The adversary produces the report $\uX_{\cA_i}^n$ (for the nodes in $\cA_i$) by passing $\tX^{i,n}_{\cA_i}$ through the DMC $Q_{\uX_{\cA_i}|\tX^i_{\cA_i}}$ and therefore the reported variables follow the Markov chain $\uX_{\cA_i} \mc \tX^i_{\cA_i} \mc \inp{\uX_{\bcA_i},\uY}$. 
        Thus, the reported variables and the side information are jointly distributed as $Q_{\uX_{[k]},\uY}$ i.i.d. 

        Let $\gamma>0$ and $\phi_n$ be a decoder with  error probability $\epsilon(\gamma,\phi_n) \leq \delta$ for some $\delta>0$. Therefore, $\eta(\gamma, \cA, W_{\cA}) \leq \delta$ for $\cA = \cA_i$ and $W_{\cA} = Q_{\bX_{\cA_i}|X_{\cA_i}}$ for all $i \in \cI$. This implies that for all $i \in \cI$ under $Q_{\ina{\uX_{\cA_i} \uX_{\bcA_i}} Y \tX^i_{\cA_i} } = Q_{\uX_{[k]}Y \tX^i_{\cA_i}}$, and therefore under $Q_{\uX_{[k]} \uY \inp{\tX^i_{\cA_i}}_{i \in \cI}}$ i.i.d., 
        \begin{align} \label{eq:kuser-blame-error-conv}
            \pr{
                \inp{
                    \phi \inp{
                        \uX^n_{[k]}, \uY^n
                    } \not\in \cA_i
                }
            } \leq \delta
        \end{align}
        and 
        \begin{align} 
        \pr&\Big(
            ({\phi({\uX^n_{[k]},\uY^n}) \not\in [k]}) \cap 
            \notag\\ &\;\;(\ham({\phi({\uX^n_{[k]},\uY^n}),f({\langle{\tX^{i,n}_{\cA_i}, \uX^n_{\bcA_i}}\rangle,\uY^n})})>\gamma) \Big) \leq \delta. \label{eq:kuser-ham-error-conv}
        \end{align}
        We have
        \begin{align*}
            \MoveEqLeft[2] \pr \inp{ \phi\inp{\uX^n_{[k]}, \uY^n} \not\in [k]}\\
            &\stackrel{(a)}{=} 1 - \pr \inp{\phi\inp{\uX^n_{[k]}, \uY^n} \in \bigcup_{i \in \cI} \inp{\bcA_i}}\\
            &\geq 1 - \sum_{i \in \cI} \pr \inp{\phi\inp{\uX^n_{[k]}, \uY^n} \not\in \cA_i}\\
            &\stackrel{(b)}{\geq} 1- |\cI|\delta, \addtocounter{equation}{1}\tag{\theequation}\label{eq:kuser-blame-err2-conv}
        \end{align*}
        where (a) is because $\bigcap_{i \in \cI} \cA_i = \varnothing$ and (b) is from \eqref{eq:kuser-blame-error-conv}. This gives that for any $i,j \in \cI$,
        \begin{align*}
            \MoveEqLeft[2] \pr \inp{\ham \inp{f\inp{\ina{\tX^{i,n}_{\cA_i}, \uX^n_{\bcA_i}}, \uY^n}, f\inp{\ina{\tX^{j,n}_{\cA_{j}}, \uX^n_{\bcA_{j}}}, \uY^n}}\leq 2\gamma}\\
            \geq& \pr 
            \left( 
                {\inp{\ham \inp{{\phi\inp{\uX^n_{[k]},\uY^n},f\inp{\ina{\tX^{j,n}_{\cA_j}, \uX^n_{\bcA_j}},\uY^n}}} \leq \gamma}} \right.\\
                &\cap 
                \left. \inp{\ham \inp{{\phi\inp{\uX^n_{[k]},\uY^n},f\inp{\ina{\tX^{i,n}_{\cA_i}, \uX^n_{\bcA_i}},\uY^n}}} \leq \gamma} \right.\\
                &\cap 
                \left. \inp{\phi\inp{\uX^n_{[k]}, \uY^n} \not\in [k]}
            \right)\\
            \geq& 1 - (|\cI|\delta + 2 \delta),
        \end{align*}
        where the first inequality is due to the triangle inequality and the last inequality follows from \eqref{eq:kuser-ham-error-conv} and \eqref{eq:kuser-blame-err2-conv} using union bound. For any $\gamma$ and $\delta$, the above holds for an increasing sequence of $n$, i.e., for any $\gamma>0$, passing to a subsequence if needed, for any $i,j \in \cI$, $\pr\inp{\ham \inp{f\inp{\ina{\tX^{i,n}_{\cA_i}, \uX^n_{\bcA_i}}, \uY^n}, f\inp{\ina{\tX^{j,n}_{\cA_{j}}, \uX^n_{\bcA_{j}}}, \uY^n}}\leq 2\gamma} \rightarrow 0$ as $n \rightarrow \infty$. Since, by the law of large numbers, $\ham \inp{f\inp{\ina{\tX^{i,n}_{\cA_i}, \uX^n_{\bcA_i}}, \uY^n}, f\inp{\ina{\tX^{j,n}_{\cA_{j}}, \uX^n_{\bcA_{j}}}, \uY^n}} \rightarrow \pr \inp{ f\inp{\ina{\tX_{\cA_i}^i, \uX_{\bcA_i}},Y} \neq f \inp{\ina{\tX_{\cA_{j}}^{j}, \uX_{\bcA_{j}}},\uY}}$ a.s., the result follows. 
    \end{proof}
    \subsection{Achievability}
The achievability proof is along the lines of that of Theorem~\ref{thm:2user} for the $k=2$, $s=1$ case.

    We start with an analog of Lemma~\ref{lem:decoding-function-2user} which is proved further ahead (in Appendix~\ref{sec:proof-of-kuser-single-letter}).
    \begin{lemma}\label{lem:kuser-single-letter}
        Let $\coll$ be an adversary structure and $f: \cX_1 \times \dots \times \cX_k \times \cY \rightarrow \cZ$ be $\coll$-viable. For every non-intersecting $\cI \subseteq [m]$, there is a $g_{\cI}: \cX_1 \times \dots \times \cX_k \times \cY \rightarrow \cZ$ such that if for any set of conditional distributions $\left\{W_{\bX_{\cA_i}|X_{\cA_i}}\right\}_{i \in \cI}$ satisfying, for all $\ux_{[k]},\uy$, and every $i,i' \in \cI$,
        \begin{multline*}
            \sum_{x_{A_i}} P_{X_{[k]},Y}(\langle x_{A_i},\ux_{{A_i}^c}\rangle,y) W_{\bX_{\cA_i}|X_{\cA_i}}(\ux_{\cA_i}|x_{\cA_i})\\
            =\sum_{x_{A_i'}} P_{X_{[k]},Y}(\langle x_{A_i'},\ux_{{A_i'}^c}\rangle,y) W_{\bX_{\cA_{i'}}|X_{\cA_{i'}}}(\ux_{\cA_{i'}}|x_{\cA_{i'}}),
        \end{multline*}
        then, it holds for every $i \in \cI$ that
        \begin{equation*}
        f(X_{[k]},Y) = g_{\cI} \left( \ina{\bX_{\cA_i}, X_{\bcA_i}} , Y \right) \text{ under } P_{X_{[k]},Y}W_{\bX_{\cA_i}|X_{\cA_i}}.
        \end{equation*}
    \end{lemma}

Analogous to the definition in Section~\ref{sec:2user}, we define the set of possible single-letter ``view'' distributions that an adversary $\cA_i\in\coll$ can induce: 
\begin{definition}
    For $i\in[m]$, 
    \begin{align*}
        \cV_i:=\big\{& Q_{\uX_{[k]},\uY}: \exists W_{\bX_{\cA_{i}}|X_{\cA_{i}}} \text{s.t. }\forall \ux_{[k]},\uy,\; 
              Q_{\uX_{[k]},\uY}(\ux_{[k]},\uy)=
\\&\quad\quad\sum_{x_{A_i}} P_{X_{[k]},Y}(\langle x_{A_i},\ux_{{A_i}^c}\rangle,\uy)W_{\bX_{\cA_{i}}|X_{\cA_{i}}}(\ux_{\cA_{i}}|x_{\cA_{i}})\big\}.
\end{align*}
For $\delta>0$, let $\cV_i^{\delta}=\bigcup_{Q\in\cV_i} \cB(Q,\delta)$, where $\cB(Q,\delta)=\{ Q': \dtv(Q,Q')\leq\delta\}$ is the set of p.m.f.s within total-variation distance $\delta$ of $Q$.
\end{definition}
We note that $P_{X_{[k]},Y}\in\cV_i$ for all $i\in[m]$ (by choosing $W_{\bX_{\cA_{i}}|X_{\cA_{i}}}$ to be the identity channel in the above definition). Hence, $\cap_{i\in\cI}\cV_i\neq\varnothing$  for any non-empty $\cI$ (recall that all non-intersecting $\cI$ are non-empty by definition).

\paragraph*{Decoder} Let $\delta>0$. We use the following decoder $\phi_n$ to prove the achievability. It receives the reported sequences $\ux^n_{[k]}$ and side information $\uy^n$ and computes the set $\cJ = \left\{i \in [m]: \type{\ux^n_{[k]},\uy^n} \in \cV_{i}^{\delta} \right\}$.
                Then, it outputs 
                $$\phi_n\left(\ux^n_{[k]},\uy^n\right) = \begin{cases}
                    \text{error}, & \cJ = \varnothing,\\
                    \bigcap_{i \in \cJ} \cA_i, & \cJ \neq \varnothing, \bigcap_{i \in \cJ} \cA_i \neq \varnothing,\\
                    g_{\cJ}(\ux^n_{[k]},\uy^n), & \cJ \neq \varnothing, \bigcap_{i \in \cJ} \cA_i = \varnothing,
                \end{cases}$$
                where $g_{\cJ}$ is from Lemma~\ref{lem:kuser-single-letter} (which provides a $g_\cI$ for every non-intersecting $\cI$, i.e., if $\cI$ is non-empty and $\bigcap_{i \in \cI} \cA_i$ is empty). Notice that if $\cJ \neq \varnothing$ and $\bigcap_{i \in \cJ} \cA_i$ is non-empty, the decoder above is shown to output the entire intersection; however it suffices to output any element in the intersection (say, the smallest) to fit the problem definition.

                To aid the proof of achievability, we prove the following technical lemma which is the analog of Lemma~\ref{lem:2user-technical-single-letter}. See further ahead (Appendix~\ref{sec:proof-of-kuser-technical-single-letter}) for a proof.
\begin{lemma}\label{lem:kuser-technical-single-letter}
For an $\coll$-viable $f$ and a non-intersecting $\cI\subseteq[m]$ along with a corresponding $g_{\cI}$ as in Lemma~\ref{lem:kuser-single-letter}, for every $i\in\cI$, there is a $\gamma_{i,\cI}:\bb{R}^{+} \rightarrow \bb{R}^{+}$, such that  $\gamma_{i,\cI}(\delta)\to 0$ as $\delta\to 0$, satisfying the following: for $\delta>0$, if $Q_{\uX_{[k]},\uY,X_{\cA_i}}$ is such that  $Q_{\uX_{[k]},\uY}\in  \cap_{j \in \cI} \cV^{\delta}_j$ and there is a $W_{\bX_{\cA_i}|X_{\cA_i}}$ such that  
$\dtv({Q_{\langle{X_{\cA_i}, \uX_{\bcA_i}}\rangle, \uY, \uX_{\cA_i}}, P_{X_{[k]}Y}W_{\bX_{\cA_i}|X_{\cA_i}}}) \leq \delta$, then
$\pr(f(\langle{X_{\cA_i}, \uX_{\bcA_i}}\rangle, \uY)\neq g_{i,\cI}(\uX_{[k]},\uY))\leq \gamma_{i,\cI}(\delta)$ under the p.m.f. $Q_{\uX_{[k]},\uY,X_{\cA_i}}$.
\end{lemma}

Using this lemma, we prove (further ahead in Appendix~\ref{sec:proofs-of-kuser-blame-hamming-bounds}) the following two lemmas which are the analogs of parts~(i) and~(ii), respectively, of Lemma~\ref{lem:2user-typicality}.
\begin{lemma}\label{lem:kuser-blame-bound}
For $i \in [m]$, if $X^n_{[k]}, Y^n, \bX^n_{\cA_i}$ are jointly distributed as 
\begin{multline*}
    Q_{X^n_{[k]}Y^n\bX^n_{\cA_i}}(x^n_{[k]},y^n,\bx^n_{\cA_i}) \\= \left( \prod_{t \in [n]} P_{X_{[k]}Y}(x_{[k],t},y_t)\right)W_{\cA_i}(\bx^n_{\cA_i}|x^n_{\cA_i})
\end{multline*}
for some (not necessarily memoryless) channel $W_{\cA_i}$,
 \[\pr\big(\type{\langle{\bX^n_{\cA_i}, X^n_{\bcA_i}}\rangle,Y^n}\notin \cV_i^{\delta} \big) \leq 2^{-\Omega(n)},\]
where the factors hidden in $\Omega(n)$ do not depend on $W_{\cA_i}$.
\end{lemma}
\begin{lemma}\label{lem:kuser-hamming-bound}
For $i \in [m]$ and a non-intersecting $\cI$ which contains $i$, there is a function $\gamma_{i,\cI}:\bb{R}^{+} \rightarrow \bb{R}^{+}$, such that $\gamma_{i,\cI}(\delta)\to 0$ as $\delta\to 0$ and, if $X^n_{[k]}, Y^n, \bX^n_{\cA_i}$ are jointly distributed as 
\begin{multline*}
    Q_{X^n_{[k]}Y^n\bX^n_{\cA_i}}(x^n_{[k]},y^n,\bx^n_{\cA_i}) \\= \left( \prod_{t \in [n]} P_{X_{[k]}Y}(x_{[k],t},y_t)\right)W_{\cA_i}(\bx^n_{\cA_i}|x^n_{\cA_i})
\end{multline*}
for some (not necessarily memoryless) channel $W_{\cA_i}$, then
\begin{align*}
\pr\Big(\big(\type{\langle{\bX^n_{\cA_i}, X^n_{\bcA_i}}\rangle,Y^n} \in &\bigcap_{j \in \cI} \cV^{\delta}_j\big) \bigcap\big(\ham\big(g_{\cI}\big(\langle{\bX^n_{\cA_i}, X^n_{\bcA_i}}\rangle, Y^n\big),
    \\&\quad 
    f\big(X^n_{[k]},Y^n\big)\big)>\gamma_{i,\cI}(\delta)\big)\Big) \leq 2^{-\Omega(n)},
\end{align*}
where the factors hidden in $\Omega(n)$ do not depend on $W_{\cA_i}$.

\end{lemma}
These two lemmas imply the achievability of Theorem~\ref{thm:kuser-main} as we argue now. Suppose adversary $\cA_i\in\coll$ applies a (not necessarily memoryless) $W_{\cA_i}$ to produce a purported $\bX_{\cA_i}^n$ from $X_{\cA_i}^n$. Then, Lemma~\ref{lem:kuser-blame-bound} asserts that the empirical distribution of the reported observations and side-information will (w.h.p.) lie in $\cV_i^{\delta}$. Under this event the $\cI$ computed by the decoder includes $i$. Hence, it will neither declare ``error'' nor will it name a user a who is not in $\cA_i$ as malicious. Now consider the event where the decoder outputs an estimate of $f(X_{[k]}^n,Y^n)$; note that the $\cJ$ below is the subset of $[m]$ computed by the decoder; since this based on the reports it received and the side-information, it is random. For $\gamma>0$,
\begin{align*}
&\pr \big( (\cJ\neq\varnothing)\cap( \cap_{j \in \cJ} \cA_j = \varnothing ) \cap\\
&\qquad\ham(g_{\cJ}(\langle{\bX^n_{\cA_i}, X^n_{\bcA_i}}\rangle,Y^n), f(X^n_{[k]},Y^n)) > \gamma \big)\\
&\leq \pr \big( (\type{\langle{\bX^n_{\cA_i}, X^n_{\bcA_i}}\rangle,Y^n}\in \cV_i^{\delta})\cap
(\cap_{j \in \cJ} \cA_j = \varnothing ) \cap\\
&\qquad(\ham(g_{\cJ}(\langle{\bX^n_{\cA_i}, X^n_{\bcA_i}}\rangle,Y^n), f(X^n_{[k]},Y^n)) > \gamma) \big) + 2^{-\Omega(n)},
\end{align*}
where the inequality follows from Lemma~\ref{lem:kuser-blame-bound}. To bound the first term, recalling that the set of all non-intersecting (and hence also non-empty; see Definition~\ref{def:non-intersecting}) subsets $\cI$ of [m] is denoted by $\mathbb{I}$, we note that the following two events are the same.
\begin{align*}
 &\left(  \left(\type{\langle{\bX^n_{\cA_i}, X^n_{\bcA_i}}\rangle,Y^n}\in \cV_i^{\delta}\right)\bigcap
\left(\bigcap_{j \in \cJ} \cA_j = \varnothing \right) \right)\\
&\qquad=
\bigcup_{\cI \in\mathbb{I}: i\in\cI} (\cJ=\cI),\\
&\qquad=
\Big(\bigcup_{\cI \in\mathbb{I}: i\in\cI}
 \big(\type{\langle{\bX^n_{\cA_i}, X^n_{\bcA_i}}\rangle,Y^n}\in \bigcap_{j\in\cI}\cV_j^{\delta}\big)\Big)\cap (\cJ=\cI)
\end{align*}
where note that $\cJ$ is the (random) subset of $[m]$ computed by the decoder. Plugging this into the first probability term above, we have
\begin{align*}
&\pr \big( (\type{\langle{\bX^n_{\cA_i}, X^n_{\bcA_i}}\rangle,Y^n}\in \cV_i^{\delta})\cap
(\cap_{j \in \cJ} \cA_j = \varnothing ) \cap\\
&\qquad\qquad(\ham(g_{\cJ}(\langle{\bX^n_{\cA_i}, X^n_{\bcA_i}}\rangle,Y^n), f(X^n_{[k]},Y^n)) > \gamma) \big)\\
&=\pr \Big(\Big(\bigcup_{\cI \in\mathbb{I}: i\in\cI}
 \big(\type{\langle{\bX^n_{\cA_i}, X^n_{\bcA_i}}\rangle,Y^n}\in \bigcap_{j\in\cI}\cV_j^{\delta}\big)\Big)\cap (\cJ=\cI)\cap\\
&\qquad\qquad(\ham(g_{\cJ}(\langle{\bX^n_{\cA_i}, X^n_{\bcA_i}}\rangle,Y^n), f(X^n_{[k]},Y^n)) > \gamma) \Big)\\
&\leq \sum_{\cI \in\mathbb{I}: i\in\cI} \pr\Big( \big(\type{\langle{\bX^n_{\cA_i}, X^n_{\bcA_i}}\rangle,Y^n}\in \bigcap_{j\in\cI}\cV_j^{\delta}\big)\Big)\cap\\
&\qquad\qquad(\ham(g_{\cI}(\langle{\bX^n_{\cA_i}, X^n_{\bcA_i}}\rangle,Y^n), f(X^n_{[k]},Y^n)) > \gamma) \Big)\\
&\leq 2^{-\Omega(n)},
\end{align*}
where the last step follows from Lemma~\ref{lem:kuser-hamming-bound} if we take $\gamma$ as $\max_{\cI\in\mathbb{I}:i\in\cI} \gamma_{i,\cI}(\delta)$. This completes the proof of achievability of Theorem~\ref{thm:kuser-main}. The proofs of supporting lemmas follow.
\subsection{Proofs of Lemmas used in the  proof of achievability of Theorem~\ref{thm:kuser-main}}\label{sec:proofs-of-general-achievability-lemmas}
    \subsubsection{Proof of Lemma~\ref{lem:kuser-single-letter}}\label{sec:proof-of-kuser-single-letter}
The proof proceeeds along the lines of the proof of Lemma~\ref{lem:decoding-function-2user} (for two users). 
For $P_{X_{[k]} Y}$ and $\cI$, define the set
    \begin{align*}
&\cQ_{P_{X_{[k]}Y}, \cI} = \Bigg\{Q_{\uX_{[k]} \uY\inp{\tX^{i}_{\cA_i}}_{i\in \cI}}:\\
            & \qquad Q_{\uX_{[k]} \uY\inp{\tX^{i}_{\cA_i}}_{i\in \cI}} =  Q_{\uX_{[k]}\uY}\prod_{i\in\cI}Q_{\tX^{i}_{\cA_i}|\uX_{[k]} \uY}\\
            &\qquad \text{ satisfying} \quad Q_{\uX_{[k]}  \uY \tX^{i}_{\cA_i}}(\ux_{[k]}, \uy, \tx^{i}_{\cA_i})\\
            &\qquad \qquad  =P_{{{X_{\cA_i} X_{\bcA_i}}}Y}({\tx^{i}_{\cA_i}, \ux_{\bcA_i}}, \uy)Q_{\uX_{\cA_i}|\tX^i_{\cA_i}}(\ux_{\cA_i}|\tx^{i}_{\cA_i}),  \\
            &\qquad \text{ where }  Q_{\uX_{[k]} \uY}(\ux_{[k]}, \uy) = \\
            &\qquad \sum_{\tx^{i}_{\cA_i}}P_{{{X_{\cA_i} X_{\bcA_i}}}Y}({\tx^{i}_{\cA_i}, \ux_{\bcA_i}}, \uy)Q_{\uX_{\cA_i}|\tX^i_{\cA_i}}(\ux_{\cA_i}|\tx^{i}_{\cA_i}).\Bigg\}
            \end{align*}

Notice that for any $\left\{W_{\bX_{\cA_i}|X_{\cA_i}}\right\}_{i \in \cI}$ satisfying the assumption in the lemma, there exists a $Q_{\uX_{[k]} \uY\inp{\tX^{i}_{\cA_i}}_{i\in \cI}}\in \cQ_{P_{X_{[k]}Y}, \cI}$.

For any $(\ux_{[k]}, \uy)$ such that $Q_{\uX_{[k]}\uY}(\ux_{[k]}, \uy)>0$ for some $Q_{\uX_{[k]} \uY\inp{\tX^{i}_{\cA_i}}_{i\in\cI}}\in \cQ_{P_{X_{[k]}Y}, \cI}$, we define $g(\ux_{[k]}, \uy) = f(\langle\tx^{i}_{\cA_i}, \ux_{\bcA_i}\rangle, \uy)$ for some $i\in \cI$ and $\tx^{i}_{\cA_i}$ such that $Q_{\tX^i_{\cA_i}|\uX_{[k]} \uY}(\tx^i_{\cA_i}|\ux_{[k]},\uy)>0$. By definition~\ref{lem:A-viable}, this also implies that $g(\ux_{[k]}, \uy) = f(\langle\tx^{i}_{\cA_i}, \ux_{\bcA_i}\rangle, \uy) = f(\langle\tx^{j}_{\cA_j}, \ux_{\bcA_j}\rangle, \uy)$ for any $j\in \cI$ and $\tx^{j}_{\cA_j}$ such that $Q_{\tX^j_{\cA_j}|\uX_{[k]} \uY}(\tx^j_{\cA_j}|\ux_{[k]},\uy)>0$.

We will argue that the function $g$ as defined above is the same for every $Q_{\uX_{[k]} \uY\inp{\tX^{i}_{\cA_i}}_{i\in \cI}}\in \cQ_{P_{X_{[k]}Y}, \cI}$ (and hence for every  $\left\{W_{\bX_{\cA_i}|X_{\cA_i}}\right\}_{i \in \cI}$ satisfying the assumption in the lemma). Suppose not, then there exists $(\ux_{[k]}, \uy)\in \cX_{[k]}\times\cY$ such that $Q^{(a)}_{\uX_{[k]}\uY}(\ux_{[k]}, \uy)>0$ and $Q^{(b)}_{\uX_{[k]}\uY}(\ux_{[k]}, \uy)>0$ for some $Q^{(a)}_{\uX_{[k]} \uY\inp{\tX^{i}_{\cA_i}}_{i\in \cI}}, Q^{(b)}_{\uX_{[k]} \uY\inp{\tX^{i}_{\cA_i}}_{i\in \cI}}\in \cQ_{P_{X_{[k]}Y}, \cI}$ resulting in distinct functions $g^{(a)}$ and $g^{(b)}$ such that $g^{(a)}(\ux_{[k]}, \uy) \neq g^{(b)}(\ux_{[k]}, \uy)$. This also implies that for some $i\in \cI$, there exist $\tx^{i}_{\cA_i}, \hat{x}^{i}_{\cA_i}$ where $ \tx^{i}_{\cA_i}\neq  \hat{x}^{i}_{\cA_i}$ such that $Q^{(a)}_{\tX^{i}_{\cA_i}|\uX_{[k]}\uY}(\tx^{i}_{\cA_i}|\ux_{[k]}, \uy), Q^{(b)}_{\tX^{i}_{\cA_i}|\uX_{[k]}\uY}(\hat{x}^{i}_{\cA_i}|\ux_{[k]}, \uy)>0$ and $g^{(a)}(\ux_{[k]}, \uy) = f(\langle \tx^{i}_{\cA_i}, \ux_{\bcA_i}\rangle, \uy)\neq f(\langle \hat{x}^{i}_{\cA_i}, \ux_{\bcA_i}\rangle, \uy) = g^{(b)}(\ux_{[k]}, \uy)$.

We define $$\bar{Q}_{\uX_{[k]} \uY\inp{\tX^{i}_{\cA_i}}_{i\in \cI}} =  \bar{Q}_{\uX_{[k]}\uY}\prod_{i\in\cI}\bar{Q}_{\tX^{i}_{\cA_i}|\uX_{[k]} \uY}$$ where  for $i\in \cI$, 
\begin{align*}
   &\bar{Q}_{\uX_{[k]}  \uY \tX^{i}_{\cA_i}}(\ux_{[k]}, \uy, \tx^{i}_{\cA_i})  =\\
   &P_{{{X_{\cA_i} X_{\bcA_i}}}Y}({\tx^{i}_{\cA_i}, \ux_{\bcA_i}}, \uy)\bar{Q}_{\uX_{\cA_i}|\tX^i_{\cA_i}}(\ux_{\cA_i}|\tx^{i}_{\cA_i}) 
\end{align*}
 for $\bar{Q}_{\uX_{\cA_i}|\tX^i_{\cA_i}}$  defined as below.
\begin{align*}
\bar{Q}_{\uX_{\cA_i}|\tX^i_{\cA_i}} = \inp{1-\lambda}{Q}^{(a)}_{\uX_{\cA_i}|\tX^i_{\cA_i}}+\lambda {Q}^{(b)}_{\uX_{\cA_i}|\tX^i_{\cA_i}}.
\end{align*}
From the definitions of $\bar{Q}_{\uX_{[k]} \uY\inp{\tX^{i}_{\cA_i}}_{i\in \cI}}, {Q}^{(a)}_{\uX_{[k]} \uY\inp{\tX^{i}_{\cA_i}}_{i\in \cI}}$ and ${Q}^{(b)}_{\uX_{[k]} \uY\inp{\tX^{i}_{\cA_i}}_{i\in \cI}}$, it follows that  $\bar{Q}_{\uX_{[k]} \uY\inp{\tX^{i}_{\cA_i}}_{i\in \cI}}\in \cQ_{P_{X_{[k]}Y}, \cI}$. As  $Q^{(a)}_{\tX^{i}_{\cA_i}|\uX_{[k]}\uY}(\tx^{i}_{\cA_i}|\ux_{[k]}, \uy)>0$ and $Q^{(b)}_{\tX^{i}_{\cA_i}|\uX_{[k]}\uY}(\hat{x}^{i}_{\cA_i}|\ux_{[k]}, \uy)>0$, we have $\bar{Q}_{\tX^{i}_{\cA_i}|\uX_{[k]}\uY}(\tx^{i}_{\cA_i}|\ux_{[k]}, \uy)>0$ and  $\bar{Q}_{\tX^{i}_{\cA_i}|\uX_{[k]}\uY}(\hat{x}^{i}_{\cA_i}|\ux_{[k]}, \uy)>0$. Consider any $j\neq i, j\in \cI$  and $\tx^{j}_{\cA_j}\in \cX_{\cA_j}$ such that $\bar{Q}_{\tX^{j}_{\cA_j}|\uX_{[k]}\uY}(\tx^{j}_{\cA_j}|\ux_{[k]}, \uy)>0$. Then, from definition of $\bar{Q}_{\uX_{[k]} \uY\inp{\tX^{i}_{\cA_i}}_{i\in \cI}}$ and definition~\ref{lem:A-viable}, we have $f(\langle \tx^{i}_{\cA_i}, \ux_{\bcA_i}\rangle, \uy) = f(\langle \tx^{j}_{\cA_j}, \ux_{\bcA_j}\rangle, \uy) = f(\langle \hat{x}^{i}_{\cA_i}, \ux_{\bcA_i}\rangle, \uy)$, leading to a contradiction. Thus, $g$ is defined uniquely.  


\subsubsection{Proof of Lemma~\ref{lem:kuser-technical-single-letter}} \label{sec:proof-of-kuser-technical-single-letter}
The proof of Lemma~\ref{lem:kuser-technical-single-letter} very closely follows that of Lemma~\ref{lem:2user-technical-single-letter} in Section~\ref{sec:2user-technical-single-letter-proof}. We first introduce some notation analogous to what was used there.
For $\cA \subseteq [k]$, define $\ttf_{\cA}:\cP(\cX_{\cA}|\cX_{\cA}) \rightarrow \cP(\cX_{[k]} \times \cY)$ as $\ttf_{\cA}(W_{\bX_{\cA}|X_{\cA}})=R_{\ina{\bX_{\cA},X_{\bcA}}, Y}$ where 
\begin{multline*}
        R_{\ina{\bX_{\cA}, X_{\bcA}}, Y}(\ina{\bx_{\cA}, x_{\bcA}}, y) \\
        = \sum_{x_{\cA}} P_{X_{[K]}Y}(\ina{x_{\cA}, x_{\bcA}}, y) W_{\bX_{\cA}|X_{\cA}}(\bx_{\cA}|x_{\cA}).
\end{multline*}
For any set $\cR \subseteq \cP(\cX_{[k]} \times \cY)$, define $\ttf_{\cA}^{-1}(\cR) = \bigcup_{R \in \cR} \ttf_{\cA}^{-1}(R)$ where $$\ttf_{\cA}^{-1}(R) = \left\{ W_{\bX_{\cA}|X_{\cA}} \in \cP(\cX_{\cA}|\cX_{\cA}): \ttf_{\cA}(W_{\bX_{\cA}|X_{\cA}}) = R\right\}.$$
        
The following lemma is the analog of Lemma~\ref{lem:comp-int}. 
\begin{lemma}\label{lem:kcompact}
For a non-intersecting $\cI\subseteq[m]$, there is a function $\varepsilon_{\cI}: \bb{R}_{>0} \rightarrow \bb{R}_{\geq 0}$ such that
\begin{enumerate}[label=(\alph*)]
    \item as $\rad \rightarrow 0$, $\varepsilon_{\cI}(\rad) \rightarrow 0$, and
    \item for $\rad>0$, if $R \in \bigcap_{i \in \cI}\cV_{i}^{\rad}$, $\exists$ $S \in \bigcap_{i \in \cI}\cV_i$ such that $\dtv(R,S)\leq \varepsilon_{\cI}(\rad)$.
\end{enumerate}
\end{lemma}
The proof is along the same lines as for Lemma~\ref{lem:comp-int}. We include it here for the sake of completeness.
\begin{proof}
    Since $\bigcap_{i \in \cI} \cV_i$ is closed and non-empty, $\min_{S \in \bigcap_{i \in \cI} \cV_i} \dtv(R,S)$ is well-defined for every $R \in \bigcap_{i \in \cI} \cV_{i}^{\rad}$.
    We will show that the function $\varepsilon_{\cI}(\rad) = \sup_{R \in \bigcap_{i \in \cI} \cV^{\rad}_{i}} \min_{S \in \bigcap_{i \in \cI} \cV_i} \dtv(R,S)$ satisfies both the given conditions.
    The fact that $\varepsilon_{\cI}$ satisfies condition (b) is obvious from its definition.
    Now, suppose it doesn't satisfy condition (a), then, taking into account that $\varepsilon_{\cI}(\rad)$ is a non-negative, non-decreasing function of $\rad$, we obtain that 
    \begin{equation} \label{eq:compactcontra}
                \exists\ \varepsilon_0 > 0 \text{ such that } \forall\ \rad > 0, \sup_{R \in \bigcap_{i \in \cI} \cV_{i}^{\rad}} \dtv\left( R, \bigcap_{i \in \cI} \cV_i \right) \geq \varepsilon_0.
    \end{equation}
    For $n \in \bb{N}$, setting $\rad = 1/n$ in \eqref{eq:compactcontra}, we obtain a sequence $R_n$ such that for every $i \in \cI$, $\dtv(R_n, \cV_i) \leq 1/n$ and $\dtv(R_n, \bigcap_{i \in \cI} \cV_i) > \varepsilon_0/2$, $n \in \bb{N}$.
    Appealing to the fact that $\cV_i$ is closed for every $i \in \cI$, we define the sequences $A_{i,n} \in \cV_i, n \in \bb{N}$ as follows:
    \begin{align}
                A_{i,n} = \arg \min_{S \in \cV_i} \dtv(R_n,S),
    \end{align}
    where $A_{i,n}$ is arbitrarily chosen to be one of the minimizers in case more than one exists.
    Now, fix some $i \in \cI$.
    Since $\cV_i$ is compact, the sequence $A_{i,n}$ has a limit point $A^{\ast} \in \cV_i$. 
    Furthermore, since $\dtv(R_n,A_{i,n}) \rightarrow 0$ as $n \rightarrow \infty$, sequence $R_n$ also has $A^{\ast}$ as one of its limit points.
    This, along with the assumption that $\dtv(R_n, \bigcap_{i \in \cI}\cV_i) > \varepsilon_0/2\ \forall\ n \in \bb{N}$ implies that $\dtv(A^{\ast},\bigcap_{i \in \cI}\cV_i) > \varepsilon_0/2$. 
    Now, in order to observe the contradiction, note that since for every $j \in \cI, j \neq i$, $\dtv(R_n,A_{j,n}) \rightarrow 0$ as $n \rightarrow \infty$ as well, all such sequences $A_{j,n}$ also have $A^{\ast}$ as one of their limit points. 
    But since $\cV_j$ is closed for all $j \in \cI$, $A^{\ast}\in \cV_j\ \forall\ j \in \cI$ and therefore $A^{\ast} \in \bigcap_{i \in \cI}V_{i}$, i.e., $\dtv(A^{\ast}, \bigcap_{i \in \cI} \cV_i) = 0 < \varepsilon_0/2$.
\end{proof}
Analogous to Lemma~\ref{lem:2continuous}, we have    
        \begin{lemma} \label{lem:kcontinuous}
            For $i \in [m]$ and non-empty $\cV \subseteq \cV_i$, there exists an $\eta_{i,\cV}:\bb{R}_{>0} \rightarrow \bb{R}_{\geq 0}$ such that
            \begin{enumerate}[label = (\alph*)]
                \item
                as $\varepsilon \rightarrow 0$, $\eta_{i,\cV}(\varepsilon) \rightarrow 0$, and
                \item
                for $\varepsilon >0$, if $Q \in \cP(\cX_{\cA_i}|\cX_{\cA_i})$ and $\dtv(\ttf_{\cA_i}(Q),\cV) \leq \varepsilon$, then, there exists $Q' \in \ttf_{\cA_i}^{-1}(\cV)$, such that $d(Q,Q') \leq \eta_{i,\cV}(\varepsilon)$.
            \end{enumerate}
        \end{lemma}
This is in fact a corollary of Lemma~\ref{lem:2continuous} if we treat $\cX_{A_i}$ and $\cX_{A_i^c}$ as $\cX_1$ and $\cX_2$, respectively, in Lemma~\ref{lem:2continuous}. Notice that $\Phi_{\cA_i}$ and $\Phi^{-1}_{\cA_i}$ here are then identical to $\Phi_1$ and $\Phi^{-1}_{1}$ in Lemma~\ref{lem:2continuous}. Therefore, we omit the proof.

        We now state the lemma analogous to Lemma~\ref{lem:old-2user-technical-single-letter}.
\begin{lemma}\label{lem:old-kuser-technical-single-letter}
For an $\coll$-viable $f$ and a non-intersecting $\cI\subseteq[m]$ along with a corresponding $g_{\cI}$ as in Lemma~\ref{lem:kuser-single-letter}, for every $i\in\cI$, there is a $\gamma'_{i,\cI}:\bb{R}^{+} \rightarrow \bb{R}^{+}$, such that  $\gamma'_{i,\cI}(\delta)\to 0$ as $\delta\to 0$, satisfying the following: for $\delta>0$, if 
$W_{\bX_{\cA_i}|X_{\cA_i}}$ is such that the p.m.f
\begin{align*}
&Q_{\uX_{[k]},\uY}(\ux_{[k]},\uy)\\
&\qquad:=
\sum_{x_{\cA_i}}
P_{X_{[k]},Y}(\langle{x_{\cA_i},\ux_{\cA_i^c}}\rangle,\uy)W_{\bX_{\cA_i}|X_{\cA_i}}(\ux_{\cA_i}|x_{\cA_i})
\end{align*}
belongs to $\bigcap_{j \in \cI} \cV^{\delta}_j$, then
\[\pr(f(X_{[k]},Y)\neq g_{i,\cI}(\langle{\bX_{\cA_i}, X_{\cA_i^c}}\rangle,Y))\leq \gamma'_{i,\cI}(\delta)\] under the p.m.f. 
$P_{X_{[k]},Y}W_{\bX_{\cA_i}|X_{\cA_i}}$.
\end{lemma}
\begin{proof}
The proof is along the lines of Lemma~\ref{lem:old-2user-technical-single-letter} in Section~\ref{sec:2user-technical-single-letter-proof}. Following that, we may argue that $\gamma'_{i,\cI}(\delta):=|\cX_{\cA_i}|\eta_{i,\bigcap_{j \in \cI} \cV_j}(\varepsilon_{\cI}(\delta))$
satisfies the required properties, where $\varepsilon_{\cI}$ and $\eta_{i,\cV}$  functions are as in Lemmas~\ref{lem:kcompact} and~\ref{lem:kcontinuous}, respectively. We omit the details.
\end{proof}

Finally, Lemma~\ref{lem:kuser-technical-single-letter} follows from Lemma~\ref{lem:old-kuser-technical-single-letter} along the same lines as the proof (in Section~\ref{sec:2user-technical-single-letter-proof}) of Lemma~\ref{lem:2user-technical-single-letter} from Lemma~\ref{lem:old-2user-technical-single-letter}.

\subsubsection{Proofs of Lemmas~\ref{lem:kuser-blame-bound} and~\ref{lem:kuser-hamming-bound}} \label{sec:proofs-of-kuser-blame-hamming-bounds}

For $\delta>0$, define
\begin{align*}
\cQ_n(\delta):=\{&Q_{X'_{[k]} Y' \bX'_{\cA_i}}\in\cP_n(\cX_{[k]} \times \cY \times \cX_{\cA_i}) : \\
 &\qquad D(Q_{X'_{[k]} Y' \bX'_{\cA_i}}||P_{X_{[k]}Y}Q_{\bX'_{\cA_i}|X'_{\cA_i}})\leq 2{\delta}^2\}.
\end{align*}
Proceeding along the lines of the proof of Lemma~\ref{lem:2user-typicality} (treating $\cX_{A_i}$ and $\cX_{A_i^c}$ as $\cX_1$ and $\cX_2$, respectively, in that proof), analogous to~\eqref{eq:lemma3parta}, we obtain
\begin{align}
\pr\inp{\type{X_{[k]}^n,Y^n,\bX_{\cA_i}^n}\in \cQ_n(\delta)}\geq 1-2^{-\Omega(n)}.\label{eq:kuser-typicality-proof1}
\end{align}
This implies Lemma~\ref{lem:kuser-blame-bound} (just as~\eqref{eq:lemma3parta} implied part~(i) of Lemma~\ref{lem:2user-typicality}).
To show Lemma~\ref{lem:kuser-hamming-bound},
\begin{align*}
&\pr\Big(\big(\type{\langle{\bX^n_{\cA_i}, X^n_{\bcA_i}}\rangle,Y^n} \in \bigcap_{j \in \cI} \cV^{\delta}_j\big) \bigcap\\
&\qquad\big(\ham\big(g_{\cI}\big(\langle{\bX^n_{\cA_i}, X^n_{\bcA_i}}\rangle, Y^n\big),
    f\big(X^n_{[k]},Y^n\big)\big)>\gamma_{i,\cI}(\delta)\big)\Big)\\
&\stackrel{\text{(a)}}{=}\pr\Big(\big(\type{\langle{\bX^n_{\cA_i}, X^n_{\bcA_i}}\rangle,Y^n} \in \bigcap_{j \in \cI} \cV^{\delta}_j\big) \bigcap
\big(\type{X_{[k]}^n,Y^n,\bX_{\cA_i}^n}\in \cQ_n(\delta)\big)
\\
&\qquad\big(\ham\big(g_{\cI}\big(\langle{\bX^n_{\cA_i}, X^n_{\bcA_i}}\rangle, Y^n\big),
    f\big(X^n_{[k]},Y^n\big)\big)>\gamma_{i,\cI}(\delta)\big)\Big)\\
&\quad + 2^{-\Omega(n)}\\
&\stackrel{\text{(b)}}{=}: C + 2^{-\Omega(n)},
\end{align*}
where (a) follows from \eqref{eq:kuser-typicality-proof1} and in (b) we define the first term as $C$. Along the same lines as we argued that a similar $C$ term is 0 in the proof of Lemma~\ref{lem:2user-typicality}, we can show that, by Lemma~\ref{lem:kuser-technical-single-letter}, the term $C$ defined above is 0.

%% file: appendix-LP.tex
\section{Proof of Remark~\ref{rmk:decidability}}\label{app:LP-remark}
In this section, we will show that the viability condition in Definition~\ref{def:s-viable} can be efficiently checked using a linear program.
The s-viability states that
$f\in\cF_s(P_{X_{[k]}Y})$ if for any collection $\cA_1,\cA_2,\ldots,\cA_m$ of distinct subsets of $[k]$ such that each $|\cA_i|\leq s$ and $\bigcap_{i=1}^m \cA_i=\varnothing$, under every joint p.m.f. $Q_{\uX_{[k]},\uY,\left(\tX^{i}_{\cA_i}\right)_{i\in[m]}}$ over 
\[ \mathcal{S}=\cX_1\times\ldots\cX_k\times\cY\times\prod_{i=1}^m\prod_{j\in \cA_i}\cX_j \]
satisfying, for each $i\in[m]$,
\begin{itemize}
\item[(a)] 
$Q_{\tX^i_{\cA_i},\,\uX_{\bar{\cA}_i},\uY}=P_{X_{[k]}Y}$ and
\item[(b)] 
$\uX_{\cA_i} \mc \tX^i_{\cA_i} \mc (\uX_{\bar{\cA}_i},\uY)$,
\end{itemize}
we have, for all $i,i'\in[m]$ (with probability 1)
\begin{align}\label{eq:alt-s-viability}
f\inp{\ina{\tX^{i}_{\cA_i},\,\uX_{\bar{\cA}_i}},\uY} = f\inp{\ina{\tX^{i'}_{\cA_{i'}},\,\uX_{\bar{\cA}_{i'}}},\uY}.
\end{align}

Fix any $\cA_1,\cA_2,\ldots,\cA_m$.
We first establish that the set of all joint p.m.f. $Q_{\uX_{[k]},\uY,\left(\tX^{i}_{\cA_i}\right)_{i\in[m]}}$ satisfying conditions (a) and (b) for all $i\in[m]$ can be described as the feasibility region of a linear program.
Any p.m.f. $Q_{\uX_{[k]},\uY,\left(\tX^{i}_{\cA_i}\right)_{i\in[m]}}$ can described by a $|\mathcal{S}|$ dimensional vector in which each element is addressed by an element of $\mathcal{S}$ such that all elements are non-negative and add up to $1$.
Next, focus on condition (a) for any $i\in[m]$.
For every $\left(\tx^i_{\cA_i},\,\ux_{\bar{\cA}_i},\uy\right)$,  
\[ Q_{\tX^i_{\cA_i},\,\uX_{\bar{\cA}_i},\uY}\left(\tx^i_{\cA_i},\,\ux_{\bar{\cA}_i},\uy\right) = P_{X_{[k]}Y}\left(\tx^i_{\cA_i},\,\ux_{\bar{\cA}_i},\uy\right).\]
Since the LHS can be computed by adding up the appropriate coordinates of the $|\mathcal{S}|$-dimensional vector and RHS is fixed,
we conclude that condition (a) for each $i$ amounts to a set of linear constraints.
Next, for each $i\in[m]$, condition (b) can be written as
\begin{align*}
 &\frac{Q_{\uX_{\cA_i},\tX^i_{\cA_i},\uX_{\bar{\cA}_i},\uY}\left(\ux_{\cA_i},\tx^i_{\cA_i},\ux_{\bar{\cA}_i},\uy\right)}{P_{X_{[k]}Y}\left(\tx^i_{\cA_i},\ux_{\bar{\cA}_i},\uy\right)} \\
 &=Q_{\uX_{\cA_i} | \tX^i_{\cA_i},\uX_{\bar{\cA}_i},\uY}\left(\ux_{\cA_i} | \tx^i_{\cA_i},\ux_{\bar{\cA}_i},\uy\right)\\
 &=Q_{\uX_{\cA_i} | \tX^i_{\cA_i}}\left(\ux_{\cA_i} | \tx^i_{\cA_i}\right)\\
 &=\frac{Q_{\uX_{\cA_i} , \tX^i_{\cA_i}}\left(\ux_{\cA_i} , \tx^i_{\cA_i}\right)}{P_{X_{\cA_i}}\left(\tx^i_{\cA_i}\right)}
\end{align*}
for all $\left(\tx^i_{\cA_i},\,\ux_{\bar{\cA}_i},\uy\right)$.
In the first and last equalities, we used the fact that, $Q_{\tX^i_{\cA_i},\,\uX_{\bar{\cA}_i},\uY}=P_{X_{[k]}Y}$ whenever condition (a) is satisfied.
Since computing marginals is a linear operation and the p.m.f. of $P_{X_{[k]}Y}$ is fixed, condition (b) for each $i\in[m]$ can be written as a set of linear constraints.
Let $R$ be the feasibility region of the set of linear constraints defined by conditions (a) and (b) for all $i\in[m]$.
Finally, the equality in \eqref{eq:alt-s-viability}, is equivalent to requiring that, for any $\left(\ux_{[k]},\uy,\left(\tx^{i}_{\cA_i}\right)_{i\in[m]}\right)$ such that for some $i,i'\in[m]$,
$
f\inp{\tx^{i}_{\cA_i},\,\ux_{\bar{\cA}_i},\uy} \neq f\inp{\tx^{i'}_{\cA_{i'}},\,\ux_{\bar{\cA}_{i'}},\uy}
$,
\[Q_{\uX_{[k]},\uY,\left(\tX^{i}_{\cA_i}\right)_{i\in[m]}}\left(\ux_{[k]},\uy,\left(\tx^{i}_{\cA_i}\right)_{i\in[m]}\right)=0.\]
Let $R'$ be the feasibility region of the set of linear constraints defined by conditions (a) and (b) for all $i\in[m]$ and the requirement that \eqref{eq:alt-s-viability} holds for all $i,i'\in[m]$.
$f$ is $s$-viable if and only if $R=R'$.
Thus, $s$-viability of $f$ is decided by the linear program that checks if $R=R'$.
Feasibility of a linear program can be checked in time polynomial in the number of variables and constraints when each coefficient in the LP can be represented by a constant number of bits.
Let $\ell=\max(|\cY|,\max_i(|\cX_i|)$.
The number of variables is upper bounded by $\ell^{s{k\choose s}}$, and number of constraints is upper bounded by $O(m^2\cdot \ell^{k+s}+m\cdot\ell^k)$ where $m$ is upper bounded by ${k\choose s}$.
Finally, the number of distinct $\cA_1,\ldots,\cA_m$ satisfying $\cap_{i=1}^m\cA_i=\emptyset$ is upper bounded by $2^{{k\choose s}}$.
Thus, the total computation time is upper bounded by $2^{O(s{k \choose s})}$ for any fixed alphabet sizes $\cX_1,\ldots,\cX_n,\cY$.

%% file: appendix-upgrades.tex
\section{The $s=1$ case for any $k$}\label{app:k=1}
\subsubsection{Preliminaries}
In this section, we will set up some preliminaries towards building a specialized protocol for recovering a function. We use the term `sender' and `user' interchangeably since the problem setup only requires the users to send reports to the decoder. It is instructive to refer to \cite{fitzi2004pseudo} and \cite{narayanan2023complete} for a detailed introduction and proofs of the observations given below.  
\begin{definition}\cite{narayanan2023complete}
    For random variables $U \in \cU$ and $V \in \cV$ jointly distributed according to $P_{UV}$, consider the partition of $U$ on alphabet $\cU$ based on the following equivalence relation:
    $$u \sim u' \iff P_{V|U}(v|u) = P_{V|U}(v|u')\ \forall v \in \cV.$$
    Define $\psi_{U \mss V}$ to be the function which maps the elements in $\cU$ to their part in the above partition, and we define $U \mss V = \psi_{U \mss V}(U)$. 
\end{definition}
If a sender observes $U$ and a receiver observes $V$, a malicious sender can send $u'$ if $U = u$ and both $u$ and $u'$ are in the same equivalence class, i.e., $P_{V|U=u} = P_{V|U=u'}$ and the decoder cannot detect this change since the Markov chain $V \mc (U \mss V) \mc U$ holds.
This can be extended for a sequence of i.i.d. variables where the sender observes $U^n$ and the decoder observes $V^n$. The sender can send $\bu^n$ where for every $t \in [n]$, $\bu_t$ and $u_t$ are in the same equivalence class without triggering a detection from the decoder. 
But, conversely, if the sender sends a $\bu^n$ such that $\bu_t$ and $u_t$ are not in the same equivalence class for a lot of values of $t \in [n]$, then the sequence received by the decoder will not be jointly typical with  of the decoder with high probability.
If the decoder is allowed to detect if $U^n$ was manipulated, a typicality test would trigger a detection with high probability. 
Intuitively, $U \mss V$ is the function of $U$ which can be verified by the decoder having side-information $V$. 
The function which can be robustly recovered by the decoder is therefore $(\psi_{U \mss V}(U_1), \dots, \psi_{U \mss V}(U_n))$.
We now present the following lemma which helps us extend this notion to more than one senders. 
\begin{lemma}\cite{fitzi2004pseudo}\label{lem:mss-single}
    \begin{enumerate}
        \item[(i)]     $U \mc (U \mss V) \mc V$
        \item[(ii)]   If $V_1$ is a function of $V_2$, then $U \mss V_1$ is a function of $U \mss V_2$.
    \end{enumerate}
\end{lemma}
 An extension of the above formulation is presented in \cite{narayanan2023complete} for 2 users when at most one of them can be corrupt as in a $(2,1)$-byzantine distributed source coding problem. Let the two users observe the sequences $U^n$ and $V^n$ and the decoder observes a sequence $W^n$; $(U_t,V_t,W_t)_{t \in [n]}$ being sampled i.i.d. from a distribution $P_{UVW}$. From the single-user scheme, it is clear that the decoder can faithfully recover $(U \mss W)$ and $(V \mss W)$ or detect the malicious party. Then, it can use this information to further recover $(W^{(1)}, U \mss W^{(1)}, V \mss W^{(1)})$ where $W^{(1)} = (W, (U \mss W), (V \mss W))$ using the single user scheme. This process of `upgrading' the side-information runs for a finite number of steps, as explained below, and the function (of $U,V,W$) which the decoder can learn after running these upgrading steps will be a robustly recoverable function. 
\begin{definition}[Upgraded variable~\cite{narayanan2023complete}]
    For a triple of random variables $(U,V,W)$ with joint distribution $P_{UVW}$, we define 
    \begin{align*}
        W^{(0)} &= W\\
        W^{(1)} &= \inp{W^{(0)}, \inp{U \mss W^{(0)}}, \inp{V \mss W^{(0)}}}\\
        W^{(2)} &= \inp{W^{(1)}, \inp{U \mss W^{(1)}}, \inp{V \mss W^{(1)}}}\\
        &\vdots\\
        W^{(i+1)} &= \inp{W^{(i)}, \inp{U \mss W^{(i)}}, \inp{V \mss W^{(i)}}}\\
        &\vdots
    \end{align*}
\end{definition}

Function $\psi_{U \mss W^{(i)}}$ induces a partitioning on $\cU$ in the following manner: the partition $\cU = \cU_1 \sqcup \cU_2 \sqcup \dots \cU_{\ell}$ satisfies that for any part of this partition, every element in the part maps to the same element in the co-domain and moreover, no two elements from different parts map to the same element in the co-domain. Since $W^{(i)}$ is a function of $W^{(i-1)}$, therefore, $\psi_{U \mss W^{(i-1)}}$ is a function of $U \mss W^{(i)}$. This means that the elements of the partition induced by $\psi_{U \mss W^{(i)}}$ are subsets of (or equal to) elements of partition induced due to $\psi_{U \mss W^{(i-1)}}$, which in turn means that the number of parts in the partition induced due to $\psi_{U \mss W^{(i)}}$ is at least as many as those in the partition induced due to $\psi_{U \mss W^{(i-1)}}$. Noting that this number cannot exceed $|\cU|$ and a similar set of partitioning exist for $\cV$, we observe that the process of upgrading must saturate after at most $|\cU||\cV|$ steps.

\begin{definition}\cite{narayanan2023complete}\label{def:msstest}
    The functions corresponding to the final partitions are therefore $\psi_{U \mss W^{|\cU||\cV|}}$ and $\psi_{V \mss W^{|\cU||\cV|}}$. We define the maximum upgraded variable $$W^{\ast} = (W, U \mss W^{|\cU||\cV|}, V \mss W^{|\cU||\cV|}).$$
\end{definition}
Let $\mu_{W^{\ast}}(u,v,w) := (u, \psi_{U \mss W^{|\cU||\cV|}}(v), \psi_{V \mss W^{|\cU||\cV|}(w)})$. Hence $\mu_{W^{\ast}}(U,V,W) = W^{\ast}$.

\begin{lemma}\label{lem:mss-properties}\cite{narayanan2023complete}
    For random variables $(U,V,W)$, the following holds:
    \begin{align*}
        &U \mc (U \mss W^{|\cU||\cV|}) \mc (V \mss W^{|\cU||\cV|},W)\\
        &V \mc (V \mss W^{|\cU||\cV|}) \mc (U \mss W^{|\cU||\cV|},W)
    \end{align*}
\end{lemma}
The protocol which takes $P_{UVW}$ and sequences $(u^n,v^n,w^n)$ as inputs which robustly recover $\mu_W(u,v,w)$ is as follows:
\begin{protocol}\label{prot:2usrmss}\cite[Lemma 4]{narayanan2023complete}.
    The decoder $\protdecoder(P_{UVW}, u^n,v^n,w^n)$ is defined below: On receiving the reports $u^n, v^n$ from users 1 and 2, respectively, and side-information $w^n$,
\begin{enumerate}
    \item Fix some parameters $\gamma_0, \dots, \gamma_{|\cU||\cV|} > 0$. 
    \item Let $w^{(0),n} = \{w^{(0)}_t\}_{t \in [n]} = w^n$.
    \item For round index $r = {0,\dots,|\cU||\cV|}$,
    \begin{enumerate}
        \item   check if $({u^n,w^{(r),n}}) \in \cT^n_{\gamma_r}(P_{U,W^{(r)}})$. If not, declare user 1 is corrupt and terminate early.
        \item   check if $(v^n,w^{(r),n}) \in \cT^n_{\gamma_r}(P_{V,W^{(r)}})$. If not, declare user 2 is corrupt and terminate early. 
        \item   if both the aforementioned checks pass, for every $t \in [n]$, assign $$w^{(r+1)}_t = \inp{w^{(r)}_t, \psi_{U \mss W^{(r)}}(u_{t}), \psi_{V \mss W^{(r)}}(v_t)}$$ and update $r$ with $r+1$.
    \end{enumerate}
    \item   Output $w^{(|\cU||\cV|+1), n}$.
\end{enumerate}
\end{protocol}
\begin{lemma}\cite[Lemma 4]{narayanan2023complete}\label{lem:2userprotocol}
    If $(U_t,V_t,W_t)$ are sampled i.i.d. from a distribution $P_{UVW}$ for $t \in [n]$, then for any $\delta>0$, there exist parameters $\gamma_0, \gamma_1, \dots, \gamma_{|\cU||\cV|} > 0$ such that the following holds:
    \begin{enumerate}
        \item
        $\protdecoder$ does not terminate early under the inputs $(U^n,V^n,W^n)\sim P_{UVW}$ i.i.d. with probability at least $1-2^{-\Omega(n)}$.

        \item 
        Suppose $\bU^n \mc U^n \mc (V^n,W^n)$ is a Markov chain, then, on inputs $P_{UVW}$ and $(\bU^n,V^n,W^n)$, $\Pi$ satisfies the following:
        \begin{enumerate}
            \item
            $\protdecoder$ declares user $2$ as corrupt with probability at most $2^{-\Omega(n)}$.
            \item
            For $r = |\cU||\cV|$, 
            \begin{multline*}
            \pr \Biggl(
                \left(
                    \ham \inp{\psi_{U \mss W^{(r)}}(U^n), \psi_{U \mss W^{(r)}}(\bU^n)} > \delta
                \right)
            \Biggr.
            \\
            \Biggl. \wedge
                \left(
                    \Pi \text{ does not terminate early}                    
                \right)
            \Biggr) \leq 2^{-\Omega(n)}.
            \end{multline*}
        \end{enumerate}
        \item 
        Suppose $\overline{V}^n \mc V^n \mc (U^n,W^n)$ is a Markov chain, then, on inputs $P_{UVW}$ and $(U^n,\overline{V}^n,W^n)$, $\Pi$ satisfies the following:
        \begin{enumerate}
            \item
            $\protdecoder$ declares user $1$ as corrupt with probability at most $2^{-\Omega(n)}$.
            \item
            For $r = |\cU||\cV|$, 
            \begin{multline*}
            \pr \Biggl(
                \left(
                    \ham \inp{\psi_{V \mss W^{(r)}}(V^n), \psi_{V \mss W^{(r)}}(\overline{V}^n)} > \delta
                \right)
            \Biggr.
            \\
            \Biggl. \wedge
                \left(
                    \Pi \text{ does not terminate early}                    
                \right)
            \Biggr) \leq 2^{-\Omega(n)}.
            \end{multline*}
        \end{enumerate}
    \end{enumerate}
\end{lemma}
\subsection{The$(k=2,s=1)$ case}
To maintain consistency across the special case of $(k=2, s=1)$ and the general case (any $k$, $s=1$), we change the notation by making the following assignments to the variables in the definition given above: $\cU \leftarrow \cX_1$, $\cV \leftarrow \cX_2$, $\cW \leftarrow \cY$, $U \leftarrow X_1$, $V \leftarrow X_2$, $W \leftarrow Y$ and $W^{\ast} \leftarrow Y^{\ast}$.
It is shown in \cite{narayanan2023complete} that the upgraded random variable $Y^{\ast}$ at the receiver can be robustly recovered. 
In this section, we connect the formulation given in \cite{narayanan2023complete} to the class of functions which are 1-viable. We show that given a distribution $P_{X_1 X_2 Y}$ the class of 1-viable functions is the same as all functions of the upgraded variable $Y^{\ast}$ in Theorem~\ref{thm:2usermss}.

\begin{theorem}\label{thm:2usermss}
    Given a distribution $P_{X_1 X_2 Y}$, $f$ is $1$-viable if and only if there exists a function $h:\cX_1 \times \cX_2 \times \cY \rightarrow \cZ$ such that $h(Y^{\ast}) = f(X_1, X_2, Y)$.
\end{theorem}
\begin{proof}
    To show that if $f(X_1,X_2,Y)$ is a function of $Y^{\ast}$, then $f$ is $1$-viable, 
    it suffices to show that $Y^{\ast}$ is $1$-viable. This is true since $Y^{\ast} = (Y, X_1 \mss Y^{|\cX_1||\cX_2|}, X_2 \mss Y^{|\cX_1||\cX_2|})$ is recoverable using Protocol~\ref{prot:2usrmss}, as shown in \cite[Lemma 4]{narayanan2023complete} and every recoverable function is $1$-viable using Theorem~\ref{thm:2user}.
    In Lemma~\ref{lem:2usermss} below, we show the converse, i.e., that if $f$ is $1$-viable, then, $f$ is a function of $Y^{\ast}$.
\end{proof}
\begin{lemma}\label{lem:2usermss}
    For a distribution $P_{X_1 X_2 Y}$, if $f$ is $1$-viable, then, there exists a function $h$ such that $h(Y^{\ast}) = f(X_1, X_2, Y)$, i.e., $f(x_1,x_2,y) = \mu_{Y^{\ast}}(x_1,x_2,y)$ for all $(x_1,x_2,y)$.
\end{lemma}
\begin{proof}
    We denote the functions $\psi_{X_1 \mss Y^{|\cX_1||\cX_2|}}$ by $\psi_{X_1}^{\ast}$ and $\psi_{X_2 \mss Y^{|\cX_1||\cX_2|}}$ by $\psi_{X_2}^{\ast}$. Furthermore, we denote $\psi^{\ast}_{X_1}(X_1)$ by $\Xdag_1$ and $\psi^{\ast}_{X_2}(X_2)$ by $\Xdag_2$.

    We will prove the contrapositive of the lemma. We will show that if 
    $f$ is such that there is no function $h$ s.t. $h(\mu_{Y^{\ast}}(x_1,x_2,y)) = f(x_1,x_2,y)$, then $f$ is not $1$-viable. Suppose there is no such $h$ for the function $f$, then there exists some pair of triples $(x_1,x_2,y), (x_1',x_2',y)$, both in the support of $P_{X_1 X_2 Y}$ such that $f(x_1', x_2', y) \neq f(x_1, x_2, y)$ but 
    \begin{equation}
        (y, \psi^{\ast}_{X_1}(x_1), \psi^{\ast}_{X_2}(x_2)) = (y, \psi^{\ast}_{X_1}(x'_1), \psi^{\ast}_{X_2}(x_2')).
    \end{equation}
    We now prove that $f$ is not $1$-viable, i.e., there exists a distribution $Q_{\uX_1 \tX_1 \uX_2 \tX_2 \uY}$ which satisfies 
    \begin{align}
        \label{eq:2user-mss-def1}
        Q_{\uX_1 \tX_2 \uY} &= P_{X_1 X_2 Y} \text{ and } \uX_1 \mc \tX_1 \mc (\uX_2, \uY),\\
        \label{eq:2user-mss-def2}
        Q_{\tX_1 \uX_2 \uY} &= P_{X_1 X_2 Y} \text{ and }\uX_2 \mc \tX_2 \mc (\uX_1, \uY),
    \end{align}
    but $f(\uX_1, \tX_2, \uY) \neq f(\tX_1, \uX_2, \uY)$ with some non-zero probability under $Q_{\uX_1 \tX_1 \uX_2 \tX_2 \uY}$.
    Towards this, note that the joint distribution $P_{X_1 X_2 Y}$ induces a distribution on $(X_1, X_2, \Xdag_1, \Xdag_2, Y)$, which is, 
    \begin{multline*}
    P_{X_1 X_2 \Xdag_1 \Xdag_2 Y}(\xdag_1, \xdag_2, y) = P_{X_1 X_2 Y}(x_1, x_2, y) \\ \times \bb{1}_{\{\psi^{\ast}_{X_1}(x_1) = \xdag_1\}} \times \bb{1}_{ \{ \psi^{\ast}_{X_2}(x_2) = \xdag_2 \} } .
    \end{multline*}
    Consider the joint distribution 
    \begin{align*}
        \MoveEqLeft[2] Q_{\Xdag_1 \Xdag_2 \uY \uX_1 \tX_1 \uX_2 \tX_2}(\xdag_1, \xdag_2, \uy,\ux_1,\tx_1,\ux_2,\tx_2)\\ 
        =& P_{\Xdag_1 \Xdag_2 Y}(\xdag_1, \xdag_2, \uy) P_{X_1|\Xdag_1}(\ux_1|\xdag_1) P_{X_2|\Xdag_2}(\ux_2|\xdag_2)\\
        &P_{X_1|\Xdag_1 X_2 Y}(\tx_1|\xdag_1, \ux_2, \uy) P_{X_2|\Xdag_2 X_1 Y}(\tx_2|\xdag_2,\ux_1,\uy). \addtocounter{equation}{1} \label{eq:21joint} \tag{\theequation}
    \end{align*}
    We first show that the aforementioned distribution satisfies (\ref{eq:2user-mss-def1}) and (\ref{eq:2user-mss-def2}) and then show that $f(\uX_1, \tX_2, \uY) \neq f(\tX_1, \uX_2, \uY)$ with non-zero probability under this distribution.
    Note that $\uX_1$ is independent of $(\uX_2, \uY)$ conditioned on $\Xdag_1$. Here, $\uX_1 \mc \Xdag_1 \mc (\uX_2,\uY)$. Furthermore, $\Xdag_1 = \psi_{\uX_1}^{\ast}(\tX_1)$ by virtue of the fact that $Q_{\Xdag_1 \Xdag_2 \uY \uX_2 \tX_1}(\xdag_1, \xdag_2, \uy, \ux_2, \tx_1) = P_{\Xdag_1 \Xdag_2 Y X_2 X_1}(\xdag_1, \xdag_2, \uy, \ux_2, \tx_1)$ for all $(\xdag_1, \xdag_2, \uy, \ux_2, \tx_1)$ and hence, $Q_{\Xdag_1 X_1} = P_{\Xdag_1 X_1}$. Hence, $\uX_1 \mc \tX_1 \mc (\uX_2, \uY)$, the Markov chain given in \eqref{eq:2user-mss-def1} holds. The Markov chain in \eqref{eq:2user-mss-def2} are shown in a similar way. We now show the condition $Q_{\uX_1 \tX_2 \uY} = P_{X_1 X_2 Y}$ from \eqref{eq:2user-mss-def1} (and the similar condition in \eqref{eq:2user-mss-def2}).
    \begin{align*}
    \MoveEqLeft[2] Q_{\uX_1 \tX_2 \uY}(\ux_1, \tx_2, \uy)\\
    =&\sum_{\xdag_1, \xdag_2, \tx_1, \ux_2} P_{\Xdag_1 \Xdag_2 Y}(\xdag_1, \xdag_2, \uy) P_{X_1|\Xdag_1}(\ux_1|\xdag_1) P_{X_2|\Xdag_2}(\ux_2|\xdag_2)\\
    &P_{X_1|\Xdag_1 X_2 Y}(\tx_1|\xdag_1, \ux_2, \uy) P_{X_2|\Xdag_2 X_1 Y}(\tx_2|\xdag_2,\ux_1,\uy)\\
    \stackrel{(a)}{=}& \sum_{\xdag_1, \xdag_2} P_{\Xdag_1 \Xdag_2 Y}(\xdag_1, \xdag_2, \uy)P_{X_1|\Xdag_1}(\ux_1|\xdag_1)\\
    & \times P_{X_2|\Xdag_2 X_1 Y}(\tx_2|\xdag_2,\ux_1,\uy)\\
    \stackrel{(b)}{=}& \sum_{\xdag_1, \xdag_2} P_{\Xdag_1 \Xdag_2 Y}(\xdag_1, \xdag_2, \uy)P_{X_1|\Xdag_1 \Xdag_2 Y}(\ux_1|\xdag_1,\xdag_2,\uy)\\    
    & \times P_{X_2|\Xdag_2 X_1 Y}(\tx_2|\xdag_2,\ux_1,\uy)\\
    =&\sum_{\xdag_1, \xdag_2} P_{X_1 \Xdag_1 \Xdag_2 Y}(\ux_1, \xdag_1, \xdag_2, \uy) P_{X_2|\Xdag_2 X_1 Y}(\tx_2|\xdag_2,\ux_1,\uy)\\
    =&\sum_{\xdag_2}P_{X_1 \Xdag_2 Y}(\ux_1, \xdag_2, \uy) P_{X_2|\Xdag_2 X_1 Y}(\tx_2|\xdag_2,\ux_1,\uy)\\
    =&P_{X_1 X_2 Y}(\ux_1,\tx_2,\uy),
    \end{align*}
    where (a) follows from marginalizing out $\tx_1$ and then $\ux_2$, (b) is due to Lemma~\ref{lem:mss-properties}. To complete the proof, we show that $Q_{\uX_1 \tX_1\uX_2 \tX_2 \uY}(\ux_1, \tx_1, \ux_2, \tx_2, \uy)>0$ and $f(\tx_1,\ux_2,\uy) \neq f(\ux_1, \tx_2, \uy)$ for the point $(\ux_1,\tx_1,\ux_2,\tx_2,\uy,\xdag_1,\xdag_2) = (x_1',x_1,x_2,x_2',y,\psi^{\ast}_{X_1}(x_1),\psi^{\ast}_{X_2}(x_2))$. 
    Note that by assumption $f(x_1,x_2,y) \neq f(x_1',x_2',y)$. Therefore it remains to show that
    \begin{align*}
        Q_{\uX_1 \tX_1 \uX_2 \tX_2 \uY \Xdag_1 \Xdag_2}(x'_1,x_1,x_2,x'_2,y,\psi_{X_1}^{\ast}(x_1),\psi^{\ast}_{X_2}(x_2))>0,
    \end{align*} which we argue in the following steps and therefore completing the proof. 
    \begin{enumerate}
        \item   Note that $P_{\Xdag_1\Xdag_2Y}(\psi^{\ast}_{X_1}(x_1),\psi^{\ast}_2(x_2),y)>0$ since $P_{X_1 X_2 Y}(x_1,x_2,y)>0$.
        \item   $P_{X_1|\Xdag_1}(x_1'|\psi^{\ast}_{X_1}(x_1)) \stackrel{(a)}{=} P_{X_1|\Xdag_1}(x_1'|\psi^{\ast}_{X_1}(x_1')) \stackrel{(b)}{>} 0$, where (a) is due to the fact that $\psi^{\ast}_{X_1}(x_1') = \psi^{\ast}_{X_1}(x_1)$ and (b) is due to the fact that $P_{X_1 \Xdag_1}(x_1',\psi_{X_1}^{\ast}(x_1'))>0$, which, in turn is true since $P_{X_1}(x_1)>0$ and $\Xdag_1 = \psi^{\ast}_{X_1}(X_1)$. Similarly, $P_{X_2|\Xdag_2}(x_2|\psi^{\ast}_{X_2}(x_2))>0$.
        \item   Since $P_{X_1 X_2 Y}(x_1,x_2,y)>0$ and $P_{\Xdag_1|X_1}(\psi^{\ast}_{X_1}(x_1)|x_1)=1, P_{X_1 X_2 Y \Xdag_1}(x_1,x_2,y,\psi^{\ast}_{X_1}(x_1))>0$, which gives that $P_{X_1|\Xdag_1 X_2 Y}(x_1|\psi^{\ast}_{X_1}(x_1),x_2,y)>0$. Similarly, $P_{X_2|\Xdag_2 X_1 Y}(x_2'|\psi^{\ast}_{X_2}(x_2),x_1',y)>0$ noting that $\psi^{\ast}_{X_2}(x_2) = \psi^{\ast}_{X_2}(x_2')$.
    \end{enumerate}
\end{proof}

    \subsection{The general $s=1$ case:}
    We now extend the characterization given in Theorem~\ref{thm:2usermss} to $(k,1)$-byzantine distributed source coding problems.
    Consider the scenario with $k$ users having inputs $X^n_1,X^n_2,\dots,X^n_k$ and the decoder's side-information is $Y^n$, distributed according to $P_{X_{[k]}Y}$ i.i.d. In order to write the characterization for the general case, we define the following variables.
    For $i,j \in [k], i \neq j$, define $Y^{(0)}_{ij} = (Y,X_{\{i,j\}^c})$
    \begin{align*}
        Y^{(1)}_{ij}  &= \inp{Y_{ij}^{(0)}, X_i \mss Y_{ij}^{(0)}, X_j \mss Y_{ij}^{(0)}}\\
                      &\vdots\\
        Y^{(r+1)}_{ij}&= \inp{Y_{ij}^{(r)}, X_i \mss Y_{ij}^{(r)}, X_j \mss Y_{ij}^{(r)}}\\
                      &\vdots
    \end{align*}
    where $\{i,j\}^c = [k] \setminus \{i,j\}$.
    Note that these variables are obtained by invoking Definition~\ref{def:msstest} with the following assignments to variables:
    $U \leftarrow X_{i}, V \leftarrow X_j, W\leftarrow (Y, X_{\{i,j\}^c}), Y_{ij}^{\ast} \leftarrow W^{\ast}$. Simply stated, $(Y^{(r)}_{ij})$ is the $r$th step upgrade of the side-information, but here, instead of using $Y$ as the side information, we will be using $(Y,X_{\{i,j\}^c})$.
    Analogous to the $(2,1)$-case, define
    $$Y^{\ast}_{ij} = \inp{Y, X_i \mss Y^{|\cX_i||\cX_j|}_{ij}, X_j \mss Y_{ij}^{|\cX_i||\cX_j|}}.$$
    Consider a set of functions 
    \begin{multline*}
    \mathcal{G} = \left\{\{h_{ij}\}_{i,j \in [k], i \neq j}: \forall i \neq j, i' \neq j',\right. \\ \left. \{i,j\} \neq \{i',j'\}, h_{ij}(Y^{\ast}_{ij}) = h_{i'j'}(Y^{\ast}_{i'j'})\right\}
    \end{multline*}
    The \emph{common upgraded variable} is $G^{\ast} = h^{\ast}_{ij}(Y^{\ast}_{ij})$ where
    $$h_{ij}^{\ast} = \arg \max_{\mathcal{G}} H(h_{ij}(Y^{\ast}_{ij})).$$
    This can be thought of as the (multi-user) G\'acs-K\"orner common variable \cite{gacskorner} of the random variables $Y^{\ast}_{ij}, i,j \in [k], i \neq j$.  
    We show that the common upgraded variable is recoverable and every recoverable function is a function of the common upgraded variable.  
    \begin{theorem}
        For a distribution $P_{X_{[k]}Y}$, $f$ is 1-robustly recoverable if and only if there exist functions $\{h_{ij}\}_{i,j\in[k],i\neq j}$ such that $h_{ij}(Y^{\ast}_{ij}) = f(X_{[k]},Y)$.
    \end{theorem}
    \begin{proof}
        We prove the following claim, which when combined with Theorem~\ref{thm:kuser-main} gives the only if direction of the theorem.
    \begin{claim}
        If $f$ is $1$-viable, then, there exist functions $\{h_{ij}\}_{i,j\in[k],i\neq j}$ such that $h_{ij}(Y^{\ast}_{ij}) = f(X_{[k]},Y)$.
    \end{claim}
    \begin{proof}
        If $f$ is $1$-viable, then by definition, the following holds:
        (notice that in Definition~\ref{def:s-viable}, the sets $\cA_1, \dots, \cA_m$ are singletons or empty, since there is no ambiguity, we denote $\tX_i$ by $\tX_i$)
        For any $\cI \subset \{1, \dots, k\}$, $|\cI| \geq 2$, for any $Q_{\uX_{[k]},\uY,\inp{\tX_i}_{i \in \cI}}$, if
        $Q_{\tX_i, \uX_{\{i\}^c}, \uY} = P_{X_{[k]}Y}$ for every $i \in \cI$ and $\uX_i \mc X_i \mc (\uX_{ \{i\}^c }, \uY)$ for every $i \in \cI$, then, $f(\tX_i^i, \uX_{\{i\}^c}, \uY) = f(\tX_j^j, \uX_{\{j\}^c}, \uY)$ for every $i,j \in \cI$.
        This is true in particular for every $\cI$ with $|\cI|=2$. 
        To complete the proof, we now invoke Lemma~\ref{lem:2usermss} for every such $\cI = \{i,j\}$, $i \neq j$, with the following variables: $X_1 \leftarrow X_i$, $X_2 \leftarrow X_j$, $Y\leftarrow (Y,X_{\{i,j\}^c})$ to obtain functions $h_{ij}$ such that $h_{ij}(Y^{\ast}_{ij}) = f(X_{[k]},Y)$.
    \end{proof}    
    To show the if direction, we argue that the scheme presented in Protocol~\ref{prot:kuserprot} $1$-robustly recovers $f$.

    \begin{protocol} \label{prot:kuserprot} On receiving the sequences $\ux^n_1, \ux_2^n, \dots, \ux^n_{k}$ from the users and side-information $\uy^n$, the decoder runs $\protdecoderr$ which is described as follows: Fix a set $S = \varnothing$ and performs the following steps:
    \begin{enumerate}
        \item 
        While $|S| \leq k-2$
        \begin{enumerate}[label=(\alph*)]
            \item 
            Pick any two users $i,j \in [k] \setminus S$.
            \item 
            Compute 
            \begin{multline*} \mathsf{o/p}_{i,j} \leftarrow \protdecoder(P_{X_i, X_j, (Y, X_{ \{i,j\}^c })},\\ \ux_i^n, \ux_j^n, (\uy^n, \ux^n_{ \{i,j\}^c }))
            \end{multline*}
            [decoder for $(s,k) = (2,1)$]
            \item 
            If $\mathsf{o/p}_{i,j} \not\in \{i,j\}$, output $h_{ij}(\mathsf{o/p}_{i,j})$ and terminate; else if $\mathsf{o/p}_{i,j} = i$, then, $S \leftarrow S \cup \{j\}$; else if $\mathsf{o/p}_{i,j} = j$, then, $S \leftarrow S \cup \{i\}$.
            \item 
            If $|S| = k-2$, then output $\mathsf{o/p}_{i,j}$ (i.e., declare $\mathsf{o/p}_{i,j}$ as malicious) and terminate.
        \end{enumerate}
    \end{enumerate}
    \end{protocol}
    Consider a sequence $(X_{1,t}, \dots, X_{k,t}, Y_t)$ sampled i.i.d. from $P_{X_{[k]}Y}$ for every $t \in [n]$. There are two possible cases of corruption which are considered as follows:\\
    The decoder receives from users $(\uX^n_{[k]},Y^n)$. 
    The decoder chooses any two users $i,j \in [k]$, assigns $(Y^n, \uX^n_{\{i,j\}^c})$ as the side information and runs $\protdecoder$ with $P_{U,V,W} = P_{X_i, X_j, (Y, X_{ \{i,j\}^c })}$, $u^n = \uX_i^n$, $v^n = \uX_j^n$.
    If all users are honest, then $\uX^n_i = X^n_i$ for every $i \in [k]$. By part (1) and sub-part (b) of part (2) of Lemma~\ref{lem:2userprotocol}, if $\protdecoder$ doesn't terminate, then it outputs $Y^{\ast}_{ij}$ w.h.p. and therefore $\protdecoderr$ outputs $G^{\ast} = h_{ij}(Y^{\ast}_{ij})$ and terminates. 
    Note that $\protdecoderr$ outputs $G^{\ast}$ if for any $i,j \in [k]$, $\protdecoder$ outputs $Y^{\ast}_{ij}$.
    Now, we consider the case where, for every $i,j \in [k]$, $\protdecoder$ terminates before outputting $Y_{ij}^{\ast}$. Whenever $\protdecoder$ terminates before outputting $Y^{\ast}_{ij}$, it means that either the corrupt party is not in $\{i,j\}$ and caused $\protdecoder$ to implicate one of $\{i,j\}$ due to corrupt side-information, or, all users providing the side-information are honest and the implicated party is the corrupt one. 
    In either case, the party in $\{i,j\}$ who was not implicated can be exonerated, and therefore can be added to the set $S$, which can be thought of as the set of `trusted users'. 
    In this way, every time $\protdecoder$ doesn't output $Y^{\ast}_{ij}$, a new party is exonerated and is added to the set $S$. When $|S| = k-2$, the only two parties which could be corrupt are in $[k] \setminus S$. When the last run of $\protdecoder$ also doesn't output $Y^{\ast}_{ij}$, the user who is blamed is therefore the corrupt party with high probability.
    \end{proof}

%% file: appendix-example.tex
\section{Proof of Claim~\ref{cl:(3,2)-example}}
\label{app:proof-for-(3,2)-example}

The second statement will follow from the first and the argument in Section~\ref{sec:example} that $W$ cannot be recovered robustly. Since $V$ is a function of $Y$, it suffices to show that $U$ is 2-robustly recoverable. Towards this, we show that for any distribution of the form $Q_{\uX_1 \uX_2 \uX_3 \uY \txa_2 \txa_3 \txb_1 \txb_3 \txc_1 \txc_2}$ such that
\begin{align}
    \label{eq:23}
    (\uX_2,\uX_3) &\mc (\txa_2,\txa_3) \mc (\uX_1,\uY),\\ 
    \label{eq:31}
    (\uX_1,\uX_3) &\mc (\txb_1,\txb_3) \mc (\uX_2,\uY),\\
    \label{eq:12}
    (\uX_1,\uX_2) &\mc (\txc_1,\txc_2) \mc (\uX_3,\uY),\\
    \label{eq:marg}
    P_{X_1X_2X_3Y} &= Q_{\txc_1\txc_2\uX_3\uY} = Q_{\uX_1\txa_2\txa_3\uY} = Q_{\txb_1\uX_2\txb_3\uY}
\end{align}
$\uX_1 = \txb_1 = \txc_1$ a.s.
We use the following claim which formalizes the intuition that when users 1 and 2 collude together, the reports given by user 1 and user 2 should be same as their respective observation and an analogous statement for collusion between user 1 and user 3. Formally,
\begin{claim}\label{claim:example-subclaim-red}
    \begin{enumerate}[label=(\alph*)]
        \item   $Q_{\uX_1 \uX_2| \txc_1 \txc_2}(x, \eras_2|x,\eras_2) = 1$ for $x \in \{0,1\}$,
        \item   $Q_{\uX_1\uX_2|\txc_1\txc_2}(x,x|x,x)=1$ for $x \in \{0,1\}$,
        \item   $Q_{\uX_1 \uX_3| \txb_1 \txb_3}(x, \eras_3|x,\eras_3) = 1$ for $x \in \{0,1\}$ and
        \item   $Q_{\uX_1\uX_3|\txb_1\txb_3}(x,x|x,x)=1$ for $x \in \{0,1\}$.
    \end{enumerate} 
\end{claim}
 $Q_{\uX_1|\txc_1}(x|x) = \sum_{x_2} Q_{\uX_1|\txc_1 \txc_2}(x|x,x_2) Q_{\txc_2|\txc_1}(x_2|x) = \sum_{x_2} Q_{\txc_2|\txc_1}(x_2|x) = 1$, where the second equality is due to Claim~\ref{claim:example-subclaim-red}. Similarly, $Q_{\uX_1|\txc_1}(x|x) = 1$. It only remains to prove Claim~\ref{claim:example-subclaim-red} to complete the proof. 

\begin{proof}[Proof of Claim~\ref{claim:example-subclaim-red}]
    We will first show that 
    \begin{equation} \label{eq:example-sc1} Q_{\uX_1 \uX_2| \txc_1 \txc_2}(b,\eras_2|x,x)=0 \text{ for } x,b \in \{0,1\}.\end{equation}
    This formalizes the intuition that if user 2 observes a bit, it cannot report $\eras_2$.
    Suppose not, then, $Q_{\uX_1 \uX_2|\txc_1 \txc_2}(b,\eras_2|x,x) > 0$ for $x,b \in \{0,1\}$. Noting that $Q_{\uX_1 \uX_2 \uX_3 \uY}(b,\eras_2,x,x) \geq Q_{\txc_1 \txc_2 \uX_3 \uY}(x,x,x,x)Q_{\uX_1 \uX_2|\txc_1 \txc_2}(b,\eras_2|x,x) = P_{X_1 X_2 X_3 Y}(x,x,x,x)Q_{\uX_1 \uX_2|\txc_1 \txc_2}(b,\eras_2|x,x)>0$, we have $Q_{\uX_2 \uY}(\eras_2,x)>0$. But, $P_{X_2 Y}(\eras_2,x) = Q_{\uX_2 \uY}(\eras_2,x) = 0$ due to \eqref{eq:marg}. This gives a contradiction. We will now show that 
    \begin{equation} \label{eq:example-sc2} Q_{\uX_2|\txc_1 \txc_2}(\eras_2|x,\eras_2) = 1 \text{ for } x \in \{0,1\}.\end{equation}
    This claim formalizes that user 2 cannot report $\eras_2$ if it received a bit when colluding with user 1.
    \begin{align*}
    \MoveEqLeft[2]\bb{P}[\uX_2 = \eras_2]\\
    =&\sum_{
        x_3,y
    }
    Q_{\uX_2 \uX_3 \uY}(\eras_2,x_3,y)\\
    =& \sum_{x_3,y} \sum_{x_1,x_2} \Pabc(x_1,x_2,x_3,y) Q_{\uX_2|\tX^3_1 \tX^3_2}(\eras_2|x_1,x_2)\\
    \stackrel{(a)}{=}& \sum_{x_1,x_3,y}
        \Pabc(x_1,x_1,x_3,y) Q_{\uX_2|\txc_1 \txc_2}(\eras_2|x_1,x_1)\\
    &+ \sum_{x_1,x_3,y}
        \Pabc(x_1,\eras_2, x_1,\eras) Q_{\uX_2|\txc_1 \txc_2}(\eras_2|x_1,\eras_2)\\
        \stackrel{(b)}{=}& \sum_{x_1 \in \{0,1\}} \Pabc(x_1,\eras_2, x_1,\eras) Q_{\uX_2|\txc_1 \txc_2}(\eras_2|x_1,\eras_2)\\
        \label{eq:sclm1p3}
        =&(1/8) Q_{\uX_2|\txc_1\txc_2}(\eras_2|0,\eras_2) + (1/8) Q_{\uX_2|\txc_1\txc_2}(\eras_2|1,\eras_2),
        \end{align*}
        where in (a), we consider the $x_2 \neq \eras_2$ and $x_2 = \eras_2$ cases separately and note that when $X_1 = x_1$ and $X_2 \neq \eras_2$, then, $X_2 = x_1$, and (b) is because of \eqref{eq:example-sc1}. From \eqref{eq:marg}, $\bb{P}[\uX_2 = \eras_2] = P_{X_2}(\eras_2) = 1/4$, which is only possible when $Q_{\uX_2|\txc_1 \txc_3}(\eras_2|x,\eras_2)=1$ for $x \in \{0,1\}$, completing the argument for \eqref{eq:example-sc2}. 

        We will now show part (a) of the claim. For the sake of contradiction, assume that $Q_{\uX_1 \uX_2| \txc_1 \txc_2}(1-x,\eras_2|x,\eras_2)>0$ for some $x \in \{0,1\}$ (we are using the fact that from \eqref{eq:example-sc2}, $\uX_2 = \eras_2$). Then, we will have that
        \begin{align*}
            \MoveEqLeft[2] Q_{\uX_1 \uX_2 \uX_3 \uY}(1-x,\eras_2,x,\eras)\\
            \geq& P_{X_1 X_2 X_3 Y}(x,\eras_2,x,\eras) Q_{\uX_1 \uX_2| \txc_1 \txc_2}(1-x,\eras_2|x,\eras_2)>0.
        \end{align*}
        Note that the same view must also be generated by the scenario when the adversarial set is users 1 and 3, i.e., there is some $z \in \{0,1\}$ such that $P_{X_1 X_2 X_3 Y}(z,\eras_2,z,\eras)Q_{X_1 X_3|\tX_1^2 \tX_3^2}(1-x,x|z,z)>0$ and 
        \begin{equation} \label{eq:example-sc3} 
        Q_{\uX_1 \uX_3|\txb_1 \txb_3}(1-x,x|z,z)>0. 
        \end{equation}
        Towards the contradiction, will show that \eqref{eq:example-sc3} is false.
        Since conditioned on their observations, the adversarial users 1 and 3 generate their report conditionally independent of user 2 and decoder's observations, and $P_{X_1 X_2 X_3 Y}(z,z,z,z)>0$ and therefore $Q_{\uX_1 \uX_2 \uX_3 \uY}(1-x,z,x,z)>0$.
        Then, it must be the case that 
        \begin{itemize}
            \item this view is also generated by the scenario when the adversarial users are 2 and 3, i.e., for some $x_2 \in \{0,1,\eras_2\}$ and $x_3 \in \{0,1,\eras_3\}$, $P_{X_1 X_2 X_3 Y}(1-x,x_2,x_3,z)>0$ and $Q_{\uX_2 \uX_3 | \txa_2,\txa_3}(z,x|x_2,x_3)>0$, which implies that $P_{X_1 Y}(1-x,z)>0$, i.e., $z = 1-x$.
            \item this view is also generated by the scenario when the adversarial users are 1 and 2, i.e., for some $x_1' \in \{0,1\}$, $x_2' \in \{0,1,\eras_2\}$, 
            $P_{X_1 X_2 X_3 Y}(x_1',x_2',x,z)>0$ and $Q_{\uX_1 \uX_2|\txc_1 \txc_2}(1-x,z|x_1',x_2')>0$, which implies that $P_{X_3 Y}(x,z)>0$, i.e., $z=x$;
        \end{itemize} 
        Thus, we have a contradiction and it completes the argument for part (a) of the claim.

        We now show that part (b) of the claim is true, i.e., if users 1 and 2 collude and  they both observe a bit, they must both report the bit correctly. The argument above (in which we proved (\ref{eq:example-sc3}) cannot be true) shows that $Q_{\uX_1 \uX_3| \txb_1 \txb_3}(1-x,x|z,z)=0$ for $x,z \in \{0,1\}$. By symmetry, we also have that
        \begin{equation}
            \label{eq:example-sc4}
            Q_{\uX_1 \uX_2|\txc_1 \txc_2}(1-x,x|z,z) = 0 \text{ for } x,z \in \{0,1\},
        \end{equation}
        which means that they cannot output complementary bits as output.
        Further, note that both of them together cannot output the negation of their observations either, i.e., 
        \begin{equation}\label{eq:example-sc5}
        Q_{\uX_1 \uX_2 | \txc_1 \txc_2}(1-z,1-z|z,z) = 0,
        \end{equation} because otherwise, it must be the case that $Q_{\uX_1 \uX_2 \uX_3 \uY}(1-z,1-z,z,z)>0$ and therefore this must also be a view under adversarial users 2 and 3, in which case $\uX_1$ is unaltered by an honest user 1, which makes $P_{X_1 Y}(1-z,z)>0$, which is impossible. Part (b) of the claim is implied by \eqref{eq:example-sc5}, \eqref{eq:example-sc4} and \eqref{eq:example-sc1}.
        
\end{proof}

%% file: paper.bbl
\begin{thebibliography}{10}

\bibitem{SlepianW73}
D.~Slepian and J.~Wolf, ``Noiseless coding of correlated information sources,''
  {\em IEEE Transactions on Information Theory}, vol.~19, no.~4, pp.~471--480,
  1973.

\bibitem{dragotti2009distributed}
P.~L. Dragotti and M.~Gastpar, eds., {\em Distributed Source Coding: Theory,
  Algorithms and Applications}.
\newblock Academic Press, 2009.

\bibitem{KosutTTransIT08}
O.~Kosut and L.~Tong, ``Distributed source coding in the presence of byzantine
  sensors,'' {\em IEEE Transactions on Information Theory}, vol.~54, no.~6,
  pp.~2550--2565, 2008.

\bibitem{KosutTISIT08}
O.~Kosut and L.~Tong, ``{The Byzantine CEO Problem},'' in {\em 2008 IEEE
  International Symposium on Information Theory}, pp.~46--50, 2008.

\bibitem{KosutTAllerton08}
O.~Kosut and L.~Tong, ``A characterization of the error exponent for the
  byzantine {CEO} problem,'' in {\em 2008 46th Annual Allerton Conference on
  Communication, Control, and Computing}, pp.~1207--1214, 2008.

\bibitem{KosutTISIT09}
O.~Kosut and L.~Tong, ``The quadratic {G}aussian {CEO} problem with byzantine
  agents,'' in {\em 2009 IEEE International Symposium on Information Theory},
  pp.~1145--1149, 2009.

\bibitem{fitzi2004pseudo}
M.~Fitzi, S.~Wolf, and J.~Wullschleger, ``Pseudo-signatures, broadcast, and
  multi-party computation from correlated randomness,'' in {\em Advances in
  Cryptology--CRYPTO 2004}, pp.~562--578, Springer, 2004.

\bibitem{narayanan2023complete}
V.~Narayanan, V.~M. Prabhakaran, N.~Sangwan, and S.~Watanabe, ``Complete
  characterization of broadcast and pseudo-signatures from correlations,'' in
  {\em Advances in Cryptology – EUROCRYPT 2023}, pp.~563--593, Springer,
  2023.
\newblock Full version: https://ia.cr/2023/233.

\bibitem{Jaggi7}
S.~Jaggi, M.~Langberg, S.~Katti, T.~Ho, D.~Katabi, and M.~M{\'e}dard,
  ``{Resilient} network coding in the presence of byzantine adversaries,'' in
  {\em IEEE INFOCOM 2007-26th IEEE International Conference on Computer
  Communications}, pp.~616--624, IEEE, 2007.

\bibitem{WangSK10}
D.~Wang, D.~Silva, and F.~R. Kschischang, ``{Robust} network coding in the
  presence of untrusted nodes,'' {\em IEEE Transactions on Information Theory},
  vol.~56, no.~9, pp.~4532--4538, 2010.

\bibitem{Yener}
X.~He and A.~Yener, ``{Strong} secrecy and reliable byzantine detection in the
  presence of an untrusted relay,'' {\em IEEE Transactions on Information
  Theory}, vol.~59, no.~1, pp.~177--192, 2012.

\bibitem{KTong}
O.~Kosut, L.~Tong, and N.~David, ``{Polytope} codes against adversaries in
  networks,'' {\em IEEE Transactions on Information Theory}, vol.~60, no.~6,
  pp.~3308--3344, 2014.

\bibitem{FanKW18}
X.~Fan, O.~Kosut, and A.~B. Wagner, ``Variable packet-error coding,'' {\em IEEE
  Transactions on Information Theory}, vol.~64, no.~3, pp.~1530--1547, 2018.

\bibitem{SangwanNP22}
N.~Sangwan, V.~Narayanan, and V.~M. Prabhakaran, ``Byzantine consensus over
  broadcast channels,'' in {\em 2022 IEEE International Symposium on
  Information Theory (ISIT)}, pp.~1157--1162, 2022.

\bibitem{NehaBDP23}
N.~Sangwan, M.~Bakshi, B.~K. Dey, and V.~M. Prabhakaran, ``Byzantine multiple
  access channels—{Part I}: Reliable communication,'' {\em IEEE Transactions
  on Information Theory}, vol.~70, no.~4, pp.~2309--2366, 2024.

\bibitem{NehaBDP25}
N.~Sangwan, M.~Bakshi, B.~K. Dey, and V.~M. Prabhakaran, ``Byzantine multiple
  access channels—{Part II}: Communication with adversary identification,''
  {\em IEEE Transactions on Information Theory}, vol.~71, no.~1, pp.~23--60,
  2025.

\bibitem{Knuth-Omega}
D.~E. Knuth, ``Big {O}micron and big {O}mega and big {T}heta,'' {\em SIGACT
  News}, vol.~8, p.~18–24, Apr. 1976.

\bibitem{csiszar2011information}
I.~Csisz{\'a}r and J.~K{\"o}rner, {\em Information theory: coding theorems for
  discrete memoryless systems}.
\newblock Cambridge University Press, 2011.

\bibitem{gacskorner}
P.~G\'acs and J.~Körner, ``Common information is far less than mutual
  information,'' {\em Problems of Control and Information Theory}, vol.~2, 01
  1973.

\end{thebibliography}
